%% file: main.tex
\documentclass[12pt,letterpaper]{article}
\usepackage[margin=1in]{geometry}
\usepackage{setspace}
\PassOptionsToPackage{hyphens}{url}
\usepackage{hyperref}
\hypersetup{colorlinks=true,
linkcolor=red,
urlcolor=blue,
citecolor=blue}
\usepackage{amsmath,amssymb,amsthm,bm}
\newcommand{\1}{\mathbf 1}
\newcommand{\E}{\mathbb E}
\newcommand{\Var}{\mathrm{Var}}
\newcommand{\Cov}{\mathrm{Cov}}
\newcommand{\R}{\mathbb R}
\newcommand{\tr}{\mathrm{tr}}

\newcommand{\vech}{\mathrm{vech}}
\newcommand{\diag}{\mathrm{diag}}
\newcommand{\ones}[1][]{\mathbf{1}_{#1}}
\newcommand{\pto}{\overset{p}\longrightarrow}
\newcommand{\vecop}{\operatorname{vec}}

\newcommand{\plim}{\operatorname*{plim}}

\newtheorem{assumption}{Assumption}
\newtheorem{theorem}{Theorem}
\newtheorem{proposition}{Proposition}
\newtheorem{lemma}{Lemma}
\newtheorem{remark}{Remark}
\usepackage{rotating}
\usepackage{csvsimple}
\usepackage{booktabs}
\usepackage{siunitx}
\usepackage{algorithm}
\usepackage{algpseudocode}

\usepackage[
hyperref=true,
backend=biber,
style=authoryear-comp,
sorting=nyt,
backref=false,
uniquename=false,
url=false,
eprint=false,
]{biblatex}
\title{Fixed-$T$ Dynamic Spatial Panel Model\\with Common Shocks\thanks{%
We thank
Badi Baltagi,
Paul Elhorst,
Xiaoyu Meng, and
Zhenlin Yang
for helpful
comments and suggestions.
We also thank participants at the 2026 RCEA International Conference in Economics, Econometrics, and Finance, Madrid;
the 20th World Conference of the Spatial Econometrics Association and 24th International Workshop on Spatial Econometrics and Statistics, Paris;
and the IAAE 2026 Annual Conference, Lisbon.
All errors are our own.
\par
The accompanying R package \texttt{dspserl} for numerical work in this paper is available at
\url{https://github.com/jessekelighine/dspserl}.
}}
\author{
Jushan Bai\thanks{Department of Economics, Columbia University. Email: \texttt{jushan.bai@columbia.edu}}
\and
Jesse Chieh Chen\thanks{Department of Economics, Columbia University. Email: \texttt{cc5229@columbia.edu}}
}
\date{August 2026}

\begin{document}

\begin{titlepage}

\maketitle

\begin{abstract}
	We study a dynamic spatial panel model with observed regressors, interactive effects,
	and contemporaneous and lagged dependence in a large-$N$, fixed-$T$ framework.
	The spatial model constitutes an $N$-dimensional simultaneous-equations system.
	In this $N$-equation view, the interactive effects introduce $N$ unit-specific loading vectors.
	Estimating them individually when $T$ is fixed
	creates the type of incidental-parameters problem underlying Nickell bias.
	We use the $N$-equation view to account for spatial simultaneity through the spatial Jacobian.
	Crucially, however, we view the same model as a $T$-equation system with $N$ observations.
	Together with a spatially enriched random-loadings (SERL) specification,
	this $T$-equation view yields a quasi-likelihood without unit-specific incidental parameters.
	We propose a computationally tractable block-coordinate algorithm.
	Simulations show small estimation errors and generally near-nominal coverage probabilities.
	Applying the method to U.S. county female labor-force participation,
	we find substantial dynamic and spatial dependence
	and an important role for education in explaining the rise in female labor-force participation.
\end{abstract}


\thispagestyle{empty}

\end{titlepage}

\setcounter{page}{1}

\section{Introduction}

Dynamic panel data often exhibit two distinct forms of dependence.
One is dependence over time, arising through lagged responses and state persistence.
The other is dependence across units, arising through economic, geographic, or network interactions.
In many applications, these two forms of dependence coexist with common shocks that affect all units but with heterogeneous intensities.
A useful empirical framework should therefore accommodate simultaneously dynamic persistence,
spatial interaction, and interactive effects.
This paper studies such a framework in a large-$N$, fixed-$T$ setting.

We consider a dynamic spatial panel model with a contemporaneous spatial lag,
a lagged dependent variable, observed regressors, and an interactive factor structure.
The specification brings together two familiar sources of cross-sectional dependence.
Weak cross-sectional dependence is captured by the spatial autoregressive structure,
while strong cross-sectional dependence is captured by latent common shocks with heterogeneous loadings.
The parameters of primary interest are the spatial coefficient $\rho$, the dynamic coefficient $\phi$,
and the slope coefficients $\beta$.
The econometric challenge is to conduct likelihood-based inference when the time dimension is fixed,
the cross section is large,
and the model contains both simultaneity across units and dynamic feedback over time.
The framework also allows the observed regressors to be correlated with unit heterogeneity and,
more generally, with the latent factor structure through the factor loadings and common shocks.

The model lies at the intersection of several established literatures.
It builds on the foundational spatial-econometrics framework of \textcite{anselin-1988}
and the spatial and dynamic spatial panel models studied by \textcite{lee-2004}, \textcite{yu-2008}, and \textcite{lee-yu-2010},
among others; see also \textcite{baltagi-2021} and \textcite{elhorst-2014}.
Its common-shock component is related to the interactive
effects literature following \textcite{bai-2009}. Within the likelihood
literature, the model is especially close in spirit to \textcite{shi-lee-2017}
and \textcite{bai-li-2021}, who also combine spatial dependence, dynamics, and
latent common components. Those papers adopt large-$N$, large-$T$ asymptotics
and estimate a growing collection of incidental parameters; the resulting QML
estimators require bias correction. The present paper instead focuses on large
$N$ and fixed $T$.

Fixed $T$ makes the treatment of individual heterogeneity central.
In the present model,
the factor loadings play a role analogous to that of unit fixed effects in conventional dynamic panel models.
Eliminating unit effects by a within transformation in the latter setting
produces the fixed-$T$ bias highlighted by \textcite{nickell-1981}.
The analogous challenge here is to avoid estimating the factor loadings unit by unit.

Our central insight is to exploit two complementary representations of the same
model. Viewed cross-sectionally, the spatial model constitutes an
$N$-dimensional simultaneous-equations system. The contemporaneous spatial
transformation is governed by $B(\rho)=I_N-\rho W$, so the conditional
likelihood retains the familiar spatial Jacobian $|B(\rho)|^T$. Treating the
interactive effects directly in this $N$-equation representation, however,
would require estimating $N$ unit-specific loading vectors and would therefore
reintroduce a fixed-$T$ incidental-parameters problem.
Crucially, following the perspective of \textcite{bai-2013,bai-2024},
we also view the same model as a $T$-equation system with $N$ observations,
where each unit contributes a $T$-dimensional time series. Conditional on the initial
observation, the temporal transformation is represented by the lower-triangular
operator $R_T(\phi)$ with unit diagonal and therefore has determinant one; it
contributes no additional parameter-dependent Jacobian term.

The change in representation does not by itself eliminate the loading heterogeneity.
To avoid estimating the loading vectors unit by unit,
we introduce a spatially enriched random-loadings (SERL) specification.
SERL projects the factor loadings on observed controls,
the initial condition,
and their spatial transformations,
thereby separating spatially propagated systematic heterogeneity from a conditionally idiosyncratic residual component.
Combining SERL with the $T$-equation representation replaces the individual loading vectors with a finite-dimensional conditional mean and a low-dimensional covariance matrix.
The resulting quasi-likelihood retains the spatial Jacobian while avoiding unit-specific incidental parameters.

Two closest fixed-$T$ contributions provide useful contrasts.
\textcite{li-yang-2021} study a dynamic spatial panel with correlated random
effects.
Their scalar, time-invariant correlated-random-effects equation is a special case of the random-loading projection underlying SERL.
Their Mundlak specification accommodates dependence between the
regressors and individual heterogeneity, but conditioning on an initial outcome
generated by the dynamic process generally leaves the conditional score
uncentered. They therefore adjust the score using information about when the
process began. They also show that spatial dependence invalidates covariance
estimation based only on raw individual-score outer products and develop a
spatial decomposition of the pooled score. \textcite{li-miao-yang-2026}, by
contrast, extend the fixed-$T$ analysis to interactive fixed effects, treating
the individual loadings as parameters and using adjusted M-estimation together
with a degrees-of-freedom correction for inference. Our approach draws on the
score decomposition of \textcite{li-yang-2021} for robust covariance estimation
but differs from both methods in its treatment of individual heterogeneity and
the initial condition. We project the factor loadings on a finite-dimensional
collection of observed controls, including spatial transformations of the
initial condition and regressor summaries.

Our approach therefore does not estimate the unit-specific fixed effects or,
more generally, the individual factor loadings.  We develop a conditional
Gaussian quasi-maximum likelihood estimator under large-$N$, fixed-$T$
asymptotics. The estimator combines the spatial Jacobian with the
low-rank-plus-diagonal covariance structure induced by the interactive effects.
The resulting criterion admits a computationally tractable block-coordinate
implementation. Conditional on the spatial parameter, the dynamic and slope
coefficients are updated by GLS, while the factor and covariance components are
updated using the low-rank structure. The spatial parameter is then updated
through a one-dimensional conditional maximization. This profiling strategy
remains computationally feasible even when the cross section is large and
extends naturally to the additional specifications considered below. Adding
$Wy_{t-1}$ only enlarges the GLS design, whereas adding a spatial-error filter
replaces the scalar outer search with a two-dimensional optimization over the
roots of a second-order spatial polynomial.
We provide an implementation of the estimation and inference procedures
in the \href{https://github.com/jessekelighine/dspserl}{\texttt{dspserl}} R package.

The simulations examine both estimation and inference. In the baseline designs,
estimation error is small and coverage is close to nominal for the regression
slopes, although first-order Wald intervals for the spatial and dynamic
parameters under-cover in some short, persistent panels. A second experiment
generates the initial outcome from a long-running process and shows that the
first spatial lag in the loading projection produces most of the reduction in
bias and improvement in coverage.
In the empirical application to county-level female labor-force participation data,
the preferred second-order specification leaves the education coefficient positive and
economically important. This result is close to the original OLS evidence of
\textcite{fogli-veldkamp-2011} and provides an alternative interpretation to the
decomposition in \textcite{tziolas-elhorst-2023}.

The remainder of the paper is organized as follows.
Section~\ref{sec:model} presents the model, introduces the two Jacobian arguments, and develops the conditional likelihood.
Section~\ref{sec:asymptotics} establishes the large-$N$, fixed-$T$ inferential theory.
Section~\ref{sec:estimation_algorithm} describes the block coordinate estimation algorithm and the associated computational details.
Section~\ref{sec:extensions} develops extensions with a lagged spatial outcome
and spatially correlated errors.
Section~\ref{sec:simulation} reports Monte Carlo evidence,
Section~\ref{sec:application} applies the method to county-level female labor
force participation, and
Section~\ref{sec:conclusion} concludes.

\section{Dynamic Spatial Panel Model with Common Shocks}\label{sec:model}

This section develops the model and likelihood under fixed $T$.
The key features are spatial dependence, dynamic feedback, and interactive effects.

\subsection{Model}

Let $i=1,...,N$ index units and $t=1,...,T$ index time.
For any positive integer $m$,
let $\ones[m]$ denote the $m$-dimensional vector of ones.
Let $W$ be a known $N\times N$ spatial weights matrix.
Typically, the diagonal of $W$ is assumed to be zero, i.e., $w_{ii}=0$
\parencite{anselin-1988,elhorst-2014}.
Consider the dynamic spatial model with interactive effects
\parencite{shi-lee-2017,bai-li-2021}:
\begin{equation}\label{eq:ysystem}
	y_t=\rho W y_t+\phi y_{t-1}+X_t\beta + \ones[N] \delta_t+\Lambda f_t+\varepsilon_t,\qquad t=1,...,T,
\end{equation}
where
\begin{equation*}
	y_t \in \R^N, \quad
	f_t \in \R^r, \quad
	\Lambda=(\lambda_1,...,\lambda_N)'\in\R^{N\times r}, \quad
	\varepsilon_t \in \R^N.
\end{equation*}
The time-varying components of $f_t$ represent latent aggregate shocks or
common conditions that affect many or all cross-sectional units. Examples
include macroeconomic fluctuations, policy and regulatory changes,
technological innovations, financial conditions, commodity-price movements,
public-health shocks, and other economy-wide demand or supply disturbances.
The heterogeneous loadings $\lambda_i$ allow the magnitude and direction of
the response to these common shocks to differ across units.
Time-invariant individual heterogeneity can be incorporated by including a constant factor.

Define $B(\rho):=I_N-\rho W$.
Equation~\eqref{eq:ysystem} can be equivalently written as
\begin{equation}\label{eq:Bsystem}
	B(\rho)y_t=\phi y_{t-1}+X_t\beta+\ones[N] \delta_t+\Lambda f_t+\varepsilon_t.
\end{equation}
We observe the initial cross section $y_0=(y_{10},...,y_{N0})'$ and condition on it throughout.
Let $x_{it}\in\R^K$ and stack $X_t$ as the $N\times K$ matrix with $i$th row $x_{it}'$.
Define the observed design or common conditioning $\sigma$-field as
\begin{equation*}
	\mathcal C_N
	:=
	\sigma(y_0,X_1,...,X_T,W).
\end{equation*}

\subsection{Transformation and Jacobian}

Recall that $y_t\in\R^N$ denotes the cross section observed at time $t$.
In this subsection we also introduce unit-specific time paths such as $y_i\in\R^T$,
so the time index $t$ and the unit index $i$ refer to objects of different dimensions.

Define the transformed innovations
\begin{equation}\label{eq:utdef}
	u_t(\rho,\phi,\beta):=B(\rho)y_t-\phi y_{t-1}-X_t\beta\in\R^N.
\end{equation}
Then \eqref{eq:Bsystem} implies
\begin{equation}\label{eq:uteq}
	u_t=\ones[N] \delta_t + \Lambda f_t+\varepsilon_t,\qquad t=1,...,T.
\end{equation}

To make the dynamic structure explicit at the unit level, define for each $i$
\begin{equation*}
	y_i := (y_{i1},...,y_{iT})' \in \R^T, \quad
	X_i\beta := (x_{i1}'\beta,...,x_{iT}'\beta)' \in \R^T, \quad
	\mathbf e_1 := (1,0,...,0)' \in \R^T,
\end{equation*}
and let $Y=(y_1,...,y_T)$ be the $N\times T$ matrix of observed cross sections.
For any matrix $C$, write $C_{i\cdot}$ for its $i$th row.
Introduce the $T\times T$ lower-triangular time-direction operator
\begin{equation}
	R_T(\phi) :=
	\begin{pmatrix}
		1      & 0      & 0      & \cdots & 0      \\
		-\phi  & 1      & 0      & \cdots & 0      \\
		0      & -\phi  & 1      & \cdots & 0      \\
		\vdots & \ddots & \ddots & \ddots & \vdots \\
		0      & \cdots & 0      & -\phi  & 1
	\end{pmatrix}, \qquad
	\det R_T(\phi) = 1.
\end{equation}
Conditional on the initial condition $y_{i0}$, the dynamic recursion can be written as
\begin{equation}\label{eq:utdef_unit}
	u_i(\rho,\phi,\beta) = R_T(\phi) y_i - \rho (WY)_{i\cdot}' - X_i\beta - \phi y_{i0}\mathbf e_1.
\end{equation}
Equation~\eqref{eq:utdef_unit} isolates the time-direction transformation.
Since $R_T(\phi)$ is lower triangular with ones on the diagonal,
the dynamic filter contributes no parameter-dependent Jacobian term once we condition on $y_{i0}$.
The only nontrivial Jacobian term in the likelihood comes from the contemporaneous spatial filter $B(\rho)=I_N-\rho W$.

\paragraph{Jacobian.}

Consider the change of variables
\begin{equation*}
(y_1,...,y_T)\longmapsto(u_1,...,u_T)
\end{equation*}
given by \eqref{eq:utdef}, conditioning on $y_0$. Stack
\begin{equation*}
y:=(y_1',...,y_T')'\in\R^{NT},
\qquad
u:=(u_1',...,u_T')'\in\R^{NT}.
\end{equation*}
The Jacobian matrix $\mathcal J:=\partial u/\partial y'$ is block lower triangular:
\begin{equation}\label{eq:jacobian_matrix}
	\mathcal J \;=\;
	\begin{pmatrix}
		B(\rho)   & 0          & 0       & \cdots    & 0      \\
		-\phi I_N & B(\rho)    & 0       & \cdots    & 0      \\
		0         & -\phi  I_N & B(\rho) & \cdots    & 0      \\
		\vdots    & \ddots     & \ddots  & \ddots    & \vdots \\
		0         & \cdots     & 0       & -\phi I_N & B(\rho)
	\end{pmatrix}.
\end{equation}
Hence
\begin{equation}\label{eq:jacobian}
	|\mathcal J|=\prod_{t=1}^T |B(\rho)| = |B(\rho)|^T.
\end{equation}
This calculation makes clear that there are two distinct structural transformations.
The first is the time-direction operator $R_T(\phi)$ in \eqref{eq:utdef_unit},
whose determinant is one.
The second is the contemporaneous cross-sectional operator $B(\rho)$,
which generates the nontrivial Jacobian contribution.
Accordingly, the conditional likelihood contains the term $T\log|B(\rho)|$ and no additional Jacobian term involving $\phi$.

\subsection{Factors}

Let the factor path matrix be
\begin{equation*}
F:=\begin{pmatrix} f_1'\\ \vdots\\ f_T'\end{pmatrix}\in\R^{T\times r}.
\end{equation*}
For each unit $i$, stack innovations and idiosyncratic errors over time:
\begin{equation*}
u_i:=(u_{i1},...,u_{iT})'\in\R^T,\qquad
\varepsilon_i:=(\varepsilon_{i1},...,\varepsilon_{iT})'\in\R^T.
\end{equation*}
Passing from the time-stacked vector $(u_1',...,u_T')'$ to the unit-stacked collection $\{u_i\}_{i=1}^{N}$
is only a reordering of coordinates.
Equivalently, it is a fixed permutation of the $NT$ entries and therefore has absolute determinant of one.
Thus, no further Jacobian term arises when the likelihood is re-written as a product over units rather than over time.
Then \eqref{eq:uteq} is equivalent to
\begin{equation}\label{eq:ui_basic}
	u_i = \delta + F\lambda_i + \varepsilon_i,\qquad i=1,...,N,
\end{equation}
where $\lambda_i\in\R^r$ is the $i$th row of $\Lambda$ as a column vector and $\delta=(\delta_1,...,\delta_T)'\in\R^T$.

\begin{remark}[Notation]
	Objects indexed by $t$ are cross-sectional $N\times1$ vectors,
	for example $y_t,u_t,\varepsilon_t\in\R^N$.
	Objects indexed by $i$ are unit-specific $T\times1$ vectors, for example $y_i,u_i,\varepsilon_i\in\R^T$.
	Thus, the same letter may denote different objects depending on whether it is indexed by $t$ or by $i$.
	We keep this notation because the intended dimension is usually clear from the index and from the context.
\end{remark}

\subsection{Time-Heteroskedastic Idiosyncratic Errors}

We assume the idiosyncratic errors  $\varepsilon_i$ are iid, allowing heteroskedasticity over time but no serial correlation.
The errors $\varepsilon_i$ are also independent of factor loadings $\Lambda$.
\begin{assumption}[Idiosyncratic Errors]\label{ass:varepsilon}
 Conditional on $\mathcal C_N,$
$\{\varepsilon_i\}_{i=1}^N$
 are identically distributed and independent across $i$, and, for some $\zeta>0$ and $C_\varepsilon <\infty$,
\[ \begin{aligned}
&\E(\varepsilon_i\mid\mathcal C_N)=0,
\qquad
\E(\varepsilon_i\varepsilon_i'\mid\mathcal C_N)
=
D_\varepsilon,
\qquad
D_\varepsilon=\diag(\sigma_1^2,\ldots,\sigma_T^2),\\
&(\varepsilon_1,\ldots,\varepsilon_N)
\perp\!\!\!\perp
\Lambda
\mid\mathcal C_N, \qquad
\sup_{N\ge1}
\E\!\left[
\|\varepsilon_i\|^{4+\zeta}
\,\middle|\,
\mathcal C_N
\right]
\le C_\varepsilon
\qquad\text{a.s.}
\end{aligned}
\]
\end{assumption}

\subsection{Spatially Enriched Random Loadings}

We model the factor loadings using a \emph{spatially enriched random-loadings} (SERL) specification.
The construction adapts the correlated-random-effects approach of
\textcite{mundlak-1978},
\textcite{chamberlain-1982},
and \textcite{wooldridge-2005}
to  factor loadings in a dynamic spatial model.
In a static spatial Durbin panel,
\textcite{debarsy-2012} uses a Mundlak projection for a time-invariant individual effect with own and spatially lagged regressor averages.
SERL extends this construction to vector-valued factor loadings in a dynamic model and allows the controls to include the initial outcome and higher-order spatial transformations.
It allows the conditional loading means to depend on observed covariates,
the initial condition,
and information propagated through the spatial structure.

Let $Q_X\in\R^{N\times p}$ collect observed summaries of the regressors.
A Mundlak-type specification may use time averages of the regressors,
whereas a Chamberlain-type specification may use the complete observed regressor history.
For a fixed integer $L\geq0$, define
\begin{equation}\label{eq:correct_projection}
	z_i
	=
	\Big[
	(Q_X)_{i\cdot}', (WQ_X)_{i\cdot}', \cdots, (W^LQ_X)_{i\cdot}',
	(y_0)_i, (Wy_0)_i, \cdots, (W^Ly_0)_i
	\Big]'
	\in\R^{(L+1)(p+1)}.
\end{equation}
Because $Q_X$ is constructed from the observed regressor history and
$z_i$ is constructed from $Q_X$, $y_0$, and $W$,
both $Q_X$ and $z_i$ are $\mathcal C_N$-measurable.
We specify
\begin{equation}\label{eq:proj}
	\lambda_i = A z_i + \eta_i,
\end{equation}
where $Az_i$ captures the component of the loading heterogeneity
systematically related to the observed information and $\eta_i$ represents
the idiosyncratic loading heterogeneity.

We refer to this as the SERL specification. It has a natural interpretation in terms of heterogeneous responses to common shocks. The loading $\lambda_i$ determines how unit $i$ responds to the common shocks $f_t$, and this response heterogeneity may be systematically related to observed characteristics, initial conditions, and their spatial transformations. The decomposition
$
\lambda_i=Az_i+\eta_i
$
separates such systematic heterogeneity from  unit-specific heterogeneity. The first component, $Az_i$,  is predictable from the common conditioning information $\mathcal{C}_N$ and captures systematic or coordinated responses associated with observed characteristics, initial conditions, and their spatial transformations.
The second component, $\eta_i$, captures purely unit-specific heterogeneity in the response to the common shock that is not systematically related to the information in $\mathcal{C}_N$, and is assumed to be  independent across $i$ conditional on $\mathcal{C}_N$.
In this sense,
SERL adapts the correlated-random-effects idea to factor loadings in a spatial environment. Assumption~\ref{ass:SERL}
formalizes this decomposition.

\begin{assumption}[Spatially enriched random loadings]\label{ass:SERL}
	For some fixed integer $L\geq0$ and coefficient matrix
	$A\in\R^{r\times(L+1)(p+1)}$,
	the loading vectors satisfy \eqref{eq:proj}.
	Conditional on $\mathcal C_N$,
	the vectors $\{\eta_i\}_{i=1}^N$ are identically distributed and independent across $i$, with
	\[
		\E(\eta_i\mid\mathcal C_N)=0,
		\qquad
		\E(\eta_i\eta_i'\mid\mathcal C_N)=\Sigma_\eta, \qquad \sup_{N\ge1}
\E\!\left[
\|\eta_i\|^{4+\zeta}
\,\middle|\,
\mathcal C_N
\right]
\le C_\eta< \infty
\qquad\text{a.s.}
	\]
	where $\Sigma_\eta$ is positive definite.
\end{assumption}
We refer to $L$ as the \emph{spatial enrichment order}.
The case $L=0$ uses only own-unit regressor summaries and the own initial outcome.
For $L>0$,
neighboring information enters through $WQ_X,\ldots,W^LQ_X$ and $Wy_0,\ldots,W^Ly_0$.
For the theoretical analysis, $L$ is treated as fixed.  When the conditional loading mean arises from a richer spatial process,
a finite value of $L$ may instead be viewed as an approximation, whose adequacy can be assessed by sensitivity
 to additional spatial transformations. Accordingly, in the empirical application, we examine alternative
  enrichment orders to assess the sensitivity of the structural estimates to the choice of $L$.

If the decomposition  is written as $\lambda_i=a+Az_i+\eta_i$ with $a\in\R^r$ as an intercept,
the common term $Fa$ is absorbed by the unrestricted time effect $\delta$.
Therefore, we omit $a$ without loss of generality.

Appendix~\ref{sec:joint-initial-condition} develops an alternative approach
that models $y_0$ jointly with the sample outcomes and projects the loadings only on regressor-based controls.

\subsection{Likelihood Conditional on \texorpdfstring{$\mathcal C_N$}{C_N}}

Substituting \eqref{eq:proj} into \eqref{eq:ui_basic} yields
\begin{equation*}
	u_i = \delta + F(Az_i+\eta_i)+\varepsilon_i = \delta + FAz_i + F\eta_i + \varepsilon_i.
\end{equation*}
Define
\begin{equation} \label{eq:ei}
 e_i := F \eta_i +\varepsilon_i
\end{equation}
By Assumptions~\ref{ass:varepsilon} and \ref{ass:SERL},
	conditional on $\mathcal C_N$,
	the vectors $\{e_i\}_{i=1}^N$ are identically distributed and independent across $i$, with
\begin{equation} \label{eq:conditional_moments}
\E(e_i\mid\mathcal C_N)=0, \qquad
\E(e_ie_i'\mid\mathcal C_N)
=
F\Sigma_\eta F'+D_\varepsilon,
\qquad
\sup_{N\ge1}
\E\!\left[
\|e_i\|^{4+\zeta}
\,\middle|\,
\mathcal C_N
\right]
\le C_e, \quad \text{a.s.}
\end{equation}
where $C_e<\infty$.
Hence,
\begin{equation}\label{eq:cond_mean}
	\E(u_i\mid\mathcal C_N) = \delta + FAz_i.
\end{equation}
and
\begin{equation}\label{eq:cond_var}
	\Var(u_i\mid\mathcal C_N) = F\Sigma_\eta F' + D_\varepsilon \;=:\; \Sigma_u.
\end{equation}

For the moment, let $\alpha$ denote the collection of unknown parameters entering the working likelihood. A precise parameterization of $\alpha$ in terms of free parameters is introduced in Section~\ref{sec:asymptotics} after imposing a normalization on the factor path $F$ to remove its rotational indeterminacy.
No Gaussian distribution is imposed on $\eta_i$, $\varepsilon_i$, or $e_i$. The Gaussian specification used below is a working likelihood based on the conditional mean and covariance in \eqref{eq:cond_mean}--\eqref{eq:cond_var}.
Combining the determinant-one time transformation in \eqref{eq:utdef_unit},
the spatial Jacobian in \eqref{eq:jacobian}, and
the determinant-one permutation from time stacking to unit stacking,
one obtains a conditional likelihood with a single Jacobian term, namely
$|B(\rho)|^T$. Let $Y=(y_1,\ldots,y_T)$. Motivated by the conditional moment restrictions,
we use the conditional Gaussian working likelihood
\begin{equation}\label{eq:lik}
	\begin{aligned}
		L_N(\alpha;Y\mid\mathcal C_N)
		&=
		|B(\rho)|^T
		\prod_{i=1}^N
		(2\pi)^{-T/2}
		|\Sigma_u|^{-1/2} \times  \\
		&  \exp\!\left(
		-\tfrac12 (u_i-\delta-FAz_i)'\Sigma_u^{-1}(u_i-\delta-FAz_i)
		\right),
	\end{aligned}
\end{equation}
where $u_i=(u_{i1},\ldots,u_{iT})'$ is obtained by rearranging $u_t$ computed from observed data in \eqref{eq:utdef}.
The corresponding working log-likelihood, up to an additive constant, is given by
\begin{equation}\label{eq:loglik}
	\ell(\alpha)
	=
	T\log|B(\rho)|
	-\frac{N}{2}\log|\Sigma_u|
	-\frac12 \sum_{i=1}^N (u_i-\delta-FAz_i)'\Sigma_u^{-1}(u_i-\delta-FAz_i).
\end{equation}

\section{Inferential Theory under Large \texorpdfstring{$N$}{N} and Fixed \texorpdfstring{$T$}{T}}\label{sec:asymptotics}

This section states the large-$N$, fixed-$T$ theory for the Gaussian QML
estimator. Gaussianity is not required. The key complication is that the
spatial transformation makes the observed-data score cross-sectionally
dependent, even when the structural residuals are independent across units.
Consequently, inference must use spatially corrected score contributions rather
than raw individual score outer products.
All expectations, variances, and probability statements in this section are
understood conditionally on $\mathcal C_N$ unless otherwise indicated.

\subsection{Parameterization, Normalization, and Sample Criterion}\label{subsec:qml_setup}

The key parameters are
\begin{equation*}
\theta := (\rho,\ \phi,\ \beta')'
\end{equation*}
and the factor path is identified only up to rotation. We assume that its top
$r\times r$ block has full rank and use the normalized representation
\begin{equation}\label{eq:F_normalization_theory}
	F=
	\begin{pmatrix}
		I_r\\ F_2
	\end{pmatrix},
	\qquad F_2\in\R^{(T-r)\times r}.
\end{equation}
Write $D_\varepsilon=\diag(\sigma_1^2,...,\sigma_T^2)$ and
$\sigma^2=(\sigma_1^2,...,\sigma_T^2)'$. The nuisance vector contains only
free coordinates,
\begin{equation*}
	\varphi
	=
	\left(
		\delta',\vecop(A)',\vecop(F_2)',
		\vech(\Sigma_\eta)',\sigma^{2\prime}
	\right)',
	\qquad
	\alpha=(\theta',\varphi')'.
\end{equation*}
Thus, only $F_2$ is included in $\alpha$, and $\Sigma_\eta$ is represented by its unique elements.  The parameter space restricts $\Sigma_\eta$ to be positive definite
and every element of $\sigma^2$ to be positive. For each unit $i$, define the transformed innovations
\begin{equation*}
u_i(\theta):=(u_{i1}(\theta),...,u_{iT}(\theta))'\in\mathbb R^T
\quad\text{where}\quad
u_t(\theta)=B(\rho)y_t-\phi y_{t-1}-X_t\beta
\end{equation*}
for $t=1,...,T$.
The nuisance mean is
\begin{equation*}
\mu_i(\varphi):=\delta+FAz_i\in\mathbb R^T,
\end{equation*}
and the inverse of the working covariance matrix is
\begin{equation*}
M(\varphi):=\Sigma_u(\varphi)^{-1}
\quad\text{where}\quad
\Sigma_u(\varphi):=D_\varepsilon+F\Sigma_\eta F'.
\end{equation*}
Define residuals $e_i(\alpha):=u_i(\theta)-\mu_i(\varphi)$.
Note
$e_i(\alpha_0)=F\eta_i+\varepsilon_i$, which is $e_i$ defined in \eqref{eq:ei}, and satisfies
\eqref{eq:conditional_moments}.
The per-unit Gaussian pseudo log-likelihood is
\begin{equation*}
\ell_i(\alpha)
:=
-\frac12\log|\Sigma_u(\varphi)|
-\frac12\,e_i(\alpha)'M(\varphi)e_i(\alpha),
\end{equation*}
and the sample criterion, normalized by $N$, is
\begin{equation}\label{eq:ell_N}
\ell_N(\alpha)
:=
\frac{T}{N}\log|B(\rho)|+\frac1N\sum_{i=1}^N \ell_i(\alpha).
\end{equation}
Let $\widehat\alpha=(\widehat\theta',\widehat\varphi')'$ be an interior maximizer of $\ell_N(\alpha)$,
and hence a solution to the corresponding first-order conditions.

\subsection{Consistency and Asymptotic Normality}\label{subsec:asymptotic_normality}

Define the ordinary unit score by
\begin{equation*}
	s_i(\alpha)
	:=
	\nabla_\alpha\ell_i(\alpha)
	+
	\frac{T}{N}\nabla_\alpha\log|B(\rho)|.
\end{equation*}
The normalized pooled score therefore satisfies
\begin{equation*}
	S_N(\alpha)
	:= \frac1N \sum_{i=1}^N s_i(\alpha)
	= \nabla_\alpha\ell_N(\alpha).
\end{equation*}
The spatial filter makes the ordinary unit scores cross-sectionally dependent even when the structural residuals are independent.
In the baseline model,
the dependence requiring correction is due to the $\rho$ and $\phi$ components,
since their reduced-form regressors depend on structural residuals from other units.
Consequently,
\begin{equation}\label{eq:Omega_raw_scores}
	\Var_N\left(\frac1{\sqrt N}\sum_{i=1}^N s_i(\alpha_0)\right)
	=
	\frac1N\sum_{i=1}^N\sum_{j=1}^N
	\Cov_N\!\big(s_i(\alpha_0),s_j(\alpha_0)\big),
\end{equation}
where, for notational simplicity, we write
\begin{equation*}
    \E_N(\,\cdot\,):=\E(\,\cdot\mid\mathcal C_N),
    \qquad
    \Var_N(\,\cdot\,):=\Var(\,\cdot\mid\mathcal C_N),
    \qquad
    \Cov_N(\,\cdot\,,\,\cdot\,):=\Cov(\,\cdot\,,\,\cdot\mid\mathcal C_N).
\end{equation*}
The diagonal-only outer product based on the raw scores generally omits the cross-unit terms in \eqref{eq:Omega_raw_scores}.
The cross-sectional dependence only enters the $\rho$ and $\phi$ components of the score,
so the remaining components already have the required unitwise form.

Following \textcite{li-yang-2021},
we therefore rearrange the cross-unit terms in the $\rho$ and $\phi$ scores while preserving the pooled score.
Let $\{g_i(\alpha)\}_{i=1}^{N}$ denote the resulting rearranged contributions,
satisfying
\begin{equation}\label{eq:g_score_decomposition}
	N S_N(\alpha)
	=
	\sum_{i=1}^N s_i(\alpha)
	=
	\sum_{i=1}^N g_i(\alpha),
\end{equation}
and that at $\alpha_0$,
the rearranged contributions form a martingale-difference array with respect to the filtration
\begin{equation*}
	\mathcal F_i
	:=
	\mathcal C_N\vee\sigma(e_1,...,e_i).
\end{equation*}
Appendix~\ref{sec:spatial_score_construction} gives the exact component-by-component construction of $g_i$.

\begin{assumption}[Regularity]\label{ass:regularity}
	\begin{enumerate}
		\item
			The parameter space is compact and $\alpha_0$ is an interior point.
			The criterion converges uniformly in probability over the parameter space to a limiting criterion uniquely maximized at $\alpha_0$.
		\item
			The criterion is twice continuously differentiable on a convex neighborhood $\mathcal A_0$ of $\alpha_0$.
		\item
			The row and column sums of the spatial weights matrices are uniformly bounded,
			and $B(\rho)^{-1}$ and the dynamic-spatial reduced-form operator are uniformly bounded over the parameter space.
		\item
			On $\mathcal A_0$,
			the eigenvalues of $\Sigma_u(\varphi)$ are uniformly bounded away from zero and infinity,
			and its first two derivatives are uniformly bounded.
			The unit-to-unit reduced-form response blocks entering $g_i(\alpha)$,
			together with their first two parameter derivatives,
			admit a common summable envelope whose block row and column sums are uniformly bounded.
			The conditionally nonrandom quantities entering the score contributions,
			including the observed regressors, projection controls, initial-outcome terms, and reduced-form means,
			together with their first two derivatives,
			have uniformly bounded cross-sectional average moments of order $4+\zeta$.
		\item
			There is a matrix function $H(\alpha)$,
			continuous at $\alpha_0$,
			such that
			\begin{equation*}
				\sup_{\alpha\in\mathcal A_0}
				\left\|
				-\nabla_{\alpha\alpha'}^2\ell_N(\alpha)-H(\alpha)
				\right\|
				\pto0,
			\end{equation*}
			where $H(\alpha_0)$ is finite and positive definite.
	\end{enumerate}
\end{assumption}

 These are standard identification and smoothness requirements for
M-estimation, augmented by stability conditions for the spatial reduced form.
With fixed $T$,
the dimension of $\alpha$ is fixed,
so consistency and the asymptotic-linearization argument follow standard finite-dimensional QMLE reasoning,
as in \textcite{white-1982}.
What is nonstandard here is the cross-sectional dependence of the raw spatial scores;
the rearrangement into $g_i$ and the martingale central limit theorem below provide the required score limit.

\begin{assumption}[Predictable score covariance]\label{ass:score_clt}
	For the rearranged contributions $g_i(\alpha_0)$ and the filtration
	$\mathcal F_i=\mathcal C_N\vee\sigma(e_1,...,e_i)$,
	the predictable quadratic variation satisfies
	\[
		\frac1N\sum_{i=1}^N
		\E\!\left[
			g_i(\alpha_0)g_i(\alpha_0)'
			\mid\mathcal F_{i-1}
		\right]
		\pto\Omega(\alpha_0),
	\]
	where $\Omega(\alpha_0)$ is finite and positive definite.
\end{assumption}

\begin{remark}[Lindeberg Condition]
	The moment bounds established in Lemma~\ref{lem:score-moment-bounds} of Appendix~\ref{sec:spatial_score_construction},
	based on Assumptions~\ref{ass:varepsilon}, \ref{ass:SERL}, and \ref{ass:regularity},
	show that,
	for some $\delta>0$,
	\begin{equation*}
		\frac1N\sum_{i=1}^N
		\E_N\|g_i(\alpha_0)\|^{2+\delta}
		=O_p(1).
	\end{equation*}
	Since $\mathcal C_N\subseteq\mathcal F_{i-1}$,
	Markov's inequality implies that
	\begin{equation*}
		\frac1N\sum_{i=1}^N
		\E\!\left[
		\|g_i(\alpha_0)\|^{2+\delta}
		\,\middle|\,\mathcal F_{i-1}
		\right]
		=O_p(1).
	\end{equation*}
	Consequently,
	for every $\epsilon>0$,
	the conditional Lindeberg condition follows:
	\begin{equation*}
		\begin{aligned}
			\frac1N\sum_{i=1}^N
			\E\!\left[
			\|g_i(\alpha_0)\|^2
			\1\{\|g_i(\alpha_0)\|>\epsilon\sqrt N\}
			\middle|\mathcal F_{i-1}
			\right]
			\leq
			\frac{1}{\epsilon^\delta N^{\delta/2}}
			\frac1N\sum_{i=1}^N
			\E\!\left[
			\|g_i(\alpha_0)\|^{2+\delta}
			\,\middle|\,\mathcal F_{i-1}
			\right]
			=o_p(1).
		\end{aligned}
	\end{equation*}
\end{remark}

Because the rearranged contributions form a martingale-difference array,
the Lindeberg condition,
Assumption~\ref{ass:score_clt},
and Theorem~3.2 and Corollary~3.1 of \textcite[pp.~58--59]{hall-heyde-1980} yield
\begin{equation*}
	\sqrt N S_N(\alpha_0)
	=
	\frac1{\sqrt N}\sum_{i=1}^N g_i(\alpha_0)
	\Rightarrow
	\mathcal N\!\left(0,\Omega(\alpha_0)\right).
\end{equation*}
The uniform-convergence and unique-maximizer conditions in Assumption~\ref{ass:regularity},
together with the standard argmax theorem,
give $\widehat\alpha\pto\alpha_0$.
A mean-value expansion of the first-order condition around $\alpha_0$ then yields
\begin{equation*}
	\sqrt N(\widehat\alpha-\alpha_0)
	=
	H(\alpha_0)^{-1}\sqrt N S_N(\alpha_0)+o_p(1).
\end{equation*}
This gives the following result.

\begin{theorem}[Consistency and asymptotic normality]\label{thm:qml_asymptotics}
	Under Assumptions~\ref{ass:varepsilon}--\ref{ass:score_clt},
	\begin{equation}\label{eq:alpha_CLT}
		\widehat\alpha\pto\alpha_0,
		\qquad
		\sqrt N(\widehat\alpha-\alpha_0)
		\Rightarrow
		\mathcal N\Big(0,H(\alpha_0)^{-1}\Omega(\alpha_0)(H(\alpha_0)^{-1})'\Big).
	\end{equation}
\end{theorem}

Because the criterion is a Gaussian
pseudo-likelihood, the information equality need not hold. Estimating
$\Omega(\alpha_0)$ therefore requires separate analysis.

\subsection{Estimation of the Asymptotic Covariance}\label{subsec:sandwich_general}

The martingale representation in \eqref{eq:g_score_decomposition} identifies the score covariance with the probability limit of its predictable quadratic variation:
\begin{equation*}
	\Omega(\alpha_0)
	=
	\plim_{N\to\infty}
	\frac1N\sum_{i=1}^N
	\E\!\left[
		g_i(\alpha_0)g_i(\alpha_0)'
		\mid\mathcal F_{i-1}
	\right].
\end{equation*}
While the full outer product of the rearranged contributions is consistent for this covariance,
we further refine the estimator to remove terms that have zero conditional expectation under the maintained martingale structure.
The full construction of the estimator $\widehat\Omega_M(\widehat\alpha)$ is given in Appendix~\ref{sec:spatial_score_construction}.
Proposition~\ref{prop:feasible_score_covariance} proves the asymptotic equivalence
\begin{equation*}
	\widehat\Omega_M(\alpha_0)
	-
	\frac1N\sum_{i=1}^N g_i(\alpha_0)g_i(\alpha_0)'
	=o_p(1).
\end{equation*}
It also proves the plug-in replacement
$\widehat\Omega_M(\widehat\alpha)-\widehat\Omega_M(\alpha_0)=o_p(1)$,
and hence $\widehat\Omega_M(\widehat\alpha)\pto\Omega(\alpha_0)$.
For improved finite-sample performance,
our preferred implementation applies a pair-only HC1-style correction.
Using the decomposition in \eqref{eq:martingale_covariance_estimator},
write
$\widehat\Omega_M=\widehat\Omega_\dagger+\widehat\Omega_v$,
where $\widehat\Omega_\dagger$ is the one-unit component and
$\widehat\Omega_v$ is the pair component.
We define
\begin{equation}\label{eq:Omega_martingale_pair_hc1}
	\widehat\Omega_{M,\mathrm{pair}}(\widehat\alpha)
	:=
	\widehat\Omega_\dagger(\widehat\alpha)
	+
	\frac{N}{N-p_\alpha}
	\widehat\Omega_v(\widehat\alpha),
\end{equation}
where $p_\alpha:=\dim(\alpha)$ counts all estimated structural and nuisance parameters.
This adjustment targets the component associated with the spatial and dynamic score rearrangement while leaving the one-unit component unchanged.
We use it as a simple finite-sample correction and do not claim that it makes the pair component exactly unbiased.
The effective sample size is $N$ because each cross-sectional unit contributes one fixed-length time path to the score decomposition.
Since $T$ and $p_\alpha$ are fixed as $N\to\infty$,
the multiplier converges to one and does not change the covariance estimator's probability limit.
We use $\widehat\Omega_{M,\mathrm{pair}}$ for our reported simulation and application inference below,
except where a table explicitly labels otherwise.
Let $\widehat H:=-\nabla_{\alpha\alpha'}^2\ell_N(\widehat\alpha)$.
The preferred feasible sandwich covariance estimator for the full estimator is
\begin{equation}\label{eq:Valpha_hat}
	\widehat{\Var}(\widehat\alpha)
	=
	\frac1N\widehat H^{-1}\widehat\Omega_{M,\mathrm{pair}}(\widehat\alpha)(\widehat H^{-1})'.
\end{equation}
The uncorrected benchmark replaces $\widehat\Omega_{M,\mathrm{pair}}$ by $\widehat\Omega_M$.

\subsection{Profile Inference for the Structural Parameters}\label{subsec:profile_sandwich}

Partition the negative Hessian conformably with
$\alpha=(\theta',\varphi')'$ and define
\begin{equation}\label{eq:H_schur}
	H_{\theta\mathbin{\cdot}\varphi}
	=
	H_{\theta\theta}
	-H_{\theta\varphi}H_{\varphi\varphi}^{-1}H_{\varphi\theta},
	\qquad
	P=H_{\theta\varphi}H_{\varphi\varphi}^{-1}.
\end{equation}
The profiled corrected contribution is
\begin{equation}\label{eq:g_profile}
	g_{\theta\mathbin{\cdot}\varphi,i}(\alpha_0)
	=
	g_{\theta,i}(\alpha_0)-P g_{\varphi,i}(\alpha_0),
\end{equation}
with covariance
\begin{equation*}
	\Omega_{\theta\mathbin{\cdot}\varphi}
	=
	\plim\frac1N\sum_{i=1}^N
	\E[g_{\theta\mathbin{\cdot}\varphi,i}g_{\theta\mathbin{\cdot}\varphi,i}'].
\end{equation*}
It follows that
\begin{equation}\label{eq:theta_CLT}
	\sqrt N(\widehat\theta-\theta_0)
	\Rightarrow
	\mathcal N\!\left(
		0,
		H_{\theta\mathbin{\cdot}\varphi}^{-1}
		\Omega_{\theta\mathbin{\cdot}\varphi}
		(H_{\theta\mathbin{\cdot}\varphi}^{-1})'
	\right).
\end{equation}
For feasible inference,
let $\widehat P$ be the sample analogue of $P$ in \eqref{eq:H_schur}.
The preferred profiled covariance is the corresponding transformation of the pair-corrected martingale covariance estimator:
\begin{equation*}
	\widehat\Omega_{\theta\mathbin{\cdot}\varphi}
	=
	\begin{bmatrix}
		I & -\widehat P
	\end{bmatrix}
	\widehat\Omega_{M,\mathrm{pair}}(\widehat\alpha)
	\begin{bmatrix}
		I \\ -\widehat P'
	\end{bmatrix}.
\end{equation*}
Then
\begin{equation}\label{eq:Vtheta_hat}
	\widehat{\Var}(\widehat\theta)
	=
	\frac1N
	\widehat H_{\theta\mathbin{\cdot}\varphi}^{-1}
	\widehat\Omega_{\theta\mathbin{\cdot}\varphi}
	(\widehat H_{\theta\mathbin{\cdot}\varphi}^{-1})'.
\end{equation}

\section{Estimation Algorithm}\label{sec:estimation_algorithm}

This section outlines a practical block algorithm to maximize the conditional Gaussian (QML)
log-likelihood \eqref{eq:loglik} when $T$ is fixed and the number of factors $r$ is small.
The parameters of interest are $(\rho,\phi,\beta)$,
while $(\delta,F,A,\Sigma_\eta,D_\varepsilon)$ are nuisance parameters.
We exploit three features:
(i) for given $(F,\Sigma_u)$, $(\delta,A)$ can be \emph{concentrated out in closed form};
(ii) under the current values of $(\delta,F,A)$, the Gaussian working model for the random loading heterogeneity $\eta_i$
yields simple EM/ECM updates for $(F,\Sigma_\eta,D_\varepsilon)$;
(iii) given $\rho$, the parameters $(\phi,\beta)$ admit a closed-form generalized least squares update and a one-dimensional conditional
maximization step in $\rho$ can be implemented using the Jacobian term $\log|B(\rho)|$.
This yields an inner-outer loop structure:
the outer loop performs a scalar conditional maximization in $\rho$, and the inner loop is a block coordinate/ECME iteration
for $(\phi,\beta)$ and $(\delta,A,F,\Sigma_\eta,D_\varepsilon)$ given $\rho$.

\subsection{Block Coordinate/ECME Iteration: Inner Loop}

Let $u_t(\rho,\phi,\beta)$ be defined in \eqref{eq:utdef} and stack $U=(u_1,...,u_N)\in\mathbb R^{T\times N}$,
$Z=(z_1,...,z_N)\in\mathbb R^{q\times N}$. Given current parameter values, iterate the following steps
until convergence, e.g., change in concentrated log-likelihood below a tolerance.

\paragraph{Inputs to Steps 1--4.}
Given current $(\rho,\phi,\beta)$, form the innovation (or transformed residual) matrix
\begin{equation*}
U(\rho,\phi,\beta):=(u_1,...,u_N)\in\mathbb R^{T\times N},\qquad
u_t(\rho,\phi,\beta):=B(\rho)y_t-\phi y_{t-1}-X_t\beta,\ \ t=1,...,T,
\end{equation*}
and let $u_i(\rho,\phi,\beta)\in\mathbb R^T$ denote the $i$th column of $U(\rho,\phi,\beta)$.
In Steps~1--4 below we treat $U(\rho,\phi,\beta)$ as the \emph{data input} and update only nuisance parameters
$(\delta,A,F,\Sigma_\eta,D_\varepsilon)$.

\paragraph{Step 1 (Input: the transformed residual matrix; concentrate out time effects $\delta$ and projection matrix $A$).}

Given current $(F,\Sigma_\eta,D_\varepsilon)$, compute
\begin{equation*}
\Sigma_u=D_\varepsilon+F\Sigma_\eta F',\qquad M:=\Sigma_u^{-1}.
\end{equation*}
This step can be computed efficiently using Woodbury identity; see Section~\ref{sec:woodbury}. Let
\begin{equation*}
\bar u:=\frac1N U(\rho,\phi,\beta)\,\ones[N],\quad \bar z:=\frac1N Z\,\ones[N],\quad
U_c:=U(\rho,\phi,\beta)-\bar u\,\ones[N]',\quad Z_c:=Z-\bar z\,\ones[N]'.
\end{equation*}
Then update $A$ by the profiled GLS formula
\begin{equation*}
A \leftarrow \widehat A(F;U)
=
\big(F'MF\big)^{-1}\,F'MU_c\,Z_c'\,\big(Z_cZ_c'\big)^{-1},
\end{equation*}
and update the free time effect by
\begin{equation*}
\delta \leftarrow \widehat\delta(F,A;U)=\bar u - F A\,\bar z.
\end{equation*}
With this ordering, $\widehat A(F;U)$ depends on $(u_i-\bar u)$ and $(z_i-\bar z)$.
See Section~\ref{sec:profile_delta_A} for details on the concentration of $\delta$ and $A$.
The displayed updates are analogous to first estimating a regression slope and then its intercept.
Equivalently, one may first concentrate out $\delta$, solve for $A$, and substitute the resulting $A$ back into the expression for $\delta$;
concentrating out $(\delta,A)$ jointly yields the same formulas.

\paragraph{Step 2 (E-step for the random loading heterogeneity $\eta_i$ given the transformed residual matrix).}

For each $i$, define
\begin{equation*}
e_i := u_i(\rho,\phi,\beta)-\delta-FAz_i \in\mathbb R^T.
\end{equation*}
Treating $\eta_i$ as missing data, the posterior is Gaussian with
\begin{align*}
V &:= \Var(\eta_i\mid u_i,z_i)
=
\big(\Sigma_\eta^{-1}+F'D_\varepsilon^{-1}F\big)^{-1},\\
m_i &:= \E(\eta_i\mid u_i,z_i)
=
V\,F'D_\varepsilon^{-1}e_i,
\qquad i=1,...,N.
\end{align*}
(Here $V$ is common across $i$ and is only $r\times r$.)

\paragraph{Step 3 (M-step for factor covariance $\Sigma_\eta$ and idiosyncratic variance matrix $D_\varepsilon$).}

Update the loading covariance by
\begin{equation*}
\Sigma_\eta \leftarrow V + \frac1N\sum_{i=1}^N m_i m_i'.
\end{equation*}
Let $r_i:=e_i-Fm_i$ denote the posterior mean of $\varepsilon_i$.
Since $\E(\varepsilon_i\varepsilon_i'\mid u_i)=r_i r_i' + FVF'$, the diagonal variances update elementwise as
\begin{equation*}
\sigma_t^2 \leftarrow \frac1N\sum_{i=1}^N r_{it}^2 \;+\; f_t'Vf_t,\qquad t=1,...,T,
\end{equation*}
where $f_t'\in\mathbb R^{1\times r}$ is row $t$ of $F$, and set $D_\varepsilon\leftarrow\diag(\sigma_1^2,...,\sigma_T^2)$.

\paragraph{Step 4 (M-step for factor path $F$; impose normalization for factor path after updating).}

Define
\begin{equation*}
s_i:=Az_i+m_i\in\mathbb R^r,\qquad S:=(s_1,...,s_N)\in\mathbb R^{r\times N}.
\end{equation*}
The closed-form ECM update for $F$ is
\begin{equation*}
\widetilde F \leftarrow \big(U(\rho,\phi,\beta)-\delta\,\ones[N]'\big)\,S'\,\big(SS' + NV\big)^{-1}.
\end{equation*}
Since $F$ is only identified up to rotation, we can post-process $\widetilde F$ by any invertible $r\times r$ matrix
$H$ without changing the likelihood.
Impose the normalization $F_{1:r,:}=I_r$ by the rotation
\begin{equation*}
H := \big(\widetilde F_{1:r,:}\big)^{-1},\qquad
F \leftarrow \widetilde F H,\qquad
A \leftarrow H^{-1}A,\qquad
\Sigma_\eta \leftarrow H^{-1}\Sigma_\eta (H^{-1})'.
\end{equation*}

\begin{remark}[Normalization with Time-Invariant Individual Heterogeneity]
	If time-invariant individual heterogeneity is included,
	the first factor is normalized to be constant over time,
	$f_{t1}=1$,
	so that
	\[
		F=[\ones[T],G],
		\qquad
		G\in\R^{T\times(r-1)},
	\]
	where $r$ is the total number of factors,
	including the constant factor.
	In this case,
	the identity-block normalization in \eqref{eq:F_normalization_theory} is replaced by
	\[
		F_{1:r,:}
		=
		\begin{pmatrix}
			1 & 0_{1\times(r-1)}\\
			\ones[r-1] & I_{r-1}
		\end{pmatrix},
	\]
	or equivalently,
	\[
		G_{1,\cdot}=0,
		\qquad
		G_{2:r,\cdot}=I_{r-1}.
	\]
	This normalization preserves the constant first factor while removing the remaining location and rotation indeterminacy.
	The factor update and subsequent rotation must then be restricted to preserve the constant column.
	The corresponding loading follows the same SERL representation as the other components of $\lambda_i$;
	its systematic component is captured by $Az_i$ and its residual variation by $\Sigma_\eta$,
	without estimating a separate parameter for every unit.
	When no constant factor is imposed,
	we retain the normalization $F_{1:r,:}=I_r$ used above.
\end{remark}

\paragraph{Step 5 (joint GLS update for the dynamic and slope coefficients, conditional on $\rho$).}
Fix $\rho$ and the current nuisance parameters $(\delta,F,A,\Sigma_\eta,D_\varepsilon)$,
hence $\Sigma_u=D_\varepsilon+F\Sigma_\eta F'$ and $M:=\Sigma_u^{-1}$.
Define the spatially transformed dependent variable
\begin{equation*}
\widetilde y_t(\rho):=B(\rho)y_t = (I_N-\rho W)y_t,\qquad t=1,...,T,
\end{equation*}
and for each unit $i$ stack
\begin{equation*}
\widetilde y_i(\rho):=\big(\widetilde y_{i1}(\rho),...,\widetilde y_{iT}(\rho)\big)'\in\mathbb R^T,\qquad
y_{i,-1}:=(y_{i0},y_{i1},...,y_{i,T-1})'\in\mathbb R^T,
\end{equation*}
\begin{equation*}
X_i:=\begin{pmatrix}x_{i1}'\\ \vdots\\ x_{iT}'\end{pmatrix}\in\mathbb R^{T\times K},\qquad
R_i:=[\,y_{i,-1}\ \ X_i\,]\in\mathbb R^{T\times (1+K)},\qquad
\vartheta:=\begin{pmatrix}\phi\\ \beta\end{pmatrix}\in\mathbb R^{1+K}.
\end{equation*}

Then the mean equation implied by \eqref{eq:Bsystem} and the SERL specification is
\begin{equation*}
\widetilde y_i(\rho)=R_i\vartheta+\delta+FAz_i+v_i,\qquad \Var(v_i\mid\mathcal C_N)=\Sigma_u.
\end{equation*}

\emph{Centered (profile-$\delta$) GLS.}
Since $\delta$ is a free $T\times 1$ time effect common across $i$, it is convenient to
difference out $\delta$ by cross-sectional demeaning. Let
\begin{equation*}
\bar{\widetilde y}(\rho):=\frac1N\sum_{i=1}^N \widetilde y_i(\rho),\quad
\bar R:=\frac1N\sum_{i=1}^N R_i,\quad
\bar z:=\frac1N\sum_{i=1}^N z_i,
\end{equation*}
and define centered objects
\begin{equation*}
\widetilde y_{i,c}(\rho):=\widetilde y_i(\rho)-\bar{\widetilde y}(\rho),\qquad
R_{i,c}:=R_i-\bar R,\qquad
z_{i,c}:=z_i-\bar z.
\end{equation*}
Then the (conditional) GLS update for $\vartheta=(\phi,\beta')'$ is the explicit closed form
\begin{equation}\label{eq:theta_GLS}
	\vartheta \leftarrow \widehat\vartheta(\rho)
	=
	\left(\sum_{i=1}^N R_{i,c}'MR_{i,c}\right)^{-1}
	\left(\sum_{i=1}^N R_{i,c}'M\big(\widetilde y_{i,c}(\rho)-FAz_{i,c}\big)\right).
\end{equation}
Equivalently, one may use the uncentered version with $\delta$ explicitly present; \eqref{eq:theta_GLS}
is exactly the profile-$\delta$ GLS solution because $M$ is common across $i$.

Given $\widehat\vartheta(\rho)$, form the transformed innovations
\begin{equation*}
u_t(\rho,\widehat\vartheta):=\widetilde y_t(\rho)-\widehat\phi(\rho)\,y_{t-1}-X_t\widehat\beta(\rho),
\qquad t=1,...,T,
\end{equation*}
and form $U\in\mathbb R^{T\times N}$ with $t$th row $u_t(\rho,\widehat\vartheta)'$.
The remaining nuisance updates, concentrating out $(\delta,A)$ and the ECM steps for
$(F,\Sigma_\eta,D_\varepsilon)$, proceed exactly as in Steps~1--4.

\subsection{Conditional Scalar Update for \texorpdfstring{$\rho$}{rho}: Outer Loop}

\paragraph{Step 6 (conditional scalar update for $\rho$).}

Since $\rho$ enters nonlinearly through $B(\rho)$ and the Jacobian term $T\log|B(\rho)|$,
there is no closed form for $\rho$.
Because $\rho$ is scalar,
we update it by a one-dimensional conditional maximization step.
This step holds the current values of
$(\phi,\beta,\delta,A,F,\Sigma_\eta,D_\varepsilon)$ fixed.
For any trial value of $\rho$,
form $U(\rho)\in\mathbb R^{T\times N}$ by stacking $u_t(\rho)'$ as its rows,
and define
\begin{equation*}
	E(\rho)
	=
	U(\rho)-\delta\ones[N]'-FAZ,
\end{equation*}
where $u_t(\rho)=B(\rho)y_t-\phi y_{t-1}-X_t\beta$.
Restrict the numerical search to a compact interval
$[\rho_{\min},\rho_{\max}]$ on which $B(\rho)=I_N-\rho W$ is nonsingular.
Up to terms that do not depend on $\rho$ within this conditional step,
the criterion to maximize is
\begin{equation}\label{eq:rho_profile_obj}
\ell_c^{\text{cond}}(\rho)
=
T\log|B(\rho)|
-\frac12\tr\big(ME(\rho)E(\rho)'\big).
\end{equation}
Maximize \eqref{eq:rho_profile_obj} over $[\rho_{\min},\rho_{\max}]$ by a standard one-dimensional routine,
e.g., golden-section/Brent's method, or by a coarse grid followed by local refinement.
After updating $\rho$,
rerun the inner loop to update the remaining parameters.

\begin{remark}[Computing $\log|B(\rho)|$]
	Throughout,
	$\log|B(\rho)|$ denotes the log absolute determinant $\log|\det B(\rho)|$.
	If the eigenvalues of $W$ are $\{\lambda_j\}_{j=1}^N$,
	possibly including complex conjugate pairs when $W$ is nonsymmetric,
	then
	\begin{equation*}
		\log|\det B(\rho)|
		=
		\sum_{j=1}^N\log|1-\rho\lambda_j|,
		\quad
		\frac{\partial}{\partial\rho}\log|\det B(\rho)|
		=
		-\operatorname{Re}\sum_{j=1}^N \frac{\lambda_j}{1-\rho\lambda_j}
		=
		-\tr\big(B(\rho)^{-1}W\big),
	\end{equation*}
	on any connected admissible region on which $B(\rho)$ is nonsingular.
	When $W$ has a real spectrum and the admissible region ensures $1-\rho\lambda_j>0$ for every $j$,
	the first expression reduces to $\sum_j\log(1-\rho\lambda_j)$.
	For a nonsymmetric $W$,
	conjugate eigenvalue pairs combine to give a real log determinant.
	The derivative is useful for derivative-based updates in $\rho$,
	and the eigenvalue representation makes evaluating $\ell_c^{\text{cond}}(\rho)$ fast.
	When feasible, the eigenvalues $\{\lambda_j\}_{j=1}^N$ can be computed once
	and stored for fast evaluation of $\log|B(\rho)|$ and its derivative at any $\rho$.
	For very large $N$,
	the log absolute determinant can instead be evaluated by sparse LU factorization or trace approximations.
\end{remark}

\subsection{Summary of the Estimation Algorithm}

Algorithm~\ref{alg:profile_rho_description} summarizes the inner-outer loop:
the outer loop performs a conditional scalar maximization in $\rho$,
while the inner loop alternates between GLS updates of $(\phi,\beta)$
and ECM-style updates of the nuisance parameters $(\delta,A,F,\Sigma_\eta,D_\varepsilon)$ given the current innovations.
One possible initialization strategy is discussed in Section~\ref{sec:initialization}.

\input{algorithm_profile_rho_description.tex}

\section{Model Extensions}\label{sec:extensions}

The baseline formulation is deliberately parsimonious, but two additions are
especially relevant in applications: a lagged spatial outcome and a spatial
autoregressive error. This section shows that both can be accommodated without
changing the fixed-$T$ logic of the estimator.
It also introduces an unrestricted second-order QML that is useful when
the outcome and error filters are difficult to distinguish empirically.

\subsection{Lagged Spatial Outcome}\label{subsec:extension-lagged-spatial-outcome}

Add $Wy_{t-1}$ to the structural equation:
\begin{equation}\label{eq:extension-spatiotemporal}
	B(\rho)y_t
	=
	\phi y_{t-1}+\xi Wy_{t-1}+X_t\beta
	+\ones[N]\delta_t+\Lambda f_t+\varepsilon_t.
\end{equation}
Conditional on $y_0$, the Jacobian from $(y_1',...,y_T')'$ to the
innovations is block lower triangular, with $B(\rho)$ on every diagonal block
and $-(\phi I_N+\xi W)$ on the first subdiagonal. Hence its determinant remains
$|B(\rho)|^T$; the lagged spatial outcome changes the transformed innovation,
but contributes no additional determinant term. The criterion in
\eqref{eq:ell_N} therefore remains valid after replacing $u_t(\theta)$ by
\begin{equation*}
	u_t(\theta_\xi)
	=
	B(\rho)y_t-\phi y_{t-1}-\xi Wy_{t-1}-X_t\beta,
	\quad\text{where}\quad
	\theta_\xi=(\rho,\phi,\xi,\beta')'.
\end{equation*}
For fixed $\rho$ and nuisance parameters, the GLS block simply adds
$Wy_{t-1}$ to the design matrix. Because the $t=1$ equation contains $Wy_0$,
the projection controls should also contain the corresponding information; the
$L=1$ specification in \eqref{eq:correct_projection} does so directly.

The dynamic reduced form is stable when
\begin{equation}\label{eq:extension-spatiotemporal-stability}
	\varrho\!\left[B(\rho)^{-1}(\phi I_N+\xi W)\right]<1,
\end{equation}
where $\varrho(\cdot)$ denotes spectral radius. For asymptotic analysis, the
inverse of the associated unit-stacked dynamic-spatial operator must also have
uniformly bounded row and column sums. Under these conditions and the analogues
of Assumptions~\ref{ass:SERL}--\ref{ass:score_clt}, the
proof of Theorem~\ref{thm:qml_asymptotics} applies to $\theta_\xi$: the estimator
is consistent and $\sqrt N$-asymptotically normal, with the same profiled
sandwich form as in \eqref{eq:Vtheta_hat}. Appendix~\ref{sec:extension-scores}
gives the additional corrected score contribution and the modified reduced-form
response blocks.

\subsection{Spatial Errors and the Second-Order QML}
\label{subsec:extension-spatial-errors}

Consider an outcome equation with a spatial autoregressive error,
\begin{align}
	B(\rho)y_t
	&=
	\phi y_{t-1}+\xi Wy_{t-1}+X_t\beta+WX_t\gamma
	+\ones[N]\delta_t+v_t,
	\label{eq:extension-spatial-error-outcome}\\
	B(\psi)v_t&=\Lambda f_t+\varepsilon_t.
	\label{eq:extension-spatial-error-process}
\end{align}

\begin{remark}[Placement of Spatial Error Filter]
Specification~\eqref{eq:extension-spatial-error-process} treats the common
and idiosyncratic innovations as entering before spatial propagation, so
$B(\psi)^{-1}$ applies to both. An alternative specification places
$\Lambda f_t$ outside the separate spatial-error process. This alternative placement is studied by
\textcite{li-yang-2021}.
 We maintain
\eqref{eq:extension-spatial-error-process}. Our specification treats the composite disturbance $v_t$
 itself as following a spatial autoregressive process, in parallel with the spatial autoregressive specification for the outcome.
 In addition, under the alternative placement,
non-Gaussian robust inference generally requires Li and Yang's hybrid
higher-moment correction rather than inference based solely on the score rearrangement used here.
\end{remark}

Applying the error filter $B(\psi)$ to the outcome equation gives the product
$B(\psi)B(\rho)$. This product is a polynomial of degree two in $W$.
To introduce a common notation for its two roots, initially label
$\kappa_1:=\rho$ and $\kappa_2:=\psi$, and write
$\kappa:=(\kappa_1,\kappa_2)'$.
Because the two filter components commute,
their ordering is immaterial.
Write
\begin{equation}\label{eq:extension-second-order-filter}
	B_*(\kappa)y_t
	:=(I_N-\kappa_1W)(I_N-\kappa_2W)y_t
	=(I_N-\rho_1W-\rho_2W^2)y_t,
\end{equation}
where $\rho_1=\kappa_1+\kappa_2$ and
$\rho_2=-\kappa_1\kappa_2$. Here $\rho$ remains the outcome spatial
coefficient from the baseline model, whereas the subscripted coefficients
$(\rho_1,\rho_2)$ describe the product polynomial. We refer to the resulting
filter as \emph{second-order} because it contains both $W$ and $W^2$.

Filtering of the spatial error also generates $W^2y_{t-1}$ and $W^2X_t$ terms.
Instead of imposing at the outset the nonlinear restrictions that link all of these coefficients
to $\psi$, we consider the more general transformed equation
\begin{equation}\label{eq:extension-second-order-model}
	B_*(\kappa)y_t
	=
	(\phi I_N+\xi_1W+\xi_2W^2)y_{t-1}
	+X_t\beta_0+WX_t\beta_1+W^2X_t\beta_2
	+\ones[N]\delta_t+\Lambda f_t+\varepsilon_t.
\end{equation}
We call \eqref{eq:extension-second-order-model} the \emph{unrestricted
second-order model}: ``second-order'' refers to the inclusion of $W^2$, while
``unrestricted'' means that the coefficients on the $I_N$, $W$, and $W^2$
terms are estimated separately rather than constrained to originate from
\eqref{eq:extension-spatial-error-outcome}--\eqref{eq:extension-spatial-error-process}.

The structural spatial-error model is a restricted special case. Expanding
$B(\psi)$ times the right-hand side of
\eqref{eq:extension-spatial-error-outcome} gives
\begin{equation}\label{eq:extension-spatial-error-restrictions}
\begin{aligned}
	\{\kappa_1,\kappa_2\}&=\{\rho,\psi\},
	&\xi_1&=\xi-\psi\phi,
	&\xi_2&=-\psi \xi,\\
	\beta_0&=\beta,
	&\beta_1&=\gamma-\psi\beta,
	&\beta_2&=-\psi\gamma.
\end{aligned}
\end{equation}
 Moreover,
$B(\psi)\ones[N]\delta_t=(1-\psi)\ones[N]\delta_t$ because
$W\ones[N]=\ones[N]$; this rescaling is absorbed into the freely estimated time
effect in \eqref{eq:extension-second-order-model}. The factor innovation
$\Lambda f_t+\varepsilon_t$ therefore retains the baseline
low-rank-plus-diagonal covariance structure.

The normalized conditional criterion for \eqref{eq:extension-second-order-model}
is
\begin{equation}\label{eq:extension-second-order-likelihood}
	\ell_{N,2}(\alpha_2)
	=
	\frac{T}{N}\sum_{j=1}^2\log|B(\kappa_j)|
	-\frac{1}{2N}\sum_{i=1}^N
	\left[\log|\Sigma_u|+e_{2i}'\Sigma_u^{-1}e_{2i}\right],
\end{equation}
where $e_{2i}$ is obtained from the transformed innovation in
\eqref{eq:extension-second-order-model}. Conditional on $(\kappa_1,\kappa_2)$,
all dynamic and control coefficients enter the same GLS block as before; the
two log-determinant terms correspond to the two components of $B_*(\kappa)$.
Thus, incorporating the spatial-error filter replaces the baseline scalar
outer search over $\rho$ with a two-dimensional optimization over
$(\kappa_1,\kappa_2)$, while leaving the inner GLS profiling steps unchanged.

Dynamic stability requires
\begin{equation}\label{eq:extension-second-order-stability}
	\varrho\!\left[
	B_*(\kappa)^{-1}(\phi I_N+\xi_1W+\xi_2W^2)
	\right]<1.
\end{equation}
The criterion is unchanged when $\kappa_1$ and $\kappa_2$ are interchanged.
Without imposing
\eqref{eq:extension-spatial-error-restrictions}, it does not identify which
root is the outcome coefficient $\rho$ and which is the error coefficient
$\psi$; only the symmetric filter coefficients $(\rho_1,\rho_2)$ are invariant
to relabeling. This distinction is useful in short panels, where direct
separation of $\rho$ and $\psi$ may be weak.
For theoretical purposes,
we label the roots by imposing $\kappa_1<\kappa_2$ and assume that the true root gap $\kappa_{20}-\kappa_{10}$ is bounded away from zero.
This ordering makes the root parameterization locally unique but does not assign structural roles to the two roots.
The equal-root case is excluded because the transformation
$(\kappa_1,\kappa_2)\mapsto(\rho_1,\rho_2)$ has Jacobian determinant $\kappa_2-\kappa_1$
and is therefore singular when the roots coincide.

Suppose this root-separation condition holds,
the ordered roots lie in the interior of the admissible set,
and the bounded-inverse, conditional-moment, smoothness, nonsingular
Hessian, and score-CLT conditions used for
Theorem~\ref{thm:qml_asymptotics} hold for the enlarged operator. Then the QML
estimator of the ordered-root parameter vector
$(\kappa_1,\kappa_2,\phi,\xi_1,\xi_2,\beta_0',\beta_1',\beta_2')'$ is consistent
and $\sqrt N$-asymptotically normal, with the profiled sandwich covariance in
\eqref{eq:Vtheta_hat}. Inference for the reported symmetric coefficients
$(\rho_1,\rho_2)$ follows by the delta method. The spatial score decomposition
must use the second-order reduced form for every endogenous direction;
Appendix~\ref{sec:extension-scores} records this modification. The same argument
covers the restricted spatial-error QML when the restrictions in
\eqref{eq:extension-spatial-error-restrictions} are imposed and locally
identified.

\section{Simulation}\label{sec:simulation}

We report two complementary experiments based on the model in
\eqref{eq:ysystem}. The first evaluates estimation and inference over a broad
parameter and network grid when the initial outcome is generated directly from
the factor loadings. The second considers the more demanding case in which the
initial outcome is inherited from a long-running spatial process and compares
alternative spatial enrichment orders in the SERL specification.
All estimators and covariance matrices reported in the Monte Carlo experiments are computed using
the companion R package \href{https://github.com/jessekelighine/dspserl}{\texttt{dspserl}}.

\subsection{Baseline Estimation and Inference}
\label{subsec:simulation-baseline}

\subsubsection{Design}

The baseline experiment has $N$ units, $T$ time periods, $K$ regressors, and
$r$ common factors. The spatial weights matrix $W$ is a row-normalized
$k$-nearest-neighbor graph constructed from random locations, with $k=8$.
The network is redrawn in each Monte Carlo replication using a seed separate
from the remaining DGP draws and is fixed over time within that replication.
The main-text experiment stores $W$ as an ordinary $N\times N$ array;
we refer to this as a \emph{dense-matrix} design.
This contrasts with the \emph{sparse-matrix} design presented in Appendix~\ref{subsec:aux-sparse-logdet},
which stores $W$ as a sparse matrix and utilizes sparse-matrix methods for the log-determinant.

For each $t$, the regressors are i.i.d.\ $x_{it}\sim\mathcal N(0,I_K)$ and
$\bar x_i$ denotes their time average. We generate loadings directly as
$\lambda_i\sim\mathcal N(0,I_r)$ and set the initial condition by
\begin{equation*}
	y_{i0}=\ones[r]'\lambda_i+\nu_i,
\end{equation*}
where $\nu_i\sim\mathcal N(0,1)$ is independent of $\lambda_i$.
Thus $y_{i0}$ is informative about the conditional mean of $\lambda_i$, while
$\bar x_i$ remains independent of the loadings.
Joint Gaussianity gives
\begin{equation*}
	\E(\lambda_i\mid y_{i0})
	=
	\frac{\ones[r]}{r+1}y_{i0},
	\qquad
	\Var(\lambda_i\mid y_{i0})
	=
	I_r-\frac{\ones[r]\ones[r]'}{r+1}.
\end{equation*}
Time effects $\delta_t$ follow a deterministic cycle,
idiosyncratic variances are time-varying,
and common shocks are generated as $f_t\sim\mathcal N(0,I_r)$.
We set $z_i=(\bar x_i',y_{i0})'$ in the SERL specification \eqref{eq:proj}. In
this DGP,  population projection parameters are
$A
	=
	\left[0_{r\times K},\frac{\ones[r]}{r+1}\right],$ and
	$\Sigma_\eta
	=
	I_r-\frac{\ones[r]\ones[r]'}{r+1}.
$
Both are nuisance parameters implied by the DGP.

Given $(W,\rho)$, we solve for $y_t$ using the spatial filter
$B(\rho)=I_N-\rho W$:
\begin{equation}
	y_t=B(\rho)^{-1}\big(\phi y_{t-1}+X_t\beta+\ones[N]\delta_t+\Lambda f_t+\varepsilon_t\big),
	\qquad t=1,...,T,
\end{equation}
with $y_{i0}$ initialized as above and
$\varepsilon_t\sim\mathcal N(0,\sigma_t^2 I_N)$.
The implementation factors $B(\rho)$ once by LU decomposition and reuses that
factorization to solve the system in every period; it does not form
$B(\rho)^{-1}$ explicitly.

We fix $r=2$ and $K=2$ with $\beta=(0.8,-0.3)$.
To hold persistence associated with the largest eigenvalue of the row-normalized
$W$ constant across values of $\rho$, we set
\begin{equation}
	\phi=0.5(1-\rho).
\end{equation}
Because the largest eigenvalue of $W$ is one,
the corresponding dynamic coefficient is $\phi/(1-\rho)=0.5$
in every design.
Accordingly, $\rho+\phi$ equals $0.60$, $0.75$, or $0.90$
and remains below one.
The time effects and idiosyncratic variances follow deterministic schedules,
\begin{equation}
	\delta_t=0.3\sin\left(2\pi \frac{t}{T}\right),
	\qquad
	\sigma_t^2=5+5 \frac{t}{T}.
\end{equation}
Common shocks are $f_t\sim\mathcal N(0,I_r)$.

For the bias and dispersion experiment, we vary $(N,T,\rho)$ according to the grid
\begin{equation}
	\begin{aligned}
		T&\in\{5,10,20\}, & N&\in\{500,1000\},\\
		\rho&\in\{0.2,0.5,0.8\},
		& \phi&=0.5(1-\rho).
	\end{aligned}
\end{equation}
for a total of $18$ parameter configurations.
For each configuration we estimate $(\rho,\phi,\beta)$ by the block-coordinate QML procedure.
We evaluate $\log|B(\rho)|$ using the eigenvalues of $W$.
Appendix~\ref{subsec:aux-sparse-logdet} repeats these $18$ configurations
using a 30-term Hutchinson trace approximation based on 25 Rademacher vectors.
We also compare this approximation with a richer trace calculation
and exact sparse LU factorization on identical panels.

Both experiments initialize the spatial search with the same truth-independent
nine-point grid over $[-0.95,0.95]$.
The inner loop permits at most 200 iterations with tolerance $10^{-8}$,
and the scalar optimizer for $\rho$ uses tolerance $10^{-10}$.
Each design requests 1000 Monte Carlo replications.
For this experiment, the outer block-coordinate loop permits at most 100
iterations.
A fit enters the bias and dispersion calculations if it either converges
strictly or has a converged inner loop and a final $\rho$ stationarity gap no
larger than $2\times10^{-6}$.
Between 995 and 1000 estimates per design meet this criterion.
The coverage experiment uses the same parameter grid.
Each design requests 1000 Monte Carlo replications.
The outer block-coordinate loop permits at most 300 iterations and requires
the strict $10^{-8}$ convergence criterion.
For each strictly converged fit,
we construct nominal $95\%$ intervals using the martingale-sample covariance estimator.
We report the uncorrected version as a benchmark and the pair-only HC1 correction in \eqref{eq:Omega_martingale_pair_hc1} as the preferred estimator.

\subsubsection{Results}

Table~\ref{tab:sim-bias-sd-dense} reports bias, multiplied by $100$,
and the standard deviation (SD) for $\rho$, $\phi$, and $\beta$
across the 18 dense-matrix designs.
The biases are generally small,
although $\hat\rho$ has a systematic downward bias.
In the original parameter units,
its average bias is $-0.0032$ and its largest absolute bias is about $0.0051$.
The SD generally decreases with larger $N$ and $T$,
and $\rho$ is estimated more precisely when its true value is larger,
consistent with stronger spatial dependence providing more information about
$\rho$. In particular, the SD of $\hat\rho$ decreases with $N$ and $T$ when
$\rho=0.8$.

\input{simulation_bias_sd.tex}

Table~\ref{tab:sim-coverage-dense} reports coverage for the dense-matrix designs without a finite-sample correction
and with the preferred pair-only HC1 correction in \eqref{eq:Omega_martingale_pair_hc1}.
Coverages for the regression slopes and $\rho$ are close to the nominal level,
while coverage for $\phi$ is lowest when $T=5$.
Appendix~\ref{subsec:aux-symmetric} holds $\rho+\phi$ fixed while varying its spatial and dynamic components.
The experiment confirms that first-order Wald coverage for $\phi$ can be inaccurate when each unit contributes only five transitions,
whereas increasing the time dimension to $T=20$ eliminates under-coverage.
We view this as an understandable short-panel limitation rather than a failure of the large-$N$ theory.
The pair-only HC1 correction modestly improves coverage for $\rho$ without materially changing the other results,
so we use it as our preferred finite-sample covariance estimator and retain the uncorrected columns as a benchmark.

We also conduct the complete experiment under the sparse-matrix design.
Appendix~\ref{subsec:aux-sparse-logdet} reports the corresponding bias, dispersion, and coverage results.
The parameters remain well estimated and dispersion is similar to that reported in Table~\ref{tab:sim-bias-sd-dense},
but coverage for $\rho$ is generally lower.
The diagnostic using identical panels shows that much of this difference is attributable to approximation error
in the log determinant.

\input{simulation_coverage.tex}

\subsection{Long-Running Initial Conditions and SERL Enrichment}
\label{subsec:simulation-initial-condition-projection}

The preceding experiments generate $y_0$ directly from the factor loadings.
We now consider the more demanding case in which $y_0$ is itself an outcome
from a long-running dynamic spatial process. This design evaluates the spatial
enrichment in \eqref{eq:correct_projection}. In particular, it asks whether a
small number of spatial transformations of the initial condition and regressor
summaries can capture the information about the loadings that propagates
through the spatial multiplier.

\subsubsection{Design}

For each replication, we draw $\lambda_i\sim\mathcal N(0,I_2)$ and generate a
stationary regressor component according to
\begin{equation*}
	h_{is}=0.8h_{i,s-1}+\sqrt{1-0.8^2}\,v_{is},
	\qquad v_{is}\sim\mathcal N(0,I_2).
\end{equation*}
The observed regressors are
\begin{equation*}
	x_{is}=0.5\lambda_i+\sqrt{1-0.5^2}\,h_{is},
\end{equation*}
so that both the initial outcome and the regressors are informative about the
individual loadings. This makes the projection exercise more demanding than a
design in which the loadings are related only to $y_0$.
The regressor process is initialized from its stationary distribution.
Starting from zero,
we run the outcome process for 100 presample periods and take its terminal
presample value as $y_0$.
We then retain the next $T$ observations as the estimation sample.
The structural parameters are $\phi=0.3$ and $\beta=(0.8,-0.3)'$,
the idiosyncratic innovations have unit variance,
and the time effects are zero.
There are two common factors whose realizations are independently standard
normal over the presample and sample periods.
We use a sparse row-normalized $k$-nearest-neighbor network with $k=8$ and
consider
\begin{equation*}
	T\in\{5,10\},\qquad N\in\{500,1000\},\qquad
	\rho\in\{0.2,0.5\}.
\end{equation*}
Unlike the experiments initialized directly at $y_0$, a long-running process
must satisfy a stability restriction during the presample recursion. For a
row-normalized $W$, the leading persistence is $\phi/(1-\rho)$,
which equals $0.375$ and $0.6$ for $\rho=0.2$ and $\rho=0.5$, respectively.
For each $N$, the network is held fixed across replications and shared across
the corresponding $(T,\rho)$ designs.
For each $T$, the factor path is held fixed across replications and shared
across the corresponding $(N,\rho)$ designs.
Each of the eight designs contains 500 Monte Carlo panels.

Let $Q_X$ contain the unit-specific time averages of the two observed
regressors. The full projection of order $L$ uses
\begin{equation*}
	z_i^{(L)}=\big[(Q_X)_i',(WQ_X)_i',...,(W^LQ_X)_i',
	(y_0)_i,(Wy_0)_i,...,(W^Ly_0)_i\big]',
\end{equation*}
and we consider $L=0,1,2,3$. These specifications contain, respectively,
$3$, $6$, $9$, and $12$ controls. We also consider a parsimonious $y$-only
$L=2$ specification that includes $Q_X$, $y_0$, $Wy_0$, and $W^2y_0$, but
does not include spatial transformations of $Q_X$. The same simulated panel
is used for all five projection specifications, so comparisons across them are
paired.
Each control is centered and standardized across units before estimation;
this rescaling does not change the span of the projection.
The log determinant is evaluated exactly by sparse LU factorization.
This avoids confounding the projection comparison with the stochastic approximation error documented in Appendix~\ref{subsec:aux-sparse-logdet}.
Inference uses the martingale-sample covariance estimator with the preferred pair-only HC1 correction in \eqref{eq:Omega_martingale_pair_hc1}.

\subsubsection{Results}

Table~\ref{tab:sim-initial-projection-summary} summarizes projection quality
and estimator performance across the eight designs.
Under
the conventional $L=0$ projection, the average maximum absolute correlation
between the loading residual and the first omitted spatial transformation of
$y_0$ is $0.167$, while the corresponding diagnostic for the regressor
summaries is $0.116$. Adding the first spatial lag reduces these diagnostics to
$0.018$ and $0.032$, respectively.
It also reduces the RMSE of $\hat\rho$ from $0.0281$ to $0.0204$ and the RMSE of $\hat\phi$ from $0.0238$ to $0.0181$.
Average coverage rises from $0.786$ to $0.944$ for $\rho$ and from $0.852$ to $0.939$ for $\phi$.
Coverage of $\beta_1$ rises from $0.910$ to $0.931$, while coverage of $\beta_2$ rises from $0.941$ to $0.947$.

Higher-order spatial transformations continue to reduce the omitted-variable
diagnostics, but they provide no further improvement in estimation. Relative
to full $L=1$, the full $L=2$ and $L=3$ projections have slightly larger bias
and RMSE for $\rho$ and $\phi$, consistent with a modest finite-sample cost
from estimating additional projection coefficients. The $y$-only $L=2$
projection performs similarly to full $L=1$ and has slightly smaller aggregate
bias and RMSE. It nevertheless leaves more residual correlation with omitted
spatial transformations of $Q_X$ than the full projection. Thus, the spatial
transformations of $y_0$ account for most of the estimation gain in this DGP,
while the spatial transformations of $Q_X$ improve the quality of the loading
projection without materially changing the structural estimates.

Table~\ref{tab:sim-initial-projection-coverage} reports design-specific
coverage. The benefit of the richer projection is largest when spatial
dependence is strong. Averaging over the four designs with $\rho=0.5$, moving
from full $L=0$ to full $L=1$ raises coverage from $0.665$ to $0.944$ for
$\rho$ and from $0.768$ to $0.937$ for $\phi$.
For $(T,N,\rho)=(5,500,0.5)$,
the corresponding changes are from $0.642$ to $0.918$ and from $0.750$ to $0.936$.
For the most difficult design,
$(T,N,\rho)=(5,1000,0.5)$, coverage rises from $0.312$ to $0.935$ for $\rho$
and from $0.496$ to $0.910$ for $\phi$.
Overall, moving from $L=0$ to $L=1$ substantially improves coverage
and supports the practical value of low-order SERL enrichment,
particularly when spatial dependence is stronger.

\input{simulation_initial_condition_projection.tex}

\section{Application: Female Labor-Force Participation}\label{sec:application}

We apply the fixed-$T$ estimator to the county-level female labor-force
participation data studied by \textcite{fogli-veldkamp-2011} and reexamined by
\textcite{tziolas-elhorst-2023}. The data contain 3,074 US counties observed at
decennial intervals from 1940 through 2000. Female labor-force participation
rose from 18.49 percent to 54.69 percent over this period. The observed controls
are the urban and rural-farm population shares, average education, population
density, and manufacturing wages. The economic question is whether changes in
participation diffuse across neighboring counties, potentially with a
one-decade lag.
We estimate all empirical specifications using the companion R package \href{https://github.com/jessekelighine/dspserl}{\texttt{dspserl}}.

\subsection{Empirical Specification and Samples}

Let $y_t$ be the $N\times1$ vector of county female labor-force participation
rates in decade $t$, and let $W$ be a row-normalized spatial-weights matrix, so
that $Wy_t$ contains weighted averages of neighboring counties' participation
rates. Writing $X_t$ for the non-spatial controls, we organize the application
around the encompassing outcome equation
\begin{equation}\label{eq:application-model}
	y_t
	=
	\rho Wy_t+\phi y_{t-1}+\xi Wy_{t-1}
	+X_t\beta+WX_t\gamma
	+\ones[N]\delta_t+\Lambda f_t+\varepsilon_t.
\end{equation}
The vector $\gamma$ contains the spatial Durbin coefficients on the spatial
lags of the controls. When a spatially lagged control is excluded, its
corresponding element of $\gamma$ is set to zero. The spatially lagged
controls can be stacked with $X_t$ and therefore enter the regressor block of
the baseline model.
The coefficient $\rho$ captures contemporaneous outcome dependence, whereas
$\xi$ captures dependence on neighboring counties' outcomes one decade
earlier. We first impose $\xi=0$, yielding the baseline model analyzed above,
and then estimate $\xi$ freely using the extension in
Section~\ref{subsec:extension-lagged-spatial-outcome}.

The lagged-neighbor term is present in the empirical model of
\textcite{fogli-veldkamp-2011}, who omit the contemporaneous outcome lag.
\textcite{tziolas-elhorst-2023} place both $Wy_t$ and $Wy_{t-1}$ in a broader
dynamic spatial Durbin model that also contains a spatial autoregressive error.
In this empirical application, the common shocks $f_t$ may represent nationwide
developments affecting the benefits and costs of female employment, including
changes in wages and discrimination,
the diffusion of household appliances,
the declining physical demands of jobs,
improved control over fertility,
and changing social norms.

Following \textcite{tziolas-elhorst-2023}, we consider three covariate and
spatial-weight specifications. The first includes wages and spatial lags of all
five controls. Requiring a balanced panel leaves 1,569 counties, and $W$ is the
row-normalized six-nearest-neighbor matrix constructed from their coordinates.
The second excludes wages but retains spatial lags of the other four controls.
The third also removes the spatially lagged controls. The last two
specifications contain 3,066 counties and use the original binary-contiguity
matrix after subsetting and row normalization. Because the 1940 outcome is the
initial condition, each estimation sample has six periods, covering
1950--2000.

All main application tables use two common factors.
This parsimonious choice is natural with only six estimation periods and matches the two-factor specification of \textcite{tziolas-elhorst-2023}.
Appendix~\ref{sec:application-factor-rank} reports three-factor estimates of the second-order model as a robustness check.
With six estimation periods, $r=3$ is the largest factor rank
for which the covariance component is estimable: under the factor
normalization, $r=3$ uses all 21 distinct elements of a $6\times6$ covariance
matrix, whereas $r=4$ would require 24 covariance parameters and is therefore
not feasible. The projection sets $L=1$ in
\eqref{eq:correct_projection}: $z_i$ contains time averages of the non-spatial
controls, their first spatial lags, $y_{i,1940}$, and $(Wy_{1940})_i$. The
controls are centered and standardized. We evaluate the spatial Jacobian
$\log|B(\rho)|$ exactly by sparse LU factorization.
This avoids the additional numerical dispersion from the faster trace approximation documented in Appendix~\ref{subsec:aux-sparse-logdet}.
All standard errors for our application estimates use the martingale-sample covariance estimator with the pair-only HC1 correction in \eqref{eq:Omega_martingale_pair_hc1}.

\subsection{Baseline Model}

The baseline sets $\xi=0$ in \eqref{eq:application-model}. It is therefore an
instance of \eqref{eq:ysystem}, with included spatially lagged controls treated
as additional observed regressors, and is covered by
Theorem~\ref{thm:qml_asymptotics}. Because $W\ones[N]=\ones[N]$, a spatially
uniform component evolves with dynamic coefficient $\phi/(1-\rho)$, and its
long-run multiplier has denominator $1-\rho-\phi$. We therefore call
$s:=\rho+\phi$ the \emph{stability sum} in the baseline model. When $\rho<1$
and $\phi\geq0$, stability of this component is equivalent to $s<1$; full
dynamic stability requires the spectral-radius condition in
\eqref{eq:extension-spatiotemporal-stability}.

To compare the estimated covariate effects with the historical change in
participation, consider a permanent common increase $\Delta x_k$ in control
$k$, holding the other controls, time effects, and common factors fixed. The
change between the initial and new steady states satisfies
\begin{equation*}
	\left[(1-\phi)I_N-\rho W\right]\Delta y_k^*
	=
	(\beta_k I_N+\gamma_k W)\ones[N]\Delta x_k.
\end{equation*}
Thus, premultiplying by $N^{-1}\ones[N]'$ gives the implied change in average
participation. Let $\Delta x_k$ equal the observed 1940--2000 change in the
sample mean of control $k$, and let
$\Delta\mathrm{LFP}=54.69-18.49=36.20$ denote the corresponding percentage-point
increase in mean female labor-force participation. Dividing the predicted
average change by $\Delta\mathrm{LFP}$ gives the reported contribution:
\begin{equation*}
	100\times\frac{1}{N}\ones[N]'
	\left[(1-\phi)I_N-\rho W\right]^{-1}
	(\beta_k I_N+\gamma_k W)\ones[N]
	\frac{\Delta x_k}{\Delta\mathrm{LFP}}.
\end{equation*}
The changes $\Delta x_k$ are those reported by
\textcite{tziolas-elhorst-2023}.

\input{application_female_lfp_baseline.tex}

Table~\ref{tab:application-female-lfp-baseline} shows that the contemporaneous spatial coefficient ranges from $0.479$ to $0.521$, while the own lag ranges from $0.364$ to $0.467$.
The resulting stability sums are $0.893$, $0.946$, and $0.870$, relatively large but below one.
The observed controls account for $116.2\%$, $73.3\%$, and $108.4\%$ of the historical increase across the three columns.
In the specification without wages or $WX$, the fall in farm population and increase in education produce the largest contributions.
The contribution sum of $108.4\%$ means that, holding the time effects and common factors fixed, the observed covariate changes imply a long-run increase in female labor-force participation of about $39.2$ percentage points, exceeding the observed increase of $36.2$ percentage points.

\subsection{Extension with a Lagged Spatial Outcome}

We next estimate \eqref{eq:application-model} with $\xi$ unrestricted, thereby
combining the delayed-neighbor channel of \textcite{fogli-veldkamp-2011} with
the contemporaneous outcome lag in \textcite{tziolas-elhorst-2023}.
Section~\ref{subsec:extension-lagged-spatial-outcome} extends the likelihood,
estimator, and asymptotic result, while Appendix~\ref{sec:extension-scores}
gives the corresponding spatial score decomposition. Thus, $Wy_{t-1}$ is not
treated as an ordinary external regressor.

With $\xi$ unrestricted, the stability sum becomes
$s:=\rho+\phi+\xi$. A spatially uniform component now evolves with dynamic
coefficient $(\phi+\xi)/(1-\rho)$, and its long-run multiplier has denominator
$1-s$. When $\rho<1$ and $\phi+\xi\geq0$, stability of this component is
equivalent to $s<1$; full stability continues to require the spectral-radius
condition in \eqref{eq:extension-spatiotemporal-stability}.

For comparison with the published decomposition, the contribution of control
$k$ is now
\begin{equation*}
	100\times\frac{1}{N}\ones[N]'
	\left[(1-\phi)I_N-(\rho+\xi)W\right]^{-1}
	(\beta_k I_N+\gamma_k W)\ones[N]
	\frac{\Delta x_k}{\Delta\mathrm{LFP}}.
\end{equation*}
For the specification without wages or spatially lagged controls, we also
report $L=0$ and $L=2$ estimates.

\input{application_female_lfp.tex}
\input{application_projection_robustness.tex}

Table~\ref{tab:application-female-lfp-extension} compares the proposed QMLE
with the corresponding bias-corrected QML estimates of
\textcite{tziolas-elhorst-2023}.
The estimated coefficient $\xi$ is negative in all three specifications, at $-0.073$, $-0.150$, and $-0.025$, while $\widehat\rho$ ranges only from $0.509$ to $0.530$.
The raw $\xi$ is not, however, the reduced-form effect corresponding to the lagged-neighbor coefficient in \textcite{fogli-veldkamp-2011}.
The dynamic
transition satisfies
\begin{equation*}
	B(\rho)^{-1}(\phi I_N+\xi W)
	=
	\phi I_N+(\xi+\rho\phi)W
	+\rho(\xi+\rho\phi)W^2+\cdots.
\end{equation*}
The leading neighbor-lag effect $\xi+\rho\phi$ is positive in every
specification, taking values $0.158$, $0.152$, and $0.173$.
Table~\ref{tab:application-projection-robustness} varies the SERL enrichment order
in the specification without wages or $WX$. Moving from $L=1$ to $L=2$ changes
the allocation between $\phi$ and $\xi$, but leaves $\rho$, the stability sum,
and the leading neighbor-lag effect relatively similar. By contrast, $L=0$
materially changes the dynamic coefficients.

\subsection{Second-Order QML}

The first-order estimates deliberately omit a spatial-error filter. To assess
whether this omission drives the covariate decomposition, our preferred
specification is the unrestricted second-order QML in
\eqref{eq:extension-second-order-model}. In the application it takes the form
\begin{equation*}
\begin{split}
	B_*(\kappa)y_t
	={}&(\phi I_N+\xi_1W+\xi_2W^2)y_{t-1}\\
	&+X_t\beta_0+WX_t\beta_1+W^2X_t\beta_2
	+\ones[N]\delta_t+\Lambda f_t+\varepsilon_t.
\end{split}
\end{equation*}
For the first two empirical specifications, the control polynomial includes
$X_t$, $WX_t$, and $W^2X_t$. For the specification that originally excludes
$WX_t$, it includes $X_t$ and $WX_t$: the latter is required because filtering
the original $X_t\beta$ term by $B(\psi)$ generates a first spatial lag. The
model contains no separately labeled spatial-error coefficient. Instead, it
allows the second-order filter and the transformed lag and control coefficients
implied by \eqref{eq:extension-spatial-error-restrictions} without imposing
those restrictions.

Let
\begin{equation*}
	Q_2
	=(1-\phi)I_N-(\rho_1+\xi_1)W-(\rho_2+\xi_2)W^2.
\end{equation*}
The contribution of control $k$ is computed as
\begin{equation*}
	100\times\frac{1}{N}\ones[N]'Q_2^{-1}
	(\beta_{0k}I_N+\beta_{1k}W+\beta_{2k}W^2)\ones[N]
	\frac{\Delta x_k}{\Delta\mathrm{LFP}},
\end{equation*}
with an omitted coefficient set to zero. We report the symmetric filter
coefficients $(\rho_1,\rho_2)$, not a structural assignment of the two roots,
and use the corrected covariance construction in
Appendix~\ref{sec:extension-scores}.

Because $W\ones[N]=\ones[N]$,
the reduced-form transition maps a spatially uniform change into another
uniform change.
We summarize its one-period persistence by the dynamic coefficient for a
uniform change,
\begin{equation*}
	d_{\mathrm{unif}}
	:=
	\frac{\phi+\xi_1+\xi_2}{(1-\kappa_1)(1-\kappa_2)}
	=
	\frac{\phi+\xi_1+\xi_2}{1-\rho_1-\rho_2}.
\end{equation*}
A value below one means that a spatially uniform change contracts from one
period to the next;
stability across all eigenvalues of $W$ is governed by
\eqref{eq:extension-second-order-stability}.

\input{application_female_lfp_second_order.tex}

Table~\ref{tab:application-second-order} reports the preferred estimates. The
two symmetric filter coefficients are stable across specifications:
$\widehat\rho_1$ ranges from $0.199$ to $0.219$, and
$\widehat\rho_2$ from $0.453$ to $0.487$.
The dynamic coefficient for a spatially uniform change ranges from $0.901$ to $0.930$, below one in every case.
The coefficient on education remains positive and precisely estimated
after the second-order filter and the additional spatial transformations are
included.
Its estimates are $0.527$, $0.646$, and $0.536$, with standard errors $0.084$, $0.072$, and $0.065$.
The corresponding contributions are $82.4\%$, $73.7\%$, and $71.8\%$, with delta-method standard errors of $38.6$, $37.2$, and $26.2$ percentage points.
Educational gains provide the largest positive contribution in all three specifications.
The decline in the farm population also produces a positive estimated contribution, but the contribution is not statistically significant.

\subsection{Discussion and Comparison with Earlier Estimates}

\textcite{fogli-veldkamp-2011} omit the contemporaneous spatial lag, corresponding to $\rho=0$ in our notation.
Their OLS estimates of the own and neighbor lags are $0.664$ and $0.195$, respectively;
their preferred GMM estimates are $0.916$ and $0.570$, whose sum exceeds one.
In our baseline model, contemporaneous feedback induces a first-order coefficient
$\rho\phi$ on $Wy_{t-1}$, ranging from $0.184$ to $0.224$. In the extension, the corresponding coefficient
 $\xi+\rho\phi$ ranges from $0.152$ to $0.173$.
 Thus, our reduced-form neighbor effect is close to their OLS estimate of $0.195$ on $Wy_{t-1}$.

\textcite{tziolas-elhorst-2023}, by contrast, allocate most spatial propagation
to $Wy_{t-1}$. Across their specifications, its coefficient is $0.504$,
$0.381$, and $0.454$, while the contemporaneous coefficient is $-0.075$,
$-0.340$, and $0.087$. They simultaneously estimate spatial-error coefficients
of $0.595$, $0.830$, and $0.432$. This is a specification
difference.  Our model sets the spatial-error coefficient to zero.
This difference also helps explain why our contemporaneous coefficient remains
near $0.5$. Because the same $W$ enters the outcome and error filters in
\textcite{tziolas-elhorst-2023}, the two components can be difficult to
distinguish with six periods. Denoting their outcome coefficient by $\rho$ and
their spatial-error coefficient by $\psi$, the two filters satisfy
\begin{equation}\label{eq:application-two-spatial-filters}
	(I_N-\psi W)(I_N-\rho W)
	=I_N-(\rho+\psi)W+\rho\psi W^2.
\end{equation}
Thus, a model with one first-order spatial filter may load part of an omitted
error filter onto its outcome coefficient. Consistent with this interpretation,
subtracting the sum of their estimated outcome and error coefficients from our
estimated spatial coefficient gives $0.009$, $0.018$, and $-0.011$ in the
three specifications. These are not parameter equivalences: the product also
induces the $W^2$ term in
\eqref{eq:application-two-spatial-filters} and spatially transforms the dynamic
and covariate components.
Moreover, a direct two-filter diagnostic reaches distinct local likelihood
optima from different starting values in every specification at both factor ranks.
We therefore do not treat the two roots as separately identified
outcome and error coefficients.

The central substantive difference concerns education, a covariate for which
one would expect a meaningful relationship with female labor-force
participation. The OLS estimate reported by
\textcite{fogli-veldkamp-2011} is $0.643$ with a standard error of $0.036$,
close to the range from $0.527$ to $0.646$ across our preferred second-order
specifications in Table~\ref{tab:application-second-order}. Their three-lag difference-GMM
estimate, $-0.975$ with a standard error of $0.587$, is instead imprecise,
indicating that their education estimate is sensitive to the estimator. The
corresponding estimates of \textcite{tziolas-elhorst-2023} are $0.017$, $0.060$,
and $0.004$, all statistically insignificant. Our preferred estimates assign education
contributions of $82.4\%$, $73.7\%$, and $71.8\%$, compared with $-0.5\%$,
$1.9\%$, and $0.2\%$ in their corresponding results reported in
Table~\ref{tab:application-female-lfp-extension}.
Some of education's association could be mediated through lagged outcomes or
absorbed by the spatial-error process. Nevertheless, the preferred estimates
give education a positive own coefficient and a substantial long-run
contribution after admitting the polynomial terms generated by such a filter.
The education discrepancy therefore is not explained solely by omission of the
spatial-error filter. It may also reflect differences in the treatment of common
shocks, initial conditions, and unit heterogeneity. By projecting county-specific
factor loadings on observed histories rather than estimating each loading
separately, our fixed-$T$ procedure may allocate persistent variation between
education and heterogeneous exposure to common shocks differently.
Taken together, our estimates are more consistent with the literature documenting an
important positive relationship between women's education and
labor-force participation; see, e.g., \textcite{goldin-2006}.

\section{Conclusion}\label{sec:conclusion}

This paper develops a likelihood-based framework for dynamic spatial panels
with common shocks in the fixed-$T$, large-$N$ regime.
The SERL specification and variance-components representation avoid estimating unit-specific
incidental parameters, while the conditional likelihood retains the spatial
Jacobian and the dynamic recursion without a parameter-dependent temporal
Jacobian. We also propose an efficient block-coordinate profile
algorithm, which exploits the structure of the likelihood to replace a
high-dimensional joint optimization with closed-form and GLS updates, low-rank
covariance updates, and a scalar search over the spatial coefficient. We
establish large-$N$ distribution theory and implement spatially corrected
sandwich inference using the score decomposition of
\textcite{li-yang-2021}. A lagged spatial outcome fits into the existing GLS
block, and a spatial-error filter preserves the profiling structure while
requiring a two-dimensional outer optimization. The baseline simulations show
small estimation error and generally near-nominal coverage.

In the empirical application, the lagged-neighbor effect is close to the original
OLS estimate of \textcite{fogli-veldkamp-2011}. Under the preferred second-order
specification, education remains positive and economically important after
allowing the transformations generated by a spatial-error filter; its
coefficient is also close to their positive OLS estimate.
These results accord with the literature documenting an important positive
relationship between women's education and labor-force participation.

\printbibliography[heading=bibintoc,title={References}]

\clearpage
\appendix

\begin{center}
\textbf{\large Appendix: Technical Details and Additional Results.}
\end{center}

\section{Efficient Evaluation: Woodbury and Determinant Lemma}\label{sec:woodbury}

Since $\Sigma_u=D_\varepsilon + F\Sigma_\eta F'$ is ``diagonal + rank-$r$'',
we have the Woodbury identity
\begin{align*}
\Sigma_u^{-1}
&=
D_\varepsilon^{-1}
-
D_\varepsilon^{-1}F\big(\Sigma_\eta^{-1}+F'D_\varepsilon^{-1}F\big)^{-1}F'D_\varepsilon^{-1},
\\
|\Sigma_u|
&=
|D_\varepsilon|\;|\Sigma_\eta|\;\big|\Sigma_\eta^{-1}+F'D_\varepsilon^{-1}F\big|
\end{align*}
These identities make likelihood evaluation feasible for
fixed $T$ and small $r$, even when $N$ is large.

\section{Concentrating Out Time Effects and the Projection Matrix}\label{sec:profile_delta_A}

Given $(F,\Sigma_u)$, the mean component $\delta+FAz_i$ enters linearly.
Let $Z:=(z_1,...,z_N)\in\R^{q\times N}$ and $U:=(u_1,...,u_N)\in\R^{T\times N}$.
Then
\begin{equation*}
\E[U\mid\mathcal C_N]= \delta\,\ones[N]' + FAZ.
\end{equation*}
Let $M:=\Sigma_u^{-1}$. Up to an additive constant, the conditional (negative) log-likelihood contribution from the quadratic form is
\begin{equation*}
Q(\delta,A;F)
:=
\tr\!\left(M\,(U-\delta\,\ones[N]' - FAZ)(U-\delta\,\ones[N]' - FAZ)'\right).
\end{equation*}

\paragraph{Step 1: concentrate out \texorpdfstring{$\delta$}{delta}}

For fixed $(F,A)$, differentiating $Q(\delta,A;F)$ with respect to $\delta$ yields
\begin{equation*}
\sum_{i=1}^N M\,(u_i-\delta-FAz_i)=0
\quad\Longrightarrow\quad
\widehat \delta(F,A)=\bar u - FA\,\bar z,
\end{equation*}
where
\begin{equation*}
\bar u := \frac{1}{N}U\ones[N]\in\R^T,
\qquad
\bar z := \frac{1}{N}Z\ones[N]\in\R^q.
\end{equation*}
Define the cross-sectionally centered matrices
\begin{equation*}
U_c := U-\bar u\,\ones[N]'\in\R^{T\times N},
\qquad
Z_c := Z-\bar z\,\ones[N]'\in\R^{q\times N}.
\end{equation*}
Substituting $\widehat\delta(F,A)$ back into the mean gives the centered representation
\begin{equation}\label{eq:centered_residual}
U-\widehat \delta(F,A)\,\ones[N]' - FAZ
=
U_c - FAZ_c.
\end{equation}
Equivalently, at the unit level, the likelihood depends on $(u_i-\bar u)$ and $(z_i-\bar z)$ rather than on $(u_i,z_i)$.

\paragraph{Step 2: concentrate out \texorpdfstring{$A$}{A} (given \texorpdfstring{$F$}{F})}

Using \eqref{eq:centered_residual}, the GLS (or QML) estimator of $A$ given $F$ solves
\begin{equation}\label{eq:Ahat_def}
\widehat A(F)
=
\arg\min_{A\in\R^{r\times q}}
\tr\!\left(M\,(U_c-FAZ_c)(U_c-FAZ_c)'\right).
\end{equation}
\paragraph{Normal equations.}
Differentiating the objective in \eqref{eq:Ahat_def} yields
\begin{equation*}
F' M (U_c-FAZ_c) Z_c' = 0,
\end{equation*}
so (whenever the indicated matrices are invertible)
\begin{equation*}
\widehat A(F)
=
\big(F' M F\big)^{-1}\,F' M U_c\,Z_c'\,\big(Z_cZ_c'\big)^{-1}.
\end{equation*}
Finally, with $\widehat A(F)$ inserted, the concentrated time effects are
\begin{equation*}
\widehat \delta(F)=\bar u - F\widehat A(F)\,\bar z.
\end{equation*}
Plugging $\widehat\delta(F)$ and $\widehat A(F)$ into \eqref{eq:lik} yields a working likelihood concentrated in
\begin{equation*}
(\rho,\phi,\beta,F,\Sigma_\eta,D_\varepsilon),
\end{equation*}
with $(\delta,A)$ eliminated analytically.

\section{Initializations for Estimation Algorithm}\label{sec:initialization}

The estimation algorithm depends on initial values for the inner loop, which can affect
convergence in non-convex problems. Here we provide one possible initialization procedure
for the inner loop.
Assume the initial condition $y_0=(y_{10},...,y_{N0})'$ is observed. Set the spatial parameter to
$\rho^{(0)}=0$ (so $B(\rho^{(0)})=I_N$). For each $t=1,...,T$, run the cross-sectional OLS regression
\begin{equation*}
y_t = \ones[N] \delta_t + \phi_t\,y_{t-1} + X_t \beta_t + \text{residual}_t,
\end{equation*}
where $y_{t-1}$ is observed for all $t$ (in particular, $y_{0}$ is observed). Obtain
$(\widehat\delta_t,\widehat\phi_t,\widehat\beta_t)$ and stack
\begin{equation*}
\widehat\vartheta_t := (\widehat\phi_t,\widehat\beta_t')' \in \mathbb R^{1+K},
\qquad t=1,...,T.
\end{equation*}
Set the initial dynamic and slope coefficients by time-averaging:
\begin{equation*}
\vartheta^{(0)} := \frac{1}{T}\sum_{t=1}^T \widehat\vartheta_t,
\qquad
\delta_t^{(0)} := \widehat\delta_t \ (t=1,...,T),
\end{equation*}
where $\vartheta^{(0)}=(\phi^{(0)},\beta^{(0)\prime})'$.

Form the $N\times T$ residual matrix (removing the dynamic and covariate part using $\vartheta^{(0)}$):
\begin{equation*}
R := \big(r_1,...,r_T\big),
\qquad
r_t := y_t - \ones[N] \delta_t^{(0)} - \phi^{(0)} y_{t-1} - X_t \beta^{(0)},
\qquad t=1,...,T.
\end{equation*}
Compute a thin singular value decomposition of the $N\times T$ residual matrix $R=UDV'$
and let $U_r$, $D_r$, and $V_r$ contain its first $r$ singular components.
The principal-component factor path and loadings can be initialized as
\begin{equation*}
	\widehat F^{\,pc}
	:=
	\sqrt T\,V_r,
	\qquad
	\widehat\Lambda^{\,pc}
	:=
	\frac1T R\widehat F^{\,pc}
	=
	\frac1{\sqrt T}U_rD_r.
\end{equation*}
This calculation is equivalent to an eigendecomposition of $RR'/T$
but does not construct the $N\times N$ matrix $RR'$.
Because $(F,\Lambda)$ are identified only up to rotation, enforce the normalization
(e.g.\ $F_{1:r,:}=I_r$) by a nonsingular matrix $H$:
\begin{equation*}
H := \big(\widehat F^{\,pc}_{1:r,:}\big)^{-1},\qquad
F^{(0)} := \widehat F^{\,pc}H,\qquad
\Lambda^{(0)} := \widehat\Lambda^{\,pc}H^{-\prime}.
\end{equation*}
Set $A^{(0)}=0$, so that the normalized loadings serve as the initial loading residuals, and set
\begin{equation*}
\Sigma_\eta^{(0)} := \frac{1}{N}\Lambda^{(0)\prime}\Lambda^{(0)}.
\end{equation*}
Finally, compute idiosyncratic residuals after removing the factor component,
\begin{equation*}
\varepsilon_t^{(0)} := r_t - \Lambda^{(0)}f_t^{(0)},
\qquad t=1,...,T,
\end{equation*}
where $f_t^{(0)\prime}$ is the $t$th row of $F^{(0)}$.
Stack $E^{(0)}:=(\varepsilon_1^{(0)},...,\varepsilon_T^{(0)})$ and set the diagonal idiosyncratic variance matrix to the cross-sectional second moments:
\begin{equation*}
D_\varepsilon^{(0)} := \diag\!\left(\frac{1}{N}E^{(0)\prime}E^{(0)}\right).
\end{equation*}
To avoid near-zero idiosyncratic variances, the diagonal entries of $D_\varepsilon^{(0)}$ may be truncated below by a small positive constant.

\section{Spatial Score Contributions and Covariance Estimator}\label{sec:spatial_score_construction}

In this section, the component-by-component construction of the spatially
corrected contributions $\{g_i(\alpha)\}_{i=1}^N$ used in
Section~\ref{sec:asymptotics} is derived,
following \textcite{li-yang-2021}.
Only the score contributions associated with the spatial parameter $\rho$ and
the dynamic parameter $\phi$ require rearrangement to form a
martingale-difference array. The remaining score contributions already have
the required unitwise form and are therefore left unchanged.
The full rearranged contribution is
\begin{equation}\label{eq:spatial_score_contribution}
	g_i(\alpha)
	=
	\begin{pmatrix}
		g_{\rho,i}(\alpha)\\
		g_{\phi,i}(\alpha)\\
		g_{\beta,i}(\alpha)\\
		g_{\delta,i}(\alpha)\\
		g_{A,i}(\alpha)\\
		g_{F_2,i}(\alpha)\\
		g_{\vech(\Sigma_\eta),i}(\alpha)\\
		g_{\sigma^2,i}(\alpha)
	\end{pmatrix}
	=
	\begin{pmatrix}
		g_{\rho,i}(\alpha)\\
		g_{\phi,i}(\alpha)\\
		s_{\beta,i}(\alpha)\\
		s_{\delta,i}(\alpha)\\
		s_{A,i}(\alpha)\\
		s_{F_2,i}(\alpha)\\
		s_{\vech(\Sigma_\eta),i}(\alpha)\\
		s_{\sigma^2,i}(\alpha)
	\end{pmatrix}.
\end{equation}
The following subsections derive the two rearranged components and record the
remaining ordinary score contributions.

\subsection{Basic and Reduced-Form Objects}
\label{subsec:score-basic-reduced-form}

For each unit $i$, let
\begin{equation*}
y_i=(y_{i1},...,y_{iT})',
\qquad
y_{i,-1}=(y_{i0},y_{i1},...,y_{i,T-1})',
\end{equation*}
\begin{equation*}
X_i=
\begin{pmatrix}
x_{i1}'\\
\vdots\\
x_{iT}'
\end{pmatrix},
\qquad
\mathbf e_1=(1,0,...,0)'.
\end{equation*}
The transformed innovation is
\begin{equation*}
u_i(\theta)
=
R_T(\phi)y_i
-
\rho\sum_{j=1}^N w_{ij}y_j
-
X_i\beta
-
\phi y_{i0}\mathbf e_1,
\qquad
\theta=(\rho,\phi,\beta')'.
\end{equation*}
The conditional mean is
\begin{equation*}
m_i(\alpha)
=
\delta+F A z_i,
\end{equation*}
and the transformed residual is
\begin{equation*}
e_i(\alpha)
=
u_i(\theta)-m_i(\alpha)
=
u_i(\theta)-\delta-F A z_i.
\end{equation*}
Define the centered covariance residual
\begin{equation*}
\mathcal R_i(\alpha)
=
M(\alpha)e_i(\alpha)e_i(\alpha)'M(\alpha)-M(\alpha).
\end{equation*}
For notational compactness below, write
\begin{equation*}
b_i(\alpha)=A z_i.
\end{equation*}

Define the $NT\times NT$ block matrix
\begin{equation*}
\mathcal C(\rho,\phi)
=
I_N\otimes R_T(\phi)-\rho W\otimes I_T,
\end{equation*}
and its inverse
\begin{equation*}
\mathcal G(\rho,\phi)
=
\mathcal C(\rho,\phi)^{-1}.
\end{equation*}
Let $\mathcal G_{ij}(\alpha)$ denote the $T\times T$ $(i,j)$ block of
$\mathcal G(\rho,\phi)$.

The inverse notation is analytical and does not prescribe a dense
$NT\times NT$ inversion.
After permuting the system into time-stacked order,
$\mathcal C(\rho,\phi)$ is block lower triangular with $B(\rho)$ on every diagonal block and $-\phi I_N$ on the first block subdiagonal.
Thus, for a right-hand side $h=(h_1',...,h_T')'$,
the solution $x=\mathcal G h$ is obtained recursively from
\begin{equation*}
	B(\rho)x_1=h_1,
	\qquad
	B(\rho)x_t=h_t+\phi x_{t-1},
	\quad t=2,...,T.
\end{equation*}
The implementation factors a sparse $B(\rho)$ once by LU decomposition and reuses those factors in this recursion.
For the exact corrected unit contributions,
it propagates the realized innovation responses through time using a small number of $N\times N$ work matrices;
it never stores the full $NT\times NT$ matrix $\mathcal G$.

Define the conditional-mean part of the reduced-form outcome by
\begin{equation*}
\bar y(\alpha)
=
\mathcal G(\rho,\phi)
\begin{pmatrix}
X_1\beta+\phi y_{10}\mathbf e_1+\delta+F A z_1\\
\vdots\\
X_N\beta+\phi y_{N0}\mathbf e_1+\delta+F A z_N
\end{pmatrix}.
\end{equation*}
Let $\bar y_i(\alpha)$ be the $i$-th $T\times 1$ block of $\bar y(\alpha)$. Then
\begin{equation*}
y_i
=
\bar y_i(\alpha)+\sum_{j=1}^N \mathcal G_{ij}(\alpha)e_j(\alpha).
\end{equation*}

For the spatial score, define
\begin{equation*}
\bar w_i(\alpha)
=
\sum_{k=1}^N w_{ik}\bar y_k(\alpha),
\end{equation*}
and
\begin{equation*}
\mathcal W_{ij}(\alpha)
=
\sum_{k=1}^N w_{ik}\mathcal G_{kj}(\alpha).
\end{equation*}
Then
\begin{equation*}
(Wy)_i
=
\bar w_i(\alpha)+\sum_{j=1}^N \mathcal W_{ij}(\alpha)e_j(\alpha).
\end{equation*}

For the dynamic score, define the lag-selection matrix
\begin{equation*}
L_T=
\begin{pmatrix}
0&0&0&\cdots&0\\
1&0&0&\cdots&0\\
0&1&0&\cdots&0\\
\vdots&\vdots&\ddots&\ddots&0\\
0&0&\cdots&1&0
\end{pmatrix}.
\end{equation*}
Then
\begin{equation*}
y_{i,-1}=y_{i0}\mathbf e_1+L_Ty_i.
\end{equation*}
Define
\begin{equation*}
\bar y_{i,-1}(\alpha)
=
y_{i0}\mathbf e_1+L_T\bar y_i(\alpha),
\end{equation*}
and
\begin{equation*}
\mathcal L_{ij}(\alpha)
=
L_T\mathcal G_{ij}(\alpha).
\end{equation*}
Then
\begin{equation*}
y_{i,-1}
=
\bar y_{i,-1}(\alpha)+\sum_{j=1}^N \mathcal L_{ij}(\alpha)e_j(\alpha).
\end{equation*}

\subsection{Contributions for the Parameters of Interest}

\subsubsection{Spatial Autoregressive Parameter \texorpdfstring{$\rho$}{rho}}

The ordinary unit score for $\rho$ is
\begin{equation}
	s_{\rho,i}(\alpha)
	=
	-\frac{T}{N}\tr\{B(\rho)^{-1}W\}
	+(Wy)_i'M(\alpha)e_i(\alpha).
	\label{eq:raw_score_rho}
\end{equation}
Substituting the reduced form for $(Wy)_i$ into
\eqref{eq:raw_score_rho}, summing over $i$, and pairing the $(i,j)$ and $(j,i)$
cross-products gives
\begin{equation*}
\begin{aligned}
\sum_{i=1}^N s_{\rho,i}(\alpha)
={}&
\sum_{i=1}^N e_i'M\bar w_i
+\sum_{i=1}^N
\left[
e_i'M\mathcal W_{ii}e_i
-\tr\!\left\{M\mathcal W_{ii}\Sigma_u\right\}
\right]\\
&+\sum_{i=1}^N\sum_{j<i}
e_i'\left[M\mathcal W_{ij}+\mathcal W_{ji}'M\right]e_j.
\end{aligned}
\end{equation*}
Here and in the next display, all objects are evaluated at $\alpha$. The trace
terms replace the spatial Jacobian because
\begin{equation*}
	\sum_{i=1}^N
	\tr\!\left\{M\mathcal W_{ii}\Sigma_u\right\}
	=
	T\tr\{B(\rho)^{-1}W\}.
\end{equation*}
The rearrangement assigns each off-diagonal pair to its larger unit index.
Accordingly, define
\begin{equation*}
\begin{aligned}
g_{\rho,i}(\alpha)
={}&
e_i(\alpha)'M(\alpha)\bar w_i(\alpha) \\
&+
\left[
e_i(\alpha)'M(\alpha)\mathcal W_{ii}(\alpha)e_i(\alpha)
-
\tr\!\left\{M(\alpha)\mathcal W_{ii}(\alpha)\Sigma_u(\alpha)\right\}
\right] \\
&+
\sum_{j<i}
e_i(\alpha)'
\left[
M(\alpha)\mathcal W_{ij}(\alpha)
+
\mathcal W_{ji}(\alpha)'M(\alpha)
\right]
e_j(\alpha).
\end{aligned}
\end{equation*}
By construction,
\begin{equation*}
\sum_{i=1}^N g_{\rho,i}(\alpha)
=
-T\tr\{B(\rho)^{-1}W\}
+
\sum_{i=1}^N (Wy)_i'M(\alpha)e_i(\alpha)
=\sum_{i=1}^N s_{\rho,i}(\alpha).
\end{equation*}
At $\alpha_0$, the diagonal trace subtraction centers the quadratic term and
the rearranged contribution is conditionally mean zero.

\subsubsection{Dynamic Parameter \texorpdfstring{$\phi$}{phi}}

The ordinary unit score for $\phi$ is
\begin{equation}
	s_{\phi,i}(\alpha)
	=
	y_{i,-1}'M(\alpha)e_i(\alpha).
	\label{eq:raw_score_phi}
\end{equation}
Likewise, substituting the reduced form for $y_{i,-1}$ into
\eqref{eq:raw_score_phi} and pairing the off-diagonal terms yields
\begin{equation*}
\begin{aligned}
\sum_{i=1}^N s_{\phi,i}(\alpha)
={}&
\sum_{i=1}^N e_i'M\bar y_{i,-1}
+\sum_{i=1}^N
\left[
e_i'M\mathcal L_{ii}e_i
-\tr\!\left\{M\mathcal L_{ii}\Sigma_u\right\}
\right]\\
&+\sum_{i=1}^N\sum_{j<i}
e_i'\left[M\mathcal L_{ij}+\mathcal L_{ji}'M\right]e_j.
\end{aligned}
\end{equation*}
The temporal transformation has determinant one, which implies the trace
identity
\begin{equation*}
	\sum_{i=1}^N
	\tr\!\left\{M\mathcal L_{ii}\Sigma_u\right\}=0.
\end{equation*}
Assigning each off-diagonal pair to the larger index gives
\begin{equation*}
\begin{aligned}
g_{\phi,i}(\alpha)
={}&
e_i(\alpha)'M(\alpha)\bar y_{i,-1}(\alpha) \\
&+
\left[
e_i(\alpha)'M(\alpha)\mathcal L_{ii}(\alpha)e_i(\alpha)
-
\tr\!\left\{M(\alpha)\mathcal L_{ii}(\alpha)\Sigma_u(\alpha)\right\}
\right] \\
&+
\sum_{j<i}
e_i(\alpha)'
\left[
M(\alpha)\mathcal L_{ij}(\alpha)
+
\mathcal L_{ji}(\alpha)'M(\alpha)
\right]
e_j(\alpha).
\end{aligned}
\end{equation*}
Therefore,
\begin{equation*}
\sum_{i=1}^N g_{\phi,i}(\alpha)
=
\sum_{i=1}^N y_{i,-1}'M(\alpha)e_i(\alpha)
=\sum_{i=1}^N s_{\phi,i}(\alpha).
\end{equation*}

\subsubsection{Slope Vector \texorpdfstring{$\beta$}{beta}}

No rearrangement is needed for the slope vector $\beta\in\R^K$; its corrected
contribution coincides with its ordinary unit score:
\begin{equation*}
s_{\beta,i}(\alpha)=g_{\beta,i}(\alpha)
=
X_i'M(\alpha)e_i(\alpha).
\end{equation*}
This is a $K\times 1$ vector.

\subsection{Nuisance-Parameter Contributions}

Every nuisance-parameter contribution already depends only on unit $i$'s
structural residual and observed controls. Hence no spatial rearrangement is
needed: each component below satisfies $g_{k,i}(\alpha)=s_{k,i}(\alpha)$.

\subsubsection{Time Effects \texorpdfstring{$\delta$}{delta}}

The contribution for the time-effect vector $\delta\in\R^T$ is
\begin{equation*}
s_{\delta,i}(\alpha)=g_{\delta,i}(\alpha)
=
M(\alpha)e_i(\alpha).
\end{equation*}
This is a $T\times 1$ vector.

\subsubsection{Projection Matrix \texorpdfstring{$A$}{A}}

The ordinary matrix-valued score for the SERL coefficient matrix
$A\in\R^{r\times q}$, which is also its corrected contribution, is
\begin{equation*}
G_{A,i}(\alpha)
=
F'M(\alpha)e_i(\alpha)z_i'.
\end{equation*}
For the vectorized parameter $\vecop(A)$, use
\begin{equation*}
s_{A,i}(\alpha)=g_{A,i}(\alpha)
=
\vecop\!\left(F'M(\alpha)e_i(\alpha)z_i'\right).
\end{equation*}
Equivalently, for the individual element $A_{ab}$,
\begin{equation*}
s_{A_{ab},i}(\alpha)=g_{A_{ab},i}(\alpha)
=
\left[F'M(\alpha)e_i(\alpha)\right]_a z_{ib}.
\end{equation*}

\subsubsection{Factor Paths \texorpdfstring{$F$}{F}}

The factor path $F$ enters both the conditional mean and the conditional
covariance. The unconstrained ordinary matrix-valued score, which is also its
corrected contribution, is
\begin{equation*}
G_{F,i}(\alpha)
=
M(\alpha)e_i(\alpha)z_i'A'
+
\mathcal R_i(\alpha)F\Sigma_\eta.
\end{equation*}
Equivalently,
\begin{equation*}
G_{F,i}(\alpha)
=
M(\alpha)e_i(\alpha)b_i(\alpha)'
+
\left[
M(\alpha)e_i(\alpha)e_i(\alpha)'M(\alpha)-M(\alpha)
\right]F\Sigma_\eta.
\end{equation*}
Element by element, for $F_{ta}$,
\begin{equation*}
s_{F_{ta},i}(\alpha)=g_{F_{ta},i}(\alpha)
=
\left[
M(\alpha)e_i(\alpha)z_i'A'
\right]_{ta}
+
\left[
\mathcal R_i(\alpha)F\Sigma_\eta
\right]_{ta}.
\end{equation*}

Since we use the normalization
\begin{equation*}
F=
\begin{pmatrix}
I_r\\
F_2
\end{pmatrix},
\end{equation*}
only the bottom $T-r$ rows are free. Therefore the contribution for
$\vecop(F_2)$ is
\begin{equation*}
s_{F_2,i}(\alpha)=g_{F_2,i}(\alpha)
=
\vecop\!\left(
\left[
M(\alpha)e_i(\alpha)z_i'A'
+
\mathcal R_i(\alpha)F\Sigma_\eta
\right]_{t=r+1,...,T}
\right).
\end{equation*}

\subsubsection{Factor-Loading Covariance \texorpdfstring{$\Sigma_\eta$}{Sigma eta}}

Let $\Sigma_\eta=(\sigma_{\eta,ab})_{a,b=1}^r$. Define
\begin{equation*}
S_{\eta,i}(\alpha)
=
F'\mathcal R_i(\alpha)F.
\end{equation*}
For diagonal entries,
\begin{equation*}
s_{\sigma_{\eta,aa},i}(\alpha)=g_{\sigma_{\eta,aa},i}(\alpha)
=
\frac12\left[S_{\eta,i}(\alpha)\right]_{aa}.
\end{equation*}
For off-diagonal entries $a<b$, since
$\sigma_{\eta,ab}=\sigma_{\eta,ba}$,
\begin{equation*}
s_{\sigma_{\eta,ab},i}(\alpha)=g_{\sigma_{\eta,ab},i}(\alpha)
=
\left[S_{\eta,i}(\alpha)\right]_{ab}.
\end{equation*}
Equivalently, if $\vech(\Sigma_\eta)$ stacks the unique elements of $\Sigma_\eta$, define
\begin{equation*}
s_{\vech(\Sigma_\eta),i}(\alpha)
=g_{\vech(\Sigma_\eta),i}(\alpha)
=
\vech_*\!\left(S_{\eta,i}(\alpha)\right),
\end{equation*}
where $\vech_*(H)$ stacks one half of the diagonal elements of $H$ and the full
off-diagonal elements of $H$, using the same triangular ordering as $\vech(\cdot)$.

\subsubsection{Idiosyncratic Variances \texorpdfstring{$D_\varepsilon$}{D epsilon}}

Let
\begin{equation*}
D_\varepsilon=\diag(\sigma_1^2,...,\sigma_T^2).
\end{equation*}
For each variance component \(\sigma_t^2\),
\begin{equation*}
s_{\sigma_t^2,i}(\alpha)=g_{\sigma_t^2,i}(\alpha)
=
\frac12
\left[
\mathcal R_i(\alpha)
\right]_{tt}.
\end{equation*}
Stacking over $t=1,...,T$,
\begin{equation*}
s_{\sigma^2,i}(\alpha)=g_{\sigma^2,i}(\alpha)
=
\frac12\diag\!\left(\mathcal R_i(\alpha)\right),
\qquad
\sigma^2=(\sigma_1^2,...,\sigma_T^2)'.
\end{equation*}
If one instead parameterizes by log variances \(\varsigma_t=\log \sigma_t^2\), then
\begin{equation*}
s_{\varsigma_t,i}(\alpha)=g_{\varsigma_t,i}(\alpha)
=
\frac12 \sigma_t^2\left[\mathcal R_i(\alpha)\right]_{tt}.
\end{equation*}

\subsection{Feasible Covariance Estimator}

Collecting the components gives the full \eqref{eq:spatial_score_contribution}.
Under the maintained conditional moment restrictions, the constructed
$\{g_i(\alpha_0)\}$ form a martingale-difference array with respect to the
filtration $\mathcal F_i:=\mathcal C_N\vee\sigma(e_1,...,e_i)$.

To exploit the orthogonality within the two rearranged components,
define the pairwise past-residual terms
\begin{equation*}
\begin{aligned}
	r_{\rho,ij}(\alpha)
	&:=
	e_i(\alpha)'
	\left[
	M(\alpha)\mathcal W_{ij}(\alpha)
	+
	\mathcal W_{ji}(\alpha)'M(\alpha)
	\right]
	e_j(\alpha),\\
	r_{\phi,ij}(\alpha)
	&:=
	e_i(\alpha)'
	\left[
	M(\alpha)\mathcal L_{ij}(\alpha)
	+
	\mathcal L_{ji}(\alpha)'M(\alpha)
	\right]
	e_j(\alpha),
	\qquad j<i.
\end{aligned}
\end{equation*}
Let
\begin{equation*}
	g_{\rho,i}^\dagger(\alpha)
	:=g_{\rho,i}(\alpha)-\sum_{j<i}r_{\rho,ij}(\alpha),
	\qquad
	g_{\phi,i}^\dagger(\alpha)
	:=g_{\phi,i}(\alpha)-\sum_{j<i}r_{\phi,ij}(\alpha),
\end{equation*}
and let $g_i^\dagger(\alpha)=g_i(\alpha)$ in all components other than those corresponding to $\rho$ and $\phi$.
Define the conformable vector
\begin{equation*}
	v_{ij}(\alpha)
	:=
	\begin{pmatrix}
		r_{\rho,ij}(\alpha)\\
		r_{\phi,ij}(\alpha)\\
		\mathbf 0
	\end{pmatrix},
\end{equation*}
where the zero block corresponds to $\beta$ and all nuisance parameters.
Then
\begin{equation}\label{eq:g_martingale_pair_decomposition}
	g_i(\alpha)
	=
	g_i^\dagger(\alpha)+\sum_{j<i}v_{ij}(\alpha).
\end{equation}

At $\alpha_0$,
$g_i^\dagger$ depends only on $e_i$,
while $v_{ij}$ is bilinear in $e_i$ and $e_j$.
Conditional independence across units and $\E_N(e_j)=0$ therefore give
\begin{equation}\label{eq:martingale_pair_orthogonality}
	\E_N\!\left[g_i^\dagger(\alpha_0)v_{ij}(\alpha_0)'\right]=0,
	\qquad
	\E_N\!\left[v_{ij}(\alpha_0)v_{ik}(\alpha_0)'\right]=0
	\quad\text{for }j\ne k.
\end{equation}
These equalities hold conditional on $\mathcal C_N$;
the corresponding cross-products need not vanish unit by unit conditional on $\mathcal F_{i-1}$.
Expanding \eqref{eq:g_martingale_pair_decomposition} and applying
\eqref{eq:martingale_pair_orthogonality} yields
\begin{equation}\label{eq:martingale_pair_covariance_identity}
	\E_N\!\left[g_i(\alpha_0)g_i(\alpha_0)'\right]
	=
	\E_N\!\left[
		g_i^\dagger(\alpha_0)g_i^\dagger(\alpha_0)'
		+
		\sum_{j<i}v_{ij}(\alpha_0)v_{ij}(\alpha_0)'
	\right].
\end{equation}
Thus, for any $\alpha$, define
\begin{equation}\label{eq:martingale_covariance_estimator}
	\widehat\Omega_M(\alpha)
	:=
	\underbrace{
	\frac1N\sum_{i=1}^N
		g_i^\dagger(\alpha)g_i^\dagger(\alpha)'
	}_{\widehat\Omega_\dagger(\alpha)}
	+
	\underbrace{
	\frac1N\sum_{i=1}^N
		\sum_{j<i}v_{ij}(\alpha)v_{ij}(\alpha)'
	}_{\widehat\Omega_v(\alpha)}.
\end{equation}
The second term replaces the square of each sum of past-unit contributions by the sum of their pairwise outer products.
Both terms in \eqref{eq:martingale_covariance_estimator} are positive semidefinite.
The estimator therefore removes finite-sample cross-products known to have zero conditional expectation without changing the population covariance target.
Proposition~\ref{prop:feasible_score_covariance} establishes its consistency and its asymptotic equivalence to the full outer product.
Since we use the following triangle inequality repeatedly,
we state it here for easy reference:
\begin{equation}\label{eq:p-power-triangle-inequality}
	\|x+y\|^p
	\leq
	2^{p-1}\bigl(\|x\|^p+\|y\|^p\bigr)
	\quad\text{for}\quad
	p\geq1.
\end{equation}

\begin{lemma}[Moment bounds for rearranged scores]\label{lem:score-moment-bounds}
	Under Assumptions~\ref{ass:varepsilon}, \ref{ass:SERL}, and \ref{ass:regularity},
	there exists $\delta>0$ such that
	\begin{align}\label{eq:score_moment_bounds}
		\frac1N\sum_{i=1}^N
		\E_N\sup_{\alpha\in\mathcal A_0}\|g_i(\alpha)\|^{2+\delta}
		=O_p(1),
		\qquad
		\frac1N\sum_{i=1}^N
		\E_N\!\left[
		\sup_{\alpha\in\mathcal A_0}
		\|\nabla_\alpha g_i(\alpha)\|^{2+\delta}
		\right]
		=O_p(1).
	\end{align}
\end{lemma}
\begin{proof}
	The proof first controls residuals evaluated near $\alpha_0$.
	It then applies those bounds to the ordinary and rearranged score contributions.
	Finally,
	it differentiates the same representations to control the score derivatives.

	\medskip
	\noindent\textit{Residual moment bounds.}
	Choose $0<\delta\leq\min\{\zeta/2,2\}$ so that $4+2\delta\leq4+\zeta$.
	Let
	\begin{equation*}
		e^0:=\left(e_1(\alpha_0)',\ldots,e_N(\alpha_0)'\right)'.
	\end{equation*}
	Writing $h(\alpha)$ for the unit-stacked conditional mean in the transformed
	model gives
	\begin{equation*}
		\mathcal C(\theta_0)y=h(\alpha_0)+e^0,
		\qquad
		y=\mathcal G(\theta_0)\big(h(\alpha_0)+e^0\big).
	\end{equation*}
	Consequently,
	\begin{equation}\label{eq:step1-s1}
		e(\alpha)
		=\mathcal C(\theta)y-h(\alpha)
		=\mathcal C(\theta)\mathcal G(\theta_0)e^0
		+\mathcal C(\theta)\mathcal G(\theta_0)h(\alpha_0)-h(\alpha).
	\end{equation}
	Let $d_i(\alpha)$ be the $i$th block of the conditionally nonrandom second
	term in \eqref{eq:step1-s1}.
	Then
	\begin{equation}\label{eq:step1-s2}
		e_i(\alpha)
		=
		d_i(\alpha)
		+\sum_{j=1}^N
		\left[\mathcal C(\theta)\mathcal G(\theta_0)\right]_{ij}
		e_j(\alpha_0).
	\end{equation}
	Assumption~\ref{ass:regularity} implies, for a finite constant $C$, that
	\begin{equation}\label{eq:step1-s3}
		\begin{aligned}
			\max_i\sum_{j=1}^N
			\sup_{\alpha\in\mathcal A_0}
			\left\{
			\left\|
			\left[\mathcal C(\theta)\mathcal G(\theta_0)\right]_{ij}
			\right\|
			+
			\left\|
			\nabla_\alpha
			\left[\mathcal C(\theta)\mathcal G(\theta_0)\right]_{ij}
			\right\|
			\right\}
			&\leq C, \\
			\max_j\sum_{i=1}^N
			\sup_{\alpha\in\mathcal A_0}
			\left\{
			\left\|
			\left[\mathcal C(\theta)\mathcal G(\theta_0)\right]_{ij}
			\right\|
			+
			\left\|
			\nabla_\alpha
			\left[\mathcal C(\theta)\mathcal G(\theta_0)\right]_{ij}
			\right\|
			\right\}
			&\leq C,
		\end{aligned}
	\end{equation}
	and
	\begin{equation}\label{eq:step1-s4}
		\frac1N\sum_{i=1}^N
		\sup_{\alpha\in\mathcal A_0}
		\left\{
		\|d_i(\alpha)\|
		+\|\nabla_\alpha d_i(\alpha)\|
		\right\}^{4+\zeta}
		=O_p(1).
	\end{equation}
	We bound the blocks and their derivatives together for notational convenience.
	By \eqref{eq:conditional_moments} and
	$4+2\delta\leq4+\zeta$,
	\begin{equation*}
		\sup_{N\geq1}\max_{1\leq i\leq N}
		\E_N\|e_i(\alpha_0)\|^{4+2\delta}<\infty.
	\end{equation*}
	Minkowski's inequality and \eqref{eq:step1-s3} give
	\begin{equation}\label{eq:step1-s6}
		\begin{aligned}
			&\left(
			\E_N\sup_{\alpha\in\mathcal A_0}
			\left\|
			\sum_{j=1}^N
			\left[\mathcal C(\theta)\mathcal G(\theta_0)\right]_{ij}
			e_j(\alpha_0)
			\right\|^{4+2\delta}
			\right)^{1/(4+2\delta)} \\
			&\qquad\leq
			\sum_{j=1}^N
			\sup_{\alpha\in\mathcal A_0}
			\left\|
			\left[\mathcal C(\theta)\mathcal G(\theta_0)\right]_{ij}
			\right\|
			\left\{
			\E_N\|e_j(\alpha_0)\|^{4+2\delta}
			\right\}^{1/(4+2\delta)}
			\leq C.
		\end{aligned}
	\end{equation}
	Apply \eqref{eq:p-power-triangle-inequality} with $p=4+2\delta$ to \eqref{eq:step1-s2}.
	Taking $\E_N\sup_{\alpha\in\mathcal A_0}$,
	using \eqref{eq:step1-s4} and \eqref{eq:step1-s6},
	and then averaging over $i$ gives
	\begin{equation}\label{eq:step1-s7}
		\frac1N\sum_{i=1}^N
		\E_N\sup_{\alpha\in\mathcal A_0}
		\|e_i(\alpha)\|^{4+2\delta}
		=O_p(1).
	\end{equation}

	\medskip
	\noindent\textit{Score-contribution moment bounds.}
	The ordinary contributions for $\beta$ and the nuisance parameters contain
	only linear and centered quadratic functions of $e_i(\alpha)$.
	Hence, \eqref{eq:step1-s7} and the average-moment
	bounds on the conditionally nonrandom score coefficients show that every
	ordinary score component has a bounded averaged $2+\delta$ moment.

	It remains to control the off-diagonal terms in $g_{\rho,i}$ and
	$g_{\phi,i}$.
	Both have the form
	\begin{equation}\label{eq:step1-s10}
		e_i(\alpha)'
		\sum_{j<i}
		\left[
		M(\alpha)\mathcal D_{ij}(\alpha)
		+\mathcal D_{ji}(\alpha)'M(\alpha)
		\right]
		e_j(\alpha),
	\end{equation}
	where $\mathcal D=\mathcal W$ for the spatial contribution and
	$\mathcal D=\mathcal L$ for the dynamic contribution.
	We first bound the summation term of \eqref{eq:step1-s10}.
	Define
	\begin{equation*}
		\overline{\mathcal D}_{ij}
		:=
		\sup_{\alpha\in\mathcal A_0}
		\left\{
		\|\mathcal D_{ij}(\alpha)\|
		+\|\nabla_\alpha\mathcal D_{ij}(\alpha)\|
		\right\}.
	\end{equation*}
	As in \eqref{eq:step1-s3},
	the envelope includes both the block and its derivative so that the same summability bounds can be used for the score contributions and their derivatives.
	Assumption~\ref{ass:regularity} gives
	\begin{equation}\label{eq:step1-s11}
		\max_i\sum_{j=1}^N\overline{\mathcal D}_{ij}
		\leq C,
		\qquad
		\max_j\sum_{i=1}^N\overline{\mathcal D}_{ij}
		\leq C.
	\end{equation}
	Because $M(\alpha)$ is uniformly bounded, we have
	\begin{equation}\label{eq:step1-s12}
		\sup_{\alpha\in\mathcal A_0}
		\left\|
		M(\alpha)\mathcal D_{ij}(\alpha)
		+\mathcal D_{ji}(\alpha)'M(\alpha)
		\right\|
		\leq
		C\left(
		\overline{\mathcal D}_{ij}
		+\overline{\mathcal D}_{ji}
		\right).
	\end{equation}
	Applying Jensen inequality, \eqref{eq:step1-s12},
	and both summability bounds in \eqref{eq:step1-s11} gives
	\begin{equation}\label{eq:step1-s14}
		\begin{aligned}
			&\frac1N\sum_{i=1}^N
			\E_N\sup_{\alpha\in\mathcal A_0}
			\left\|
			\sum_{j<i}
			\left[
			M(\alpha)\mathcal D_{ij}(\alpha)
			+\mathcal D_{ji}(\alpha)'M(\alpha)
			\right]
			e_j(\alpha)
			\right\|^{4+2\delta} \\
			&\qquad\leq
			\frac{C}{N}
			\sum_{i=1}^N\sum_{j<i}
			\left(
			\overline{\mathcal D}_{ij}
			+\overline{\mathcal D}_{ji}
			\right)
			\E_N\sup_{\alpha\in\mathcal A_0}
			\|e_j(\alpha)\|^{4+2\delta} \\
			&\qquad\leq
			\frac{C}{N}\sum_{j=1}^N
			\E_N\sup_{\alpha\in\mathcal A_0}
			\|e_j(\alpha)\|^{4+2\delta}
			=O_p(1).
		\end{aligned}
	\end{equation}
	Now reconsider \eqref{eq:step1-s10}.
	For each $i$,
	Cauchy-Schwarz yields
	\begin{equation*}
		\begin{aligned}
			&\E_N\sup_{\alpha\in\mathcal A_0}
			\left|e_i(\alpha)'
			\sum_{j<i}
			\left[
			M(\alpha)\mathcal D_{ij}(\alpha)
			+\mathcal D_{ji}(\alpha)'M(\alpha)
			\right]
			e_j(\alpha)
			\right|^{2+\delta} \\
			&\qquad\leq
			\left(
			\E_N\sup_{\alpha\in\mathcal A_0}
			\|e_i(\alpha)\|^{4+2\delta}
			\right)^{1/2}
			\left(
			\E_N\sup_{\alpha\in\mathcal A_0}
			\left\|\sum_{j<i}
			\left[
			M(\alpha)\mathcal D_{ij}(\alpha)
			+\mathcal D_{ji}(\alpha)'M(\alpha)
			\right]
			e_j(\alpha)\right\|^{4+2\delta}
			\right)^{1/2}.
		\end{aligned}
	\end{equation*}
	Averaging the preceding inequality over $i$,
	applying Cauchy-Schwarz across $i$,
	and using \eqref{eq:step1-s7} and \eqref{eq:step1-s14} gives
	\begin{equation}\label{eq:step1-s15}
		\begin{aligned}
			&\frac1N\sum_{i=1}^N
			\E_N\sup_{\alpha\in\mathcal A_0}
			\left|e_i(\alpha)'
			\sum_{j<i}
			\left[
			M(\alpha)\mathcal D_{ij}(\alpha)
			+\mathcal D_{ji}(\alpha)'M(\alpha)
			\right]
			e_j(\alpha)
			\right|^{2+\delta}
			=O_p(1).
		\end{aligned}
	\end{equation}
	The linear and diagonal centered terms in $g_{\rho,i}$ and $g_{\phi,i}$ are
	controlled by \eqref{eq:step1-s7}
	and the coefficient bounds in Assumption~\ref{ass:regularity},
	in the same way as the ordinary scores.
	This completes the first moment bound in \eqref{eq:score_moment_bounds}.

	\medskip
	\noindent\textit{Derivative moment bounds.}
	It remains to bound the derivatives.
	For scalar parameter $a$,
	the relevant identities are
	\begin{equation}\label{eq:step1-s17}
		\partial_a\mathcal G
		=-\mathcal G(\partial_a\mathcal C)\mathcal G,
		\qquad
		\partial_aM
		=-M(\partial_a\Sigma_u)M,
	\end{equation}
	and
	\begin{equation*}
		\partial_a(x'Qy)
		=(\partial_a x)'Qy
		+x'(\partial_a Q)y
		+x'Q(\partial_a y).
	\end{equation*}
	Differentiating \eqref{eq:step1-s1} also gives
	\begin{equation}\label{eq:step1-s19}
		\partial_a e(\alpha)
		=
		(\partial_a\mathcal C(\theta))\mathcal G(\theta_0)e^0
		+(\partial_a\mathcal C(\theta))\mathcal G(\theta_0)h(\alpha_0)
		-\partial_a h(\alpha).
	\end{equation}
	The derivative envelopes in \eqref{eq:step1-s3} and \eqref{eq:step1-s11}, together with
	\eqref{eq:step1-s17}--\eqref{eq:step1-s19}, show that every component of
	$\nabla_\alpha g_i(\alpha)$ has the same weighted linear or quadratic form
	already bounded above.
	Products of block arrays satisfying the row and column summability bounds retain those bounds because each product sum is bounded by the product of the corresponding summability bounds.
	Repeating \eqref{eq:step1-s6}--\eqref{eq:step1-s15} therefore establishes the second moment bound in \eqref{eq:score_moment_bounds}.
\end{proof}

\begin{proposition}[Consistency of the feasible score covariance] \label{prop:feasible_score_covariance}
	Under Assumptions~\ref{ass:varepsilon}--\ref{ass:score_clt},
	\begin{equation}\label{eq:score_covariance_convergence}
		\frac1N\sum_{i=1}^N
		g_i(\alpha_0)g_i(\alpha_0)'
		\pto\Omega(\alpha_0).
	\end{equation}
	Moreover,
	\begin{equation}\label{eq:martingale_covariance_equivalence}
		\widehat\Omega_M(\alpha_0)
		-
		\frac1N\sum_{i=1}^N
		g_i(\alpha_0)g_i(\alpha_0)'
		=o_p(1),
	\end{equation}
	so $\widehat\Omega_M(\alpha_0)\pto\Omega(\alpha_0)$.
	If $\widehat\alpha\pto\alpha_0$, then
	\begin{equation}\label{eq:score_covariance_feasible_convergence}
		\widehat\Omega_M(\widehat\alpha)
		-
		\widehat\Omega_M(\alpha_0)
		=o_p(1).
	\end{equation}
\end{proposition}
\begin{proof}
	\noindent\textit{Outer-product convergence at the truth.}
	Let $\delta>0$ be as in Lemma~\ref{lem:score-moment-bounds}.
	Let $A_i:=g_i(\alpha_0)g_i(\alpha_0)'$ and
	\begin{equation*}
		Z_i
		:=
		A_i
		-
		\E\!\left(
		A_i
		\mid\mathcal F_{i-1}
		\right ).
	\end{equation*}
	Because $g_i(\alpha_0)$ is $\mathcal F_i$-measurable,
	$Z_i$ is $\mathcal F_i$-measurable.
	Moreover,
	\begin{equation*}
		\begin{aligned}
			\E(Z_i\mid\mathcal F_{i-1})
			&=
			\E(A_i\mid\mathcal F_{i-1})
			-
			\E(
			\E(A_i\mid\mathcal F_{i-1})
			\mid
			\mathcal F_{i-1}
			)
			=0.
		\end{aligned}
	\end{equation*}
	Thus,
	after vectorization,
	$\{Z_i,\mathcal F_i\}$ is a martingale-difference array.
	Set $p=1+\delta/2\in(1,2]$ and
	apply \eqref{eq:p-power-triangle-inequality} to the definition of $Z_i$.
	Then, take $\E_N$ to obtain
	\begin{equation}\label{eq:Zi-convexity-bound}
		\E_N\|Z_i\|_{\mathrm F}^p
		\leq
		2^{p-1}
		\Big(
		\E_N\|A_i\|_{\mathrm F}^p
		+
		\E_N
		\left\|
		\E(A_i\mid\mathcal F_{i-1})
		\right\|_{\mathrm F}^p
		\Big)
	\end{equation}
	Consider the second term on the right-hand side of \eqref{eq:Zi-convexity-bound}.
	Apply Jensen inequality to function $x\mapsto\|x\|_{\mathrm F}^p$, then take $\E_N$ on both sides to obtain
	\begin{equation}\label{eq:conditional-Jensen-Ai}
		\E_N
		\left\|
		\E(A_i\mid\mathcal F_{i-1})
		\right\|_{\mathrm F}^p
		\leq
		\E_N
		\left[
		\E\!\left(
		\|A_i\|_{\mathrm F}^p
		\,\middle|\,
		\mathcal F_{i-1}
		\right)
		\right]
		=
		\E_N\|A_i\|_{\mathrm F}^p.
	\end{equation}
	Note that we have
	\begin{equation}\label{eq:rank-one-frobenius}
		\|A_i\|_{\mathrm F}
		=
		\|g_i(\alpha_0)g_i(\alpha_0)'\|_{\mathrm F}
		=
		\|g_i(\alpha_0)\|^2
	\end{equation}
	Hence, combining \eqref{eq:Zi-convexity-bound},
	\eqref{eq:conditional-Jensen-Ai},
	and \eqref{eq:rank-one-frobenius},
	and using $2p=2+\delta$,
	gives
	\begin{equation*}
		\E_N\|Z_i\|_{\mathrm F}^p
		\leq
		2^p\E_N\|A_i\|_{\mathrm F}^p
		=
		2^p\E_N\|g_i(\alpha_0)\|^{2+\delta}.
	\end{equation*}
	The first bound in \eqref{eq:score_moment_bounds} therefore implies
	\begin{equation*}
		\frac1N\sum_{i=1}^N \E_N\|Z_i\|_{\mathrm F}^p
		\leq 2^p \frac1N\sum_{i=1}^N \E_N \|g_i(\alpha_0)\|^{2+\delta}
		=O_p(1).
	\end{equation*}
	Applying the martingale von Bahr--Esseen inequality coordinate-wise
	\parencite{von-bahr-esseen-1965},
	with the fixed-dimensional constants absorbed into constant $C_p$,
	gives
	\begin{equation*}
		\begin{aligned}
			\E_N\!\left\|
			\frac1N\sum_{i=1}^N Z_i
			\right\|_{\mathrm F}^p
			&\leq
			\frac{C_p}{N^p}
			\sum_{i=1}^N\E_N\|Z_i\|_{\mathrm F}^p \\
			&=
			\frac{C_p}{N^{p-1}}
			\left(
			\frac1N\sum_{i=1}^N
			\E_N\|Z_i\|_{\mathrm F}^p
			\right)
			=O_p(N^{1-p})
			=o_p(1).
		\end{aligned}
	\end{equation*}
	For every $\epsilon>0$,
	Markov's inequality gives
	\begin{equation*}
		\Pr\!\left(
		\left\|
		\frac1N\sum_{i=1}^N Z_i
		\right\|_{\mathrm F}
		>\epsilon
		\,\middle|\,
		\mathcal C_N
		\right)
		\leq
		\epsilon^{-p}
		\E_N\!\left\|
		\frac1N\sum_{i=1}^N Z_i
		\right\|_{\mathrm F}^p
		=o_p(1).
	\end{equation*}
	Because the conditional probability is bounded by one,
	its convergence in probability to zero implies that its expectation also converges to zero.
	Hence, we have
	\begin{equation}\label{eq:g_outer_product_lln}
		\frac1N\sum_{i=1}^N Z_i
		= \frac1N\sum_{i=1}^N A_i - \frac1N\sum_{i=1}^N \E(A_i\mid\mathcal F_{i-1})
		= o_p(1).
	\end{equation}
	Combining \eqref{eq:g_outer_product_lln} with the predictable-variation convergence in Assumption~\ref{ass:score_clt} establishes \eqref{eq:score_covariance_convergence}.

	\medskip
	\noindent\textit{Martingale orthogonalization.}
	For compactness, write
	\begin{equation*}
		R_i:=\sum_{j<i}v_{ij}(\alpha_0).
	\end{equation*}
	By \eqref{eq:g_martingale_pair_decomposition},
	$g_i(\alpha_0)=g_i^\dagger(\alpha_0)+R_i$.
	Consequently,
	\begin{equation}\label{eq:martingale_opg_difference}
	\begin{aligned}
		&g_i(\alpha_0) g_i(\alpha_0)'
		-
		\left[
		g_i^\dagger(\alpha_0)g_i^\dagger(\alpha_0)'
		+
		\sum_{j<i}v_{ij}(\alpha_0)v_{ij}(\alpha_0)'
		\right]\\
		&\qquad=
		g_i^\dagger(\alpha_0)R_i'
		+R_i g_i^\dagger(\alpha_0)'
		+
		\sum_{j<i}\sum_{\substack{k<i\\k\ne j}}
		v_{ij}(\alpha_0)v_{ik}(\alpha_0)'.
	\end{aligned}
	\end{equation}
	Denote the matrix on either side of \eqref{eq:martingale_opg_difference} by $D_i$.
	Each term on the right side of \eqref{eq:martingale_opg_difference} contains either one past-unit residual appearing only once or two distinct past-unit residuals appearing once each.
	The conditional mean-zero and cross-unit independence restrictions therefore give zero after conditioning on $\mathcal C_N$,
	as recorded in \eqref{eq:martingale_pair_orthogonality}.

	We now show explicitly that $N^{-1}\sum_iD_i=o_p(1)$.
	First decompose
	\begin{equation*}
		\frac1N\sum_{i=1}^N D_i
		=
		\frac1N\sum_{i=1}^N
		\left\{D_i-\E(D_i\mid\mathcal F_{i-1})\right\}
		+
		\frac1N\sum_{i=1}^N
		\E(D_i\mid\mathcal F_{i-1}).
	\end{equation*}
	The first summand is an average of martingale differences.
	The triangle inequality gives
	\begin{equation*}
		\|D_i\|_{\mathrm F}
		\leq
		2\|g_i^\dagger(\alpha_0)\|
		\sum_{j<i}\|v_{ij}(\alpha_0)\|
		+
		\left(\sum_{j<i}\|v_{ij}(\alpha_0)\|\right)^2.
	\end{equation*}
	With $p=1+\delta/2$ as above,
	the residual moment bound in \eqref{eq:step1-s7},
	the block summability bound in \eqref{eq:step1-s11},
	and Cauchy--Schwarz therefore imply
	\begin{equation*}
		\frac1N\sum_{i=1}^N\E_N\|D_i\|_{\mathrm F}^p
		=O_p(1).
	\end{equation*}
	Conditional Jensen inequality,
	the martingale von Bahr--Esseen inequality,
	and Markov's inequality,
	applied exactly as in the outer-product argument above,
	then yield
	\begin{equation}\label{eq:martingale-orthogonalization-innovation}
		\frac1N\sum_{i=1}^N
		\left\{D_i-\E(D_i\mid\mathcal F_{i-1})\right\}
		=o_p(1).
	\end{equation}

	It remains to control the predictable summand.
	We work entry-by-entry and let $(m,n)$ index an entry of $D_i$.
	Since $g_i^\dagger(\alpha_0)$ is a linear or centered quadratic function of $e_i$ and $v_{ij}(\alpha_0)$ is bilinear in $e_i$ and $e_j$,
	integration over $e_i$ gives
	\begin{equation*}
		\E(D_{i,mn}\mid\mathcal F_{i-1})
		=
		\sum_{j<i}a_{ij,mn}'e_j
		+
		\sum_{j<i}\sum_{\substack{k<i\\k\ne j}}
		e_j'B_{ijk,mn}e_k,
	\end{equation*}
	where $a_{ij,mn}$ is a $T$-vector and $B_{ijk,mn}$ is a $T\times T$ matrix.
	They depend on the true parameters,
	the spatial-response blocks,
	the conditioning variables,
	and the conditional moments of $e_i$,
	but are fixed conditional on $\mathcal C_N$ and are not additional parameters or estimators.
	Define their averages over the current-unit index by
	\begin{equation*}
		\overline a_{j,mn}
		:=
		\frac1N\sum_{i>j}a_{ij,mn},
		\qquad
		\overline B_{jk,mn}
		:=
		\frac1N\sum_{i>\max\{j,k\}}B_{ijk,mn}.
	\end{equation*}
	Averaging over $i$ therefore gives
	\begin{equation*}
		\frac1N\sum_{i=1}^N
		\E(D_{i,mn}\mid\mathcal F_{i-1})
		=
		\sum_{j=1}^N\overline a_{j,mn}'e_j
		+
		\sum_{j=1}^N\sum_{\substack{k=1\\k\ne j}}^N
		e_j'\overline B_{jk,mn}e_k.
	\end{equation*}
	The row and column summability in \eqref{eq:step1-s11},
	the moment bounds on the conditionally nonrandom score coefficients,
	and weighted Cauchy--Schwarz imply
	\begin{equation*}
		\sum_{j=1}^N\|\overline a_{j,mn}\|^2
		+
		\sum_{j=1}^N\sum_{\substack{k=1\\k\ne j}}^N
		\|\overline B_{jk,mn}\|_{\mathrm F}^2
		=O_p(N^{-1}).
	\end{equation*}
	Conditional independence,
	$\E_N(e_j)=0$,
	and the uniformly bounded second and fourth moments of $e_j$ imply
	\begin{equation*}
		\E_N
		\left|
		\frac1N\sum_{i=1}^N
		\E(D_{i,mn}\mid\mathcal F_{i-1})
		\right|^2
		\leq
		C\left(
		\sum_{j=1}^N\|\overline a_{j,mn}\|^2
		+
		\sum_{j=1}^N\sum_{\substack{k=1\\k\ne j}}^N
		\|\overline B_{jk,mn}\|_{\mathrm F}^2
		\right)
		=O_p(N^{-1}).
	\end{equation*}
	Because the dimension of $D_i$ is fixed,
	the same bounds hold for the Frobenius norm after summing over $(m,n)$.
	Markov's inequality therefore gives
	\begin{equation}\label{eq:martingale-orthogonalization-predictable}
		\frac1N\sum_{i=1}^N
		\E(D_i\mid\mathcal F_{i-1})
		=o_p(1).
	\end{equation}
	Combining \eqref{eq:martingale-orthogonalization-innovation} and
	\eqref{eq:martingale-orthogonalization-predictable} gives
	\begin{equation*}
		\frac1N\sum_{i=1}^N D_i=o_p(1).
	\end{equation*}
	By the definition of $D_i$,
	this is equivalent to \eqref{eq:martingale_covariance_equivalence}.

	\medskip
	\noindent\textit{Plug-in replacement.}
	The residual and block-envelope bounds used in Lemma~\ref{lem:score-moment-bounds} apply separately to $g_i^\dagger(\alpha)$ and to the pairwise terms $v_{ij}(\alpha)$.
	Together with Jensen and Markov inequalities,
	they imply
	\begin{equation}\label{eq:martingale_covariance_derivative_bounds}
	\begin{aligned}
		&\frac1N\sum_{i=1}^N
		\sup_{\alpha\in\mathcal A_0}
		\|\nabla_\alpha g_i^\dagger(\alpha)\|^2
		=O_p(1),
		\qquad
		\frac1N\sum_{i=1}^N
		\|g_i^\dagger(\alpha_0)\|^2
		=O_p(1),\\
		&\frac1N\sum_{i=1}^N\sum_{j<i}
		\sup_{\alpha\in\mathcal A_0}
		\|\nabla_\alpha v_{ij}(\alpha)\|^2
		=O_p(1),
		\qquad
		\frac1N\sum_{i=1}^N\sum_{j<i}
		\|v_{ij}(\alpha_0)\|^2
		=O_p(1).
	\end{aligned}
	\end{equation}
	With probability approaching one, $\widehat\alpha\in\mathcal A_0$.
	The mean-value theorem,
	consistency of $\widehat\alpha$,
	and \eqref{eq:martingale_covariance_derivative_bounds} give
	\begin{equation}\label{eq:martingale_covariance_plugin_l2}
	\begin{aligned}
		\frac1N\sum_{i=1}^N
		\|g_i^\dagger(\widehat\alpha)-g_i^\dagger(\alpha_0)\|^2
		&=o_p(1),\\
		\frac1N\sum_{i=1}^N\sum_{j<i}
		\|v_{ij}(\widehat\alpha)-v_{ij}(\alpha_0)\|^2
		&=o_p(1).
	\end{aligned}
	\end{equation}
	Apply the identity
	\begin{equation*}
		aa'-bb'
	=
	(a-b)b'+b(a-b)'+(a-b)(a-b)'
	\end{equation*}
	first to $a=g_i^\dagger(\widehat\alpha)$ and $b=g_i^\dagger(\alpha_0)$,
	and then to every pair $a=v_{ij}(\widehat\alpha)$ and $b=v_{ij}(\alpha_0)$.
	The triangle and Cauchy--Schwarz inequalities,
	\eqref{eq:martingale_covariance_derivative_bounds},
	and \eqref{eq:martingale_covariance_plugin_l2} then yield
	\begin{equation*}
		\|\widehat\Omega_M(\widehat\alpha)-\widehat\Omega_M(\alpha_0)\|_{\mathrm F}
		=o_p(1).
	\end{equation*}
	This proves \eqref{eq:score_covariance_feasible_convergence}.
	\qedhere
\end{proof}

\section{Rearranged Scores for the Model Extensions}\label{sec:extension-scores}

The main likelihood and asymptotic argument for the lagged-spatial-outcome
extension appear in Section~\ref{subsec:extension-lagged-spatial-outcome}. This
appendix records the additional objects needed to implement its corrected
spatial covariance estimator. Let $L_T$ be the lag-selection matrix defined in
Appendix~\ref{sec:spatial_score_construction}. In unit-stacked order, define
\begin{equation*}
	\mathcal C_\xi(\rho,\phi,\xi)
	=
	I_N\otimes R_T(\phi)
	-\rho W\otimes I_T
	-\xi W\otimes L_T,
	\qquad
	\mathcal G_\xi=\mathcal C_\xi^{-1}.
\end{equation*}
The conditional-mean component of the reduced-form outcome is
\begin{equation*}
	\bar y_\xi
	=
	\mathcal G_\xi
	\begin{pmatrix}
		X_1\beta+\phi y_{10}\mathbf e_1+\xi(Wy_0)_1\mathbf e_1+\delta+FAz_1\\
		\vdots\\
		X_N\beta+\phi y_{N0}\mathbf e_1+\xi(Wy_0)_N\mathbf e_1+\delta+FAz_N
	\end{pmatrix}.
\end{equation*}
Accordingly,
$y_i=\bar y_{\xi,i}+\sum_j\mathcal G_{\xi,ij}e_{\xi,j}$. To construct the
spatial sandwich covariance estimator for
$\theta_\xi=(\rho,\phi,\xi,\beta')'$, it remains to derive the corrected score
contribution associated with $\xi$.
If $\mathcal G_{\xi,ij}$ is the $(i,j)$ block of $\mathcal G_\xi$, define
\begin{equation*}
	\mathcal K_{ij}
	:=
	\left[(W\otimes L_T)\mathcal G_\xi\right]_{ij}.
\end{equation*}
Writing $\bar y_{i,-1}$ for the conditional-mean path of $y_{i,-1}$, let
\begin{equation*}
	\bar w_{i,-1}:=\sum_{k=1}^N w_{ik}\bar y_{k,-1}.
\end{equation*}
The corrected contribution for the coefficient $\xi$ on $Wy_{t-1}$ is
\begin{equation*}
\begin{aligned}
g_{\xi,i}(\alpha)
={}&
e_i(\alpha)'M(\alpha)\bar w_{i,-1}(\alpha) \\
&+
\left[
e_i(\alpha)'M(\alpha)\mathcal K_{ii}(\alpha)e_i(\alpha)
-\tr\!\left\{M(\alpha)\mathcal K_{ii}(\alpha)\Sigma_u(\alpha)\right\}
\right] \\
&+
\sum_{j<i}
e_i(\alpha)'
\left[
M(\alpha)\mathcal K_{ij}(\alpha)
+\mathcal K_{ji}(\alpha)'M(\alpha)
\right]
e_j(\alpha).
\end{aligned}
\end{equation*}
The trace correction is zero under the causal time ordering, but retaining it
makes the parallel with the other dynamic score contributions explicit. Adding
$g_{\xi,i}$ to the structural block gives
\begin{equation*}
	\sum_{i=1}^N g_{\xi,i}(\alpha_0)
	=
	\sum_{i=1}^N (Wy_{-1})_i'M(\alpha_0)e_i(\alpha_0).
\end{equation*}
The corrected $\rho$ and $\phi$ contributions in
Appendix~\ref{sec:spatial_score_construction} must likewise be evaluated using
$\mathcal G_\xi$: their response blocks become
$[(W\otimes I_T)\mathcal G_\xi]_{ij}$ and
$[(I_N\otimes L_T)\mathcal G_\xi]_{ij}$, respectively. The nuisance
contributions and profiling calculation retain their form, evaluated at the
extended residual and with $\theta=\theta_\xi$.
For the martingale covariance estimator in \eqref{eq:martingale_covariance_estimator},
the pairwise vector $v_{ij}$ is augmented by the corresponding past-residual term from $g_{\xi,i}$.
The pairwise outer-product correction therefore occupies the $3\times3$ block for $(\rho,\phi,\xi)$,
while its rows and columns for $\beta$ and the nuisance parameters remain zero.

For the unrestricted second-order model in
\eqref{eq:extension-second-order-model}, define the unit-stacked operator
\begin{equation*}
\begin{split}
	\mathcal C_2
	={}&I_N\otimes R_T(\phi)
	-\rho_1W\otimes I_T-\rho_2W^2\otimes I_T\\
	&-\xi_1W\otimes L_T-\xi_2W^2\otimes L_T,
	\qquad \mathcal G_2=\mathcal C_2^{-1}.
\end{split}
\end{equation*}
The response operators for the five endogenous directions, in the root
parameterization used in \eqref{eq:extension-second-order-likelihood}, are
\begin{equation*}
\begin{aligned}
	\mathcal Q_{\kappa_1}&=WB(\kappa_2)\otimes I_T,
	&\mathcal Q_{\kappa_2}&=WB(\kappa_1)\otimes I_T,\\
	\mathcal Q_\phi&=I_N\otimes L_T,
	&\mathcal Q_{\xi_1}&=W\otimes L_T,
	&\mathcal Q_{\xi_2}&=W^2\otimes L_T.
\end{aligned}
\end{equation*}
Set $\mathcal D_d=\mathcal Q_d\mathcal G_2$ and let
$\mathcal D_{d,ij}$ denote its $(i,j)$ block. For each direction $d$, let
$\bar q_{d,i}$ be the corresponding conditional-mean regressor: the
$i$th block of $\mathcal Q_d\bar y_2$ for a root, and respectively
$\bar y_{i,-1}$, $\sum_kw_{ik}\bar y_{k,-1}$, or
$\sum_k(W^2)_{ik}\bar y_{k,-1}$ for the three dynamic coefficients. Then the
corrected contribution has the common form
\begin{equation}\label{eq:second-order-corrected-score}
\begin{split}
	g_{d,i}(\alpha_2)
	={}&e_{2i}'M\bar q_{d,i}
	+e_{2i}'M\mathcal D_{d,ii}e_{2i}
	-\tr\!\left(M\mathcal D_{d,ii}\Sigma_u\right)\\
	&+\sum_{j<i}e_{2i}'
	\left(M\mathcal D_{d,ij}+\mathcal D_{d,ji}'M\right)e_{2j}.
\end{split}
\end{equation}
For the two root directions, the trace terms aggregate to the derivatives of
the two spatial Jacobians. For the causal lag directions, the analogous traces
are zero under the time ordering. Consequently, summing
\eqref{eq:second-order-corrected-score} over $i$ reproduces the corresponding
component of the pooled score in
\eqref{eq:extension-second-order-likelihood}. The exogenous-regressor and
nuisance contributions retain the construction in
Appendix~\ref{sec:spatial_score_construction}; replacing $\mathcal G$ by
$\mathcal G_2$ throughout and profiling the nuisance block therefore gives the
feasible covariance estimator stated in Section~\ref{subsec:extension-spatial-errors}.

\section{Auxiliary Simulation Diagnostics}\label{sec:auxiliary-simulations}

This section presents counterparts to the main simulation tables
under the sparse-matrix design
and collects auxiliary diagnostics that clarify the numerical
and finite-sample patterns in the baseline experiment.

\subsection{Sparse-Matrix Implementation and Log-Determinant Evaluation}\label{subsec:aux-sparse-logdet}

Tables~\ref{tab:sim-bias-sd-sparse} and \ref{tab:sim-coverage-sparse} report the sparse-matrix counterparts of
Tables~\ref{tab:sim-bias-sd-dense} and \ref{tab:sim-coverage-dense}.
Because $W$ is row normalized and the search restricts $|\rho|<1$,
the sparse-matrix design uses the convergent expansion
\begin{equation}\label{eq:sparse-logdet-trace-expansion}
	\log|I_N-\rho W|
	=
	-\sum_{m=1}^{\infty}\frac{\rho^m}{m}\tr(W^m).
\end{equation}
The sparse-matrix calculations truncate
\eqref{eq:sparse-logdet-trace-expansion} after 30 terms
and apply the Hutchinson estimator with 25 independent Rademacher vectors:
\begin{equation}\label{eq:sparse-logdet-hutchinson}
	\log|I_N-\rho W|
	\approx
	-\sum_{m=1}^{30}\frac{\rho^m}{m}
	\left(
		\frac{1}{25}\sum_{j=1}^{25}v_j'W^m v_j
	\right).
\end{equation}
The products $W^m v_j$ are computed recursively by sparse matrix-vector multiplication,
so the calculation does not convert $W$ to a dense matrix.
The point estimates remain accurate and have dispersion similar to the dense-matrix results.
In the original parameter units,
the average bias of $\hat\rho$ is $-0.0053$ and its largest absolute bias is about $0.0101$.
Pair-corrected coverage for $\rho$ averages $0.912$,
compared with $0.947$ in the dense-matrix design in Table~\ref{tab:sim-coverage-dense},
while coverage for $\phi$ and the regression slopes is more similar across the dense- and sparse-matrix designs.

\input{simulation_bias_sd_sparse.tex}

\input{simulation_coverage_sparse.tex}

Table~\ref{tab:sim-sparse-logdet} compares three determinant calculations in the six designs with $T=20$.
Each replication uses the same simulated panel under a 30-term Hutchinson approximation with 25 Rademacher vectors,
a 60-term approximation with 100 vectors,
and exact sparse LU factorization.
All 3600 estimations converge.

Average coverage for $\rho$ is $0.872$ under the faster 30-term, 25-vector approximation,
$0.924$ under the richer approximation,
and $0.943$ under exact sparse LU.
The trace approximation leaves the reported standard errors nearly unchanged but increases the empirical dispersion of $\hat\rho$ across replications.
Because the network is redrawn in every replication,
its approximation error becomes an additional source of Monte Carlo variation that is not part of the conditional asymptotic variance.
Coverage for $\phi$ remains between $0.935$ and $0.960$ under all three calculations.
Thus the trace method offers a useful speed advantage,
but exact sparse LU is preferable when the objective is precise inference for $\rho$ and the factorization is computationally feasible.

\input{simulation_sparse_logdet_diagnostic.tex}

\subsection{Short-Panel Dynamic-Coefficient Coverage}\label{subsec:aux-symmetric}

Table~\ref{tab:sim-symmetric-pilot} holds $\rho+\phi=0.9$ and varies the decomposition between spatial and dynamic dependence.
The design uses the eigenvalue implementation for log determinants,
so it isolates the role of the time dimension from trace-approximation error.
At $T=5$,
coverage for $\phi$ ranges from $0.845$ to $0.939$ and tends to improve as $\phi$ falls.
At $T=20$,
it ranges from $0.935$ to $0.965$ across the same parameter pairs.
Coverage for $\rho$ and the regression slopes is generally close to nominal in both cases.

This pattern is consistent with a finite-sample limitation of first-order Wald inference for the dynamic coefficient.
When $T$ is very short,
each unit contributes only a small number of transitions from which to learn the dynamic response,
and the studentized statistic for $\phi$ can remain non-Gaussian at the sample sizes considered here.
The absence of under-coverage at $T=20$ shows that a moderately larger time dimension resolves this problem.

\input{simulation_symmetric_design_pilot.tex}

\subsection{Unstable-Design Stress Test}\label{sec:explosive-stress-test}

Table~\ref{tab:simulation-explosive-stress} reports bias,
sampling dispersion,
and coverage under a simulation design that motivated the stable parameterization used in the main text.
The standard errors are computed using a simple outer-product of the rearranged scores $g_i(\widehat\alpha)$
instead of the $\widehat\Omega_{M,\text{pair}}(\widehat\alpha)$.
It fixes $(\rho,\phi)=(0.8,0.5)$,
so the spatially uniform component evolves at rate $\phi/(1-\rho)=2.5$.
The finite sample remains well defined because $T$ is fixed,
but the design lies outside the dynamic stability region and its outcome scale grows rapidly with the horizon.
These results are consistent with a finite-horizon amplification mechanism:
for each fixed $T$,
the reduced-form operator remains finite,
but its sensitivity to $\rho$ and $\phi$ can grow rapidly with the horizon
and produce unusually small Monte Carlo standard deviations of
$\hat\rho$ and $\hat\phi$
without changing the large-$N$ rate.
The broadly accurate sandwich coverage at $T=10,20$ is therefore reassuring.

\input{simulation_explosive_stress.tex}

\section{Application Factor-Rank Sensitivity}\label{sec:application-factor-rank}

For the empirical application, Table~\ref{tab:application-second-order-rank}
compares the preferred two-factor second-order estimates with the corresponding three-factor estimates.
The symmetric filter coefficients remain close across the two ranks,
and the education coefficient remains positive in every specification.
The farm-population coefficient is negative at both ranks in every specification.
Although the magnitudes of both covariate coefficients vary with the factor rank,
the conclusion that education is economically important does not depend on the chosen factor rank.

\input{application_second_order_rank.tex}

\section{Joint QML Treatment of the Initial Outcome}\label{sec:joint-initial-condition}

The baseline analysis conditions on the observed initial cross section and includes $y_0$ in the loading projection.
This section presents an alternative that models $y_0$ jointly with $y_1,\ldots,y_T$
and projects the loadings only on functions of the observed regressors.

\subsection{Initial Equation and Regressor-Based Loading Projection}

Suppose that initial-period regressors $X_0$ are observed.
Specify the initial cross section by
\begin{equation}\label{eq:joint-initial-sar}
B(\rho_{\mathrm I})y_0
=
X_0\beta_{\mathrm I}
+WX_0\gamma_{\mathrm I}
+\ones[N]\delta_0
+\Lambda f_0
+\varepsilon_0,
\qquad
B(\rho_{\mathrm I})=I_N-\rho_{\mathrm I}W.
\end{equation}
The term $WX_0\gamma_{\mathrm I}$ may be omitted,
and both direct initial-regressor terms may be omitted when $X_0$ is unavailable.
Equation~\eqref{eq:joint-initial-sar} is best viewed as an auxiliary spatial model for the initial cross section, allowing the component of $y_0$ associated with the latent unit heterogeneity to lie in the loading space spanned by $\Lambda$.
Accordingly, $f_0$ is a coefficient vector on this loading space and need not be interpreted as a realization of the common shocks governing the sample-period outcomes.
The parameters $\rho_{\mathrm I}$, $\beta_{\mathrm I}$, $\gamma_{\mathrm I}$, $\delta_0$, and $f_0$ are auxiliary nuisance parameters introduced only to model the initial cross section and are not parameters of substantive interest.

For completeness, the sample-period outcomes continue to satisfy, for
$B(\rho)=I_N-\rho W$,
\begin{equation}\label{eq:joint-sample-period}
B(\rho)y_t
=
\phi y_{t-1}
+
X_t\beta
+
\mathbf{1}_N\delta_t
+
\Lambda f_t
+
\varepsilon_t,
\qquad
t=1,\ldots,T.
\end{equation}
Thus, the joint specification consists of the auxiliary initial equation \eqref{eq:joint-initial-sar} together with the original sample-period model \eqref{eq:joint-sample-period}. Only the parameters in the sample-period model are of substantive interest.

Let $Q_X$ collect fixed-dimensional summaries of the regressor history.
For a fixed $L\geq 0$, define
\begin{equation*}
z_i^X
=
\left[
(Q_X)_{i\cdot}',
(WQ_X)_{i\cdot}',
\ldots,
(W^LQ_X)_{i\cdot}'
\right]',
\qquad
\mathcal C_N^X
=
\sigma(X_0,X_1,\ldots,X_T,W).
\end{equation*}
Write $Z_X=(z_1^X,\ldots,z_N^X)'$ and specify
\begin{equation}\label{eq:joint-initial-loading-projection}
\lambda_i
=
Az_i^X+\eta_i,
\qquad
\E(\eta_i\mid\mathcal C_N^X)=0,
\qquad
\E(\eta_i\eta_i'\mid\mathcal C_N^X)=\Sigma_\eta.
\end{equation}
As in the SERL specification, $Az_i^X$ captures the component of loading
heterogeneity that is predictable from the conditioning information and may
therefore vary systematically across units through the spatial transformations
in $z_i^X$. The residual $\eta_i$ captures the remaining purely idiosyncratic
loading heterogeneity and is assumed to be independent across $i$, conditional
on $\mathcal C_N^X$.

\subsection{The Joint Working Likelihood}

For the initial period, define
\begin{equation*}
	u_0
	=
	B(\rho_{\mathrm I})y_0
	-X_0\beta_{\mathrm I}
	-WX_0\gamma_{\mathrm I}.
\end{equation*}
For $t=1,...,T$, retain
\begin{equation*}
	u_t
	=
	B(\rho)y_t-\phi y_{t-1}-X_t\beta.
\end{equation*}
The unit-level residual vector is
\begin{equation*}
	e_i^*(\alpha)
	=
	\begin{pmatrix}
		u_{i0}-\delta_0-f_0'Az_i^X\\
		u_{i1}-\delta_1-f_1'Az_i^X\\
		\vdots\\
		u_{iT}-\delta_T-f_T'Az_i^X
	\end{pmatrix}.
\end{equation*}
At the true parameter,
\begin{equation*}
	e_i^*=F^*\eta_i+\varepsilon_i^*,
	\qquad
	F^*=\begin{pmatrix}f_0'\\F\end{pmatrix},
	\qquad
	\varepsilon_i^*=(\varepsilon_{i0},...,\varepsilon_{iT})'.
\end{equation*}
Here $f_0$ is the auxiliary coefficient vector on the loading space for the
initial equation, whereas the rows of $F$ retain their interpretation as the
sample-period factor realizations. The top $r\times r$ block of the
sample-period factor path $F$ remains normalized to $I_r$, which also fixes the
rotation of $f_0$.

Suppose that,
conditional on $\mathcal C_N^X$,
the vectors $(\eta_i',\varepsilon_i^{*\prime})'$ are iid across $i$,
with $\eta_i$ orthogonal to $\varepsilon_i^*$ and
\begin{equation*}
	D_\varepsilon^*
	=
	\diag(\sigma_0^2,\sigma_1^2,...,\sigma_T^2).
\end{equation*}
Then
\begin{equation}\label{eq:joint-initial-covariance}
	\E(e_i^*\mid\mathcal C_N^X)=0,
	\qquad
	\E(e_i^*e_i^{*\prime}\mid\mathcal C_N^X)
	=
	\Sigma_u^*
	:=
	F^*\Sigma_\eta F^{*\prime}+D_\varepsilon^*.
\end{equation}

The transformation from $(y_0',y_1',...,y_T')'$ to the structural residuals is block lower triangular.
Its diagonal blocks are $B(\rho_{\mathrm I})$ once and $B(\rho)$ $T$ times,
so its determinant is
\begin{equation*}
	|B(\rho_{\mathrm I})|\,|B(\rho)|^T.
\end{equation*}
Consequently,
the normalized Gaussian working log-likelihood is
\begin{equation}\label{eq:joint-initial-qml}
	\ell_N^*(\alpha)
	=
	\frac1N\log|B(\rho_{\mathrm I})|
	+\frac{T}{N}\log|B(\rho)|
	-\frac12\log|\Sigma_u^*|
	-
	\frac1{2N}\sum_{i=1}^N
	e_i^*(\alpha)'\Sigma_u^{*-1}e_i^*(\alpha).
\end{equation}
No Gaussianity is imposed;
\eqref{eq:joint-initial-qml} is a working likelihood based on
\eqref{eq:joint-initial-covariance}.

\subsection{Estimation and Asymptotic Covariance}

For fixed $(\rho_{\mathrm I},\rho)$,
the transformed outcomes are observed,
the initial coefficients $(\beta_{\mathrm I},\gamma_{\mathrm I})$
and the sample-period coefficients $(\phi,\beta)$ enter linearly,
and the GLS and Gaussian-working-model updates apply to the stacked $T+1$ system.
The covariance update uses $F^*$ and $D_\varepsilon^*$,
while the update for $F^*$ estimates the auxiliary coefficient $f_0$
together with the sample-period factor path $F$,
subject to the normalization above.
Profiling these blocks leaves a bounded two-dimensional optimization over
$(\rho_{\mathrm I},\rho)$.

To describe the corrected covariance,
similar to the Jacobian operator $\mathcal J$ in
\eqref{eq:jacobian_matrix},
we stack the initial and sample-period cross sections by time,
with all $N$ units within each period,
and write the joint operator as
\begin{equation}\label{eq:joint-initial-operator}
\widetilde{\mathcal C}^*
=
\begin{pmatrix}
B(\rho_{\mathrm I}) & 0 & 0 & \cdots & 0 \\
-\phi I_N & B(\rho) & 0 & \cdots & 0 \\
0 & -\phi I_N & B(\rho) & \cdots & 0 \\
\vdots & \ddots & \ddots & \ddots & \vdots \\
0 & \cdots & 0 & -\phi I_N & B(\rho)
\end{pmatrix},
\qquad
\widetilde{\mathcal G}^*
:=
(\widetilde{\mathcal C}^*)^{-1}.
\end{equation}
For scalar coefficients
$d\in\{\rho_{\mathrm I},\rho,\phi\}$,
define the time-stacked response direction by
\begin{equation*}
\widetilde{\mathcal Q}_d
:=
-\frac{\partial\widetilde{\mathcal C}^*}{\partial d},
\qquad
\widetilde{\mathcal D}_d
:=
\widetilde{\mathcal Q}_d\widetilde{\mathcal G}^*.
\end{equation*}
Thus,
$\widetilde{\mathcal Q}_{\rho_{\mathrm I}}$ has $W$ in the
initial-period diagonal block,
whereas $\widetilde{\mathcal Q}_{\rho}$ has $W$ in each
sample-period diagonal block.
The matrix $\widetilde{\mathcal Q}_{\phi}$ has $I_N$ in each block
along the first subdiagonal.

Permute the time-stacked system
$\widetilde{\mathcal C}^*$,
$\widetilde{\mathcal G}^*$,
$\widetilde{\mathcal Q}_d$, and
$\widetilde{\mathcal D}_d$
to a unitwise ordering,
denoted by
$\mathcal C^*$,
$\mathcal G^*$,
$\mathcal Q_d$, and
$\mathcal D_d$, respectively.
Let $\mathcal D_{d,ij}$ denote the
$(T+1)\times(T+1)$ $(i,j)$ block of $\mathcal D_d$.
Replacing the response matrices in
Appendix~\ref{sec:spatial_score_construction}
by these blocks gives the corrected contribution for each endogenous direction.
The trace corrections reproduce the two Jacobian derivatives:
one copy of
$\tr\{B(\rho_{\mathrm I})^{-1}W\}$
for $\rho_{\mathrm I}$ and
$T$ copies of
$\tr\{B(\rho)^{-1}W\}$
for $\rho$.
All nuisance contributions retain their unitwise form in the enlarged
$T+1$ system.

Because $T+1$ and the nuisance dimension remain fixed,
the consistency and asymptotic-normality argument in
Section~\ref{sec:asymptotics}
carries over when both spatial filters and the joint reduced-form operator
are uniformly bounded,
the enlarged population criterion is uniquely maximized,
and the augmented Hessian and predictable score covariance are nonsingular.
The argument of Proposition~\ref{prop:feasible_score_covariance}
then establishes consistency of the feasible covariance estimator
based on the augmented contributions.

\end{document}

%% file: algorithm_profile_rho_description.tex
\begin{algorithm}[h]
	\caption{Block coordinate estimation with a conditional outer search over $\rho$}\label{alg:profile_rho_description}
	\begin{algorithmic}[1]
		\Require Initial value for $\rho$, starting values for $(\phi,\beta)$, and starting values for $(\delta,A,F,\Sigma_\eta,D_\varepsilon)$
		\Repeat
		\State Given current $\rho$, form transformed innovation matrix $U(\rho,\phi,\beta)$
		\Repeat
		\State Update $(\delta,A)$ by concentration given the current factor path and covariance parameters
		\State Compute the conditional moments of the loading heterogeneity 
		\State Update $(\Sigma_\eta,D_\varepsilon)$
		\State Update factor path $F$ and renormalize
		\State Update $(\phi,\beta)$ by GLS
		\Until{inner convergence}
		\State Update $\rho$ by maximizing the conditional criterion over the admissible parameter space,
		holding all current non-$\rho$ parameters fixed
		\Until{outer convergence}
		\State \Return $(\widehat\rho,\widehat\phi,\widehat\beta,\widehat\delta,\widehat A,\widehat F,\widehat\Sigma_\eta,\widehat D_\varepsilon)$
	\end{algorithmic}
\end{algorithm}

%% file: simulation_bias_sd.tex
\begin{sidewaystable}[t]
	\centering
	\scriptsize
	\begin{filecontents*}{code/output/simulation_bias_sd.csv}
T,N,knn_type,rho,phi,beta1,beta2,n_success,n_strict_converged,max_accepted_rho_gap,bias_x100_rho,bias_x100_phi,bias_x100_beta1,bias_x100_beta2,sd_rho,sd_phi,sd_beta1,sd_beta2
5,500,dense,0.2,0.4,0.8,-0.3,997,833,1.82297081494087e-06,-0.459755626240546,-0.302983201397767,-0.0772326447901689,-0.373345296097861,0.0367537630642786,0.0581250161935009,0.0632980099687339,0.0635800044472219
5,500,dense,0.5,0.25,0.8,-0.3,995,818,1.93177352103557e-06,-0.508767124907687,0.0767479178730612,0.187117956792235,0.00499677425928773,0.0282771606035543,0.0417230800511041,0.0604028164233799,0.0619321130005808
5,500,dense,0.8,0.1,0.8,-0.3,996,879,1.99411658385884e-06,-0.410240487755971,0.107813663750896,0.0410308700807518,-0.26219522410439,0.0154831120421616,0.0222716007622484,0.0622169358857475,0.0636840716166945
5,1000,dense,0.2,0.4,0.8,-0.3,1000,847,1.31122775545722e-06,-0.290869791903224,-0.0172913782247633,-0.0676766295099115,-0.108750383696419,0.0257869530598367,0.0442831782362208,0.0427595099353243,0.0437803404682156
5,1000,dense,0.5,0.25,0.8,-0.3,999,850,1.68164221625577e-06,-0.389253023228031,0.0534087342929079,0.0448507906716537,0.00108206359191642,0.0199189520190167,0.0306376496845181,0.0447360053415976,0.0432682803687384
5,1000,dense,0.8,0.1,0.8,-0.3,999,902,1.66576987958855e-06,-0.308412603007015,0.00208319631494809,-0.102222628229031,0.212235128717554,0.010757123751992,0.0168043482754715,0.0448731237757432,0.043376145034306
10,500,dense,0.2,0.4,0.8,-0.3,999,983,9.04782084754086e-07,-0.426129384641473,0.0219216392103941,0.031671654276242,-0.0537886146826994,0.0275651548773917,0.0166321236701168,0.03897484131746,0.0410396416523506
10,500,dense,0.5,0.25,0.8,-0.3,1000,980,9.25355457037647e-07,-0.419057958168378,0.0391350850493259,-0.0768715996899183,-0.177246939874934,0.0213903244357121,0.0163950329533106,0.0404207063410089,0.0413786052983261
10,500,dense,0.8,0.1,0.8,-0.3,1000,991,2.91874263869119e-07,-0.291776219289003,0.0566835779725764,0.178069020020826,-0.316671757714054,0.0112445015237067,0.011294096291161,0.0405796936128631,0.0419548249571904
10,1000,dense,0.2,0.4,0.8,-0.3,1000,993,5.55107213973116e-07,-0.309488170300271,-0.0046897704087291,0.146172046484803,-0.111603922867638,0.0190285377767062,0.0118000718196109,0.0287240352423378,0.0292565376996012
10,1000,dense,0.5,0.25,0.8,-0.3,1000,994,1.29151019989671e-06,-0.205721918009028,0.00522740928790425,0.0577644630144761,0.238683995972338,0.0149485671697791,0.0112829790906047,0.0288594972646195,0.0290274057388683
10,1000,dense,0.8,0.1,0.8,-0.3,999,997,2.360081913233e-07,-0.217960224652876,0.0497509911844319,0.114788930851705,0.000495086994163522,0.00792891538123292,0.00774471813336469,0.0290210421761908,0.0285338382669714
20,500,dense,0.2,0.4,0.8,-0.3,1000,1000,9.99422797298699e-09,-0.337797597352946,-0.0337085979357066,0.110762113734874,-0.0249249096561815,0.0207940082108163,0.00982980557262985,0.0284885543518861,0.0278104814624903
20,500,dense,0.5,0.25,0.8,-0.3,1000,1000,9.98397853280153e-09,-0.312124215808448,0.0526052094717584,0.107890277497186,0.0104538060066439,0.016462913634552,0.00999219037839394,0.0289267387315838,0.0287556965685492
20,500,dense,0.8,0.1,0.8,-0.3,1000,1000,9.99754501407324e-09,-0.264326162158603,0.0524370332361495,-0.0170923418326818,-0.0518937916227675,0.00921491552617981,0.00730420711747235,0.0284963420728295,0.0273567070199463
20,1000,dense,0.2,0.4,0.8,-0.3,1000,1000,9.99743718366197e-09,-0.179474380190505,-0.0350828617232773,-0.00530507667365947,-0.149246425445337,0.0145819397885216,0.00707250261179715,0.01905259145482,0.0192873100315018
20,1000,dense,0.5,0.25,0.8,-0.3,1000,1000,9.99364679898918e-09,-0.303045696864704,0.0387195078927566,-0.01041316260688,-0.0215787438832505,0.0113458171338136,0.00698032024693047,0.0201233595645998,0.0203750511276125
20,1000,dense,0.8,0.1,0.8,-0.3,1000,1000,9.99689941938442e-09,-0.202726430748857,0.0306704716201892,-0.0977214030421941,-0.0186401590910964,0.00646411708473614,0.00503111719546107,0.0192572783291392,0.020171893486486
5,500,sparse,0.2,0.4,0.8,-0.3,998,825,1.58252738949316e-06,-1.01411095455599,-0.430299368978348,-0.086014302807377,-0.0756704523738878,0.0359586654586775,0.0552949004747228,0.0601017098872895,0.0629423591758688
5,500,sparse,0.5,0.25,0.8,-0.3,995,827,1.97997127326932e-06,-0.746373890915683,-0.0252487339780716,-0.280182944395845,0.169441416527909,0.026902914833434,0.0432769233848982,0.0621340439108369,0.0609777264243145
5,500,sparse,0.8,0.1,0.8,-0.3,999,872,1.62949776283394e-06,-0.623024013082227,0.153590696025054,0.199131740137407,-0.108015713028205,0.0139806929616721,0.0221132905153283,0.0611079014368615,0.0618981675224667
5,1000,sparse,0.2,0.4,0.8,-0.3,999,855,1.51875221085507e-06,-0.384782637795807,-0.125440025310162,-0.269896260368342,-0.114075090120059,0.0244475678812855,0.0475984515397446,0.0436915593290088,0.0433212222340386
5,1000,sparse,0.5,0.25,0.8,-0.3,995,824,1.57637906422758e-06,-0.466824691687437,-0.00803963577670528,-0.188996169569729,-0.173436887442317,0.0182474929517717,0.0338792636513691,0.0442104625664803,0.0430015435939337
5,1000,sparse,0.8,0.1,0.8,-0.3,997,875,1.8146786940898e-06,-0.316775855776813,0.168886299506137,0.0398788117109989,-0.0454049625120508,0.00980867920840817,0.0162594704435037,0.0419434621323432,0.0429373749507834
10,500,sparse,0.2,0.4,0.8,-0.3,999,973,1.05069730052532e-06,-0.81136317591462,0.0371281981368402,0.209609355799524,-0.0741729552938347,0.0238143042003844,0.016229310472787,0.040932198965873,0.0400117815632172
10,500,sparse,0.5,0.25,0.8,-0.3,1000,982,1.63777004308363e-06,-0.687096915992539,0.0580853901494105,-0.077960032436405,0.103439422766668,0.0182149423258882,0.0157426399157064,0.0402417081281556,0.041722780354582
10,500,sparse,0.8,0.1,0.8,-0.3,1000,993,3.1720817739167e-07,-0.443795682516784,0.0856231664909612,0.10813153746334,-0.00201703991531594,0.00979493859577369,0.0107036609650286,0.0414048909049964,0.0409091944626397
10,1000,sparse,0.2,0.4,0.8,-0.3,1000,996,4.48100873073898e-07,-0.554970926661293,-0.0274063307564293,-0.00543838966083919,-0.0707535876376089,0.0166167439640114,0.011563718679672,0.0280614563379432,0.0283710215063798
10,1000,sparse,0.5,0.25,0.8,-0.3,1000,993,4.7102109979047e-07,-0.365493402823214,0.0663466514239098,0.0505495892475708,0.10362603461277,0.0122607190118867,0.0107604317141799,0.0269614091680698,0.0283050827187717
10,1000,sparse,0.8,0.1,0.8,-0.3,1000,997,1.93774521428658e-07,-0.299876914050505,0.0563290613702016,-0.0873134282936568,0.0533461636466706,0.00707078181842633,0.00790938602887919,0.0285189617972791,0.0294373410816582
20,500,sparse,0.2,0.4,0.8,-0.3,1000,1000,9.99386917666101e-09,-0.610500201764859,0.0176308193448922,0.0805430225325764,-0.0334422620023321,0.0170306140388318,0.00943437527940112,0.0283808634089384,0.0274721701783631
20,500,sparse,0.5,0.25,0.8,-0.3,1000,1000,9.99518562361246e-09,-0.718714912390452,0.0948163143044409,-0.0201213084026198,-0.117619531822637,0.0135148997115718,0.00992473135270245,0.0279833555814847,0.0278808787114743
20,500,sparse,0.8,0.1,0.8,-0.3,1000,1000,9.98491900272569e-09,-0.524571467041621,0.115795003752012,0.0711337487243799,0.165323386797148,0.00698146336563234,0.00721168114758791,0.0284724965824085,0.0277604191036248
20,1000,sparse,0.2,0.4,0.8,-0.3,1000,1000,9.96105220529842e-09,-0.392719895995017,0.00194710325869338,0.0318075392295233,-0.0700391377540614,0.0120292396267913,0.00681172162761316,0.0204110439328039,0.0194542000716454
20,1000,sparse,0.5,0.25,0.8,-0.3,1000,1000,9.9944592601986e-09,-0.326149746215349,0.0535845627109437,0.0633956132715966,0.0117348344808486,0.00875332386571972,0.00671192276914995,0.0196274927834841,0.0202288914525905
20,1000,sparse,0.8,0.1,0.8,-0.3,1000,1000,9.99904770093707e-09,-0.248210202053009,0.0580160026854786,-0.113680397507061,0.0602819350630982,0.00488185262486115,0.00493799679222201,0.0208699440029233,0.0196068040943593
\end{filecontents*}
\csvreader[
	respect underscore=true,
	column names={
	T=\Tval,
	N=\Nval,
	knn_type=\knntype,
	rho=\truerho,
	phi=\truephi,
	beta1=\truebetaone,
	beta2=\truebetatwo,
	bias_x100_rho=\BIASrho,
	sd_rho=\SDrho,
	bias_x100_phi=\BIASphi,
	sd_phi=\SDphi,
	bias_x100_beta1=\BIASbetaone,
	sd_beta1=\SDbetaone,
	bias_x100_beta2=\BIASbetatwo,
	sd_beta2=\SDbetatwo,
	},
	filter strcmp={\knntype}{dense},
	tabular=rrrrrr|rrrrrrrr,
	table head=
	\toprule
	\multicolumn{1}{c}{$T$} &
	\multicolumn{1}{c}{$N$} &
	\multicolumn{1}{c}{$\rho^{\text{true}}$} &
	\multicolumn{1}{c}{$\phi^{\text{true}}$} &
	\multicolumn{1}{c}{$\beta_1^{\text{true}}$} &
	\multicolumn{1}{c}{$\beta_2^{\text{true}}$} &
	\multicolumn{2}{c}{$\hat\rho$} &
	\multicolumn{2}{c}{$\hat\phi$} &
	\multicolumn{2}{c}{$\hat\beta_1$} &
	\multicolumn{2}{c}{$\hat\beta_2$} \\
	\cmidrule(lr){7-8}\cmidrule(lr){9-10}
	\cmidrule(lr){11-12}\cmidrule(lr){13-14}
	&&&&&&
	\multicolumn{1}{c}{$100\times\mathrm{Bias}$} &
	\multicolumn{1}{c}{SD} &
	\multicolumn{1}{c}{$100\times\mathrm{Bias}$} &
	\multicolumn{1}{c}{SD} &
	\multicolumn{1}{c}{$100\times\mathrm{Bias}$} &
	\multicolumn{1}{c}{SD} &
	\multicolumn{1}{c}{$100\times\mathrm{Bias}$} &
	\multicolumn{1}{c}{SD} \\
	\midrule,
	table foot = \bottomrule
	]{code/output/simulation_bias_sd.csv}{}{
	\Tval & \Nval &
	\num[group-digits=false,round-precision=1]{\truerho} &
	\num[group-digits=false,round-precision=1]{\truephi} &
	\num[group-digits=false,round-precision=1]{\truebetaone} &
	\num[group-digits=false,round-precision=1]{\truebetatwo} &
	\num[group-digits=false,round-precision=4]{\BIASrho} &
	\num[group-digits=false,round-precision=5]{\SDrho} &
	\num[group-digits=false,round-precision=4]{\BIASphi} &
	\num[group-digits=false,round-precision=5]{\SDphi} &
	\num[group-digits=false,round-precision=4]{\BIASbetaone} &
	\num[group-digits=false,round-precision=5]{\SDbetaone} &
	\num[group-digits=false,round-precision=4]{\BIASbetatwo} &
	\num[group-digits=false,round-precision=5]{\SDbetatwo}
	}
	\caption{\textsc{Dense-matrix Design}.
	Bias and standard deviations.
	Each row is defined by $T$, $N$, and the true parameter values.
	The remaining columns report Monte Carlo bias, multiplied by 100,
	and the unscaled sample standard deviation of $\hat\rho$, $\hat\phi$,
	$\hat\beta_1$, and $\hat\beta_2$.
	Each of the 18 designs requests 1000 replications;
	a fit is accepted if it either converges strictly or has a converged inner loop and a final $\rho$ stationarity gap no larger than $2\times10^{-6}$.
	Between 995 and 1000 estimates per design enter the calculations,
	for 17,984 accepted estimates in total.
	The log determinant is evaluated from the eigenvalues of $W$.}
	\label{tab:sim-bias-sd-dense}
\end{sidewaystable}

%% file: simulation_coverage.tex
\begin{sidewaystable}[t]
	\centering
	\scriptsize
	\begin{filecontents*}{code/output/simulation_coverage_spatial_comparison.csv}
"spec_id","T","N","knn_type","rho","phi","beta1","beta2","coverage_no_hc1_rho","coverage_no_hc1_phi","coverage_pair_hc1_rho","coverage_pair_hc1_phi","coverage_beta1","coverage_beta2"
19,5,500,"dense",0.2,0.4,0.8,-0.3,0.951991828396323,0.861082737487232,0.957099080694586,0.861082737487232,0.947906026557712,0.955056179775281
25,5,500,"dense",0.5,0.25,0.8,-0.3,0.958543983822042,0.902932254802831,0.962588473205258,0.902932254802831,0.964610717896865,0.950455005055612
31,5,500,"dense",0.8,0.1,0.8,-0.3,0.943718592964824,0.932663316582915,0.950753768844221,0.933668341708543,0.957788944723618,0.955778894472362
22,5,1000,"dense",0.2,0.4,0.8,-0.3,0.956521739130435,0.849342770475228,0.96056622851365,0.849342770475228,0.944388270980789,0.9474216380182
28,5,1000,"dense",0.5,0.25,0.8,-0.3,0.950304259634888,0.879310344827586,0.95131845841785,0.879310344827586,0.95131845841785,0.947261663286004
34,5,1000,"dense",0.8,0.1,0.8,-0.3,0.94572864321608,0.927638190954774,0.949748743718593,0.92964824120603,0.935678391959799,0.965829145728643
20,10,500,"dense",0.2,0.4,0.8,-0.3,0.937813440320963,0.936810431293882,0.949849548645938,0.936810431293882,0.94383149448345,0.950852557673019
26,10,500,"dense",0.5,0.25,0.8,-0.3,0.94,0.941,0.956,0.942,0.951,0.944
32,10,500,"dense",0.8,0.1,0.8,-0.3,0.908,0.947,0.918,0.95,0.935,0.954
23,10,1000,"dense",0.2,0.4,0.8,-0.3,0.944,0.948,0.951,0.948,0.955,0.948
29,10,1000,"dense",0.5,0.25,0.8,-0.3,0.934,0.951,0.943,0.951,0.946,0.947
35,10,1000,"dense",0.8,0.1,0.8,-0.3,0.928,0.943,0.934,0.946,0.934,0.966
21,20,500,"dense",0.2,0.4,0.8,-0.3,0.933,0.949,0.953,0.949,0.943,0.945
27,20,500,"dense",0.5,0.25,0.8,-0.3,0.909,0.941,0.94,0.944,0.953,0.949
33,20,500,"dense",0.8,0.1,0.8,-0.3,0.9,0.947,0.932,0.956,0.949,0.946
24,20,1000,"dense",0.2,0.4,0.8,-0.3,0.926,0.953,0.938,0.954,0.949,0.948
30,20,1000,"dense",0.5,0.25,0.8,-0.3,0.941,0.938,0.957,0.94,0.934,0.95
36,20,1000,"dense",0.8,0.1,0.8,-0.3,0.914,0.941,0.934,0.947,0.939,0.95
1,5,500,"sparse",0.2,0.4,0.8,-0.3,0.940755873340143,0.870275791624106,0.943820224719101,0.870275791624106,0.955056179775281,0.960163432073544
7,5,500,"sparse",0.5,0.25,0.8,-0.3,0.921827411167513,0.889340101522843,0.926903553299492,0.891370558375635,0.951269035532995,0.945177664974619
13,5,500,"sparse",0.8,0.1,0.8,-0.3,0.926706827309237,0.936746987951807,0.937751004016064,0.938755020080321,0.958835341365462,0.947791164658635
4,5,1000,"sparse",0.2,0.4,0.8,-0.3,0.923076923076923,0.831983805668016,0.92914979757085,0.831983805668016,0.950404858299595,0.966599190283401
10,5,1000,"sparse",0.5,0.25,0.8,-0.3,0.931632653061224,0.871428571428571,0.936734693877551,0.871428571428571,0.954081632653061,0.951020408163265
16,5,1000,"sparse",0.8,0.1,0.8,-0.3,0.925777331995988,0.92778335005015,0.929789368104313,0.929789368104313,0.942828485456369,0.954864593781344
2,10,500,"sparse",0.2,0.4,0.8,-0.3,0.902,0.945,0.92,0.945,0.943,0.945
8,10,500,"sparse",0.5,0.25,0.8,-0.3,0.907,0.951,0.92,0.952,0.939,0.944
14,10,500,"sparse",0.8,0.1,0.8,-0.3,0.882882882882883,0.940940940940941,0.900900900900901,0.951951951951952,0.931931931931932,0.950950950950951
5,10,1000,"sparse",0.2,0.4,0.8,-0.3,0.913,0.948,0.92,0.948,0.945,0.942
11,10,1000,"sparse",0.5,0.25,0.8,-0.3,0.901,0.95,0.909,0.951,0.965,0.948
17,10,1000,"sparse",0.8,0.1,0.8,-0.3,0.899,0.938,0.904,0.941,0.949,0.947
3,20,500,"sparse",0.2,0.4,0.8,-0.3,0.88,0.952,0.918,0.953,0.953,0.946
9,20,500,"sparse",0.5,0.25,0.8,-0.3,0.864,0.938,0.896,0.941,0.953,0.944
15,20,500,"sparse",0.8,0.1,0.8,-0.3,0.864,0.943,0.897,0.952,0.949,0.95
6,20,1000,"sparse",0.2,0.4,0.8,-0.3,0.869,0.937,0.883,0.937,0.938,0.937
12,20,1000,"sparse",0.5,0.25,0.8,-0.3,0.867,0.949,0.883,0.95,0.939,0.955
18,20,1000,"sparse",0.8,0.1,0.8,-0.3,0.841,0.945,0.854,0.948,0.948,0.95
\end{filecontents*}
\csvreader[
	respect underscore=true,
	column names={
	2=\Tval,
	3=\Nval,
	4=\knntype,
	5=\truerho,
	6=\truephi,
	7=\truebetaone,
	8=\truebetatwo,
	9=\uncorrectedrho,
	10=\uncorrectedphi,
	11=\hcspatialrho,
	12=\hcspatialphi,
	13=\spatialbetaone,
	14=\spatialbetatwo
	},
	filter strcmp={\knntype}{"dense"},
	tabular=rr|rrrr|rr|rr|rr,
	table head=
	\toprule
	&&&&&&
	\multicolumn{2}{c}{No correction} &
	\multicolumn{2}{c}{Pair-only HC1} &
	\multicolumn{2}{c}{Regression slopes} \\
	\cmidrule(lr){7-8}\cmidrule(lr){9-10}\cmidrule(lr){11-12}
	$T$ & $N$ &
	$\rho^{\mathrm{true}}$ & $\phi^{\mathrm{true}}$ &
	$\beta_1^{\mathrm{true}}$ & $\beta_2^{\mathrm{true}}$ &
	$\rho$ & $\phi$ & $\rho$ & $\phi$ & $\beta_1$ & $\beta_2$ \\
	\midrule,
	table foot=\bottomrule
	]{code/output/simulation_coverage_spatial_comparison.csv}{}{
	\Tval & \Nval &
	\num[group-digits=false,round-precision=1]{\truerho} &
	\num[group-digits=false,round-precision=1]{\truephi} &
	\num[group-digits=false,round-precision=1]{\truebetaone} &
	\num[group-digits=false,round-precision=1]{\truebetatwo} &
	\num[group-digits=false,round-precision=3]{\uncorrectedrho} &
	\num[group-digits=false,round-precision=3]{\uncorrectedphi} &
	\num[group-digits=false,round-precision=3]{\hcspatialrho} &
	\num[group-digits=false,round-precision=3]{\hcspatialphi} &
	\num[group-digits=false,round-precision=3]{\spatialbetaone} &
	\num[group-digits=false,round-precision=3]{\spatialbetatwo}
	}
	\caption{\textsc{Dense-matrix Design}.
	Coverage of nominal $95\%$ confidence intervals.
	The first two coverage columns use the martingale-sample covariance estimator in \eqref{eq:martingale_covariance_estimator} without a finite-sample correction,
	while the next two use the pair-only HC1 correction in \eqref{eq:Omega_martingale_pair_hc1}.
	Because the correction leaves coverage for $\beta_1$ and $\beta_2$ unchanged in every design,
	the regression-slope coverage is reported once.
	Each of the 18 designs requests 1000 Monte Carlo replications;
	between 979 and 1000 replications per design produce strictly converged estimates and complete sandwich inference,
	for 17,930 successful replications in total.
	The log determinant is evaluated from the eigenvalues of $W$.}
	\label{tab:sim-coverage-dense}
\end{sidewaystable}

%% file: simulation_initial_condition_projection.tex
\newcommand{\initialprojectionlabel}[1]{%
	\ifcase#1 Full $L=0$%
	\or Full $L=1$%
	\or Full $L=2$%
	\or Full $L=3$%
	\or $y$-only $L=2$%
	\fi}

\begin{sidewaystable}[t]
	\centering
	\begin{filecontents*}{code/output/simulation_initial_condition_projection_table_summary.csv}
"projection_id","projection","n_controls","reps","n_success","mean_projection_r2","mean_max_abs_corr_omitted_y","mean_max_abs_corr_omitted_x","bias_rho","rmse_rho","coverage_rho","bias_phi","rmse_phi","coverage_phi","bias_beta1","rmse_beta1","coverage_beta1","bias_beta2","rmse_beta2","coverage_beta2","covariance_method","finite_sample_correction","logdet_method"
0,"full_L0",3,4000,3915,0.445717458959353,0.167032695342366,0.116125170399673,-0.0182247365934732,0.0280640627671108,0.785696040868455,0.0119666765947803,0.0237967679090091,0.851851851851852,-0.0126539477003855,0.0310405511515223,0.910344827586207,0.00687551189662045,0.025800402331554,0.940740740740741,"martingale_sample","pair_hc1","sparse_lu"
1,"full_L1",6,4000,3910,0.460367638803308,0.0175320475179156,0.0321054408036402,-0.00671744018891688,0.0203962446992465,0.944245524296675,0.00446409659084627,0.0180851601952195,0.939130434782609,-0.00966006489947937,0.0292280818496997,0.931457800511509,0.00469792851179191,0.0248186173535961,0.946803069053708,"martingale_sample","pair_hc1","sparse_lu"
2,"full_L2",9,4000,3913,0.462828560866232,0.00754348124326543,0.0140934191221787,-0.00754682272275059,0.0208010035370543,0.940710452338359,0.00471832491590859,0.0181169804124838,0.938665985177613,-0.00974026227596738,0.0292256275741253,0.929465882954255,0.00482665416862548,0.0248662365697493,0.947866087400971,"martingale_sample","pair_hc1","sparse_lu"
3,"full_L3",12,4000,3910,0.465273335862756,0.00289146540296141,0.00563566089977163,-0.00843216398184458,0.0213109502823222,0.936572890025575,0.00491231288213296,0.0183032200464023,0.936572890025575,-0.00994243834702346,0.0293196622349386,0.929156010230179,0.00475221417424196,0.0248182399418633,0.947570332480818,"martingale_sample","pair_hc1","sparse_lu"
4,"y_only_L2",5,4000,3918,0.459389718820668,0.00769380500972793,0.0527227343502086,-0.00610934985211329,0.0200153678521521,0.942827973455845,0.00429006195118437,0.0179648872872931,0.940786115364982,-0.00960917004961941,0.0291297752941919,0.929555895865237,0.0048487171413972,0.0249643288353175,0.946401225114854,"martingale_sample","pair_hc1","sparse_lu"
\end{filecontents*}
\csvreader[
	respect underscore=true,
	column names={
	1=\projectionid,
	3=\ncontrols,
	7=\corry,
	8=\corrx,
	9=\biasrho,
	10=\rmserho,
	11=\coveragerho,
	12=\biasphi,
	13=\rmsephi,
	14=\coveragephi,
	17=\coveragebetaone,
	20=\coveragebetatwo
	},
	tabular=lrrr|rrr|rrr|rr,
	table head=
	\toprule
	&&&& \multicolumn{3}{c}{$\rho$} &
	\multicolumn{3}{c}{$\phi$} &
	\multicolumn{2}{c}{Coverage} \\
	\cmidrule(lr){5-7}\cmidrule(lr){8-10}\cmidrule(lr){11-12}
	Projection & $q$ & $\operatorname{Corr}_y$ &
	$\operatorname{Corr}_x$ & Bias & RMSE & Coverage &
	Bias & RMSE & Coverage & $\beta_1$ & $\beta_2$ \\
	\midrule,
	table foot=\bottomrule
	]{code/output/simulation_initial_condition_projection_table_summary.csv}{}{
	\initialprojectionlabel{\projectionid} & \ncontrols &
	\num[round-precision=3]{\corry} &
	\num[round-precision=3]{\corrx} &
	\num[round-precision=3]{\biasrho} &
	\num[round-precision=3]{\rmserho} &
	\num[round-precision=3]{\coveragerho} &
	\num[round-precision=3]{\biasphi} &
	\num[round-precision=3]{\rmsephi} &
	\num[round-precision=3]{\coveragephi} &
	\num[round-precision=3]{\coveragebetaone} &
	\num[round-precision=3]{\coveragebetatwo}
	}
	\caption{Long-running initial conditions:
	projection and estimation performance averaged across 8 network designs.
	The full projection includes spatial transformations of both
	$Q_X$ and $y_0$ through order $L$; the $y$-only specification includes
	$Q_X$ without spatial transformations and spatial transformations of $y_0$
	through order two. The column $q$ is the number of controls.
	$\operatorname{Corr}_y$ and $\operatorname{Corr}_x$ are the mean maximum
	absolute correlations between the loading residuals and the next omitted
	spatial transformations of $y_0$ and $Q_X$, respectively.
	The regressors contain a factor loading component with coefficient $0.5$.
	The eight designs combine $T\in\{5,10\}$, $N\in\{500,1000\}$, and
	$\rho\in\{0.2,0.5\}$. Bias and RMSE are computed relative to the true
	parameter value in each design and pooled across strictly converged fits.
	Each design has 500 Monte Carlo replications.
	All designs evaluate the log determinant exactly by sparse LU
	factorization.
	Inference uses the martingale-sample covariance estimator with the pair-only
	HC1 correction. Depending on the projection, 3910--3918 of the 4000 attempted
	fits converged strictly and produced spatial-sandwich standard errors.}
	\label{tab:sim-initial-projection-summary}
\end{sidewaystable}

\begin{sidewaystable}[t]
	\centering
	\begin{filecontents*}{code/output/simulation_initial_condition_projection_table_coverage.csv}
"spec_id","T","N","rho","coverage_full_L0_rho","coverage_full_L0_phi","coverage_full_L1_rho","coverage_full_L1_phi","coverage_full_L2_rho","coverage_full_L2_phi","coverage_full_L3_rho","coverage_full_L3_phi","coverage_y_only_L2_rho","coverage_y_only_L2_phi","covariance_method","finite_sample_correction","logdet_method"
1,5,500,0.2,0.900212314225053,0.934182590233546,0.940042826552463,0.940042826552463,0.933618843683084,0.944325481798715,0.928270042194093,0.940928270042194,0.938428874734607,0.942675159235669,"martingale_sample","pair_hc1","sparse_lu"
5,5,500,0.5,0.642241379310345,0.75,0.918454935622318,0.935622317596567,0.919148936170213,0.938297872340425,0.904555314533623,0.932754880694143,0.915401301518438,0.939262472885033,"martingale_sample","pair_hc1","sparse_lu"
3,5,1000,0.2,0.871428571428571,0.910204081632653,0.944672131147541,0.92827868852459,0.940695296523517,0.924335378323108,0.938271604938272,0.917695473251029,0.94758064516129,0.929435483870968,"martingale_sample","pair_hc1","sparse_lu"
7,5,1000,0.5,0.312244897959184,0.495918367346939,0.934560327198364,0.910020449897751,0.926078028747433,0.909650924024641,0.922448979591837,0.908163265306122,0.922448979591837,0.910204081632653,"martingale_sample","pair_hc1","sparse_lu"
2,10,500,0.2,0.924,0.948,0.94,0.948,0.938,0.944,0.94,0.946,0.936,0.946,"martingale_sample","pair_hc1","sparse_lu"
6,10,500,0.5,0.884,0.92,0.964,0.95,0.958,0.952,0.953907815631262,0.949899799599198,0.97,0.954,"martingale_sample","pair_hc1","sparse_lu"
4,10,1000,0.2,0.92,0.942,0.95,0.946,0.948,0.944,0.95,0.942,0.944,0.946,"martingale_sample","pair_hc1","sparse_lu"
8,10,1000,0.5,0.82,0.906,0.96,0.954,0.96,0.952,0.952,0.954,0.966,0.958,"martingale_sample","pair_hc1","sparse_lu"
\end{filecontents*}
\csvreader[
	respect underscore=true,
	column names={
	2=\Tval,
	3=\Nval,
	4=\truerho,
	5=\Lzerorho,
	6=\Lzerophi,
	7=\Lonerho,
	8=\Lonephi,
	9=\Ltworho,
	10=\Ltwophi,
	11=\Lthreerho,
	12=\Lthreephi,
	13=\yonlyrho,
	14=\yonlyphi
	},
	tabular=rrr|rr|rr|rr|rr|rr,
	table head=
	\toprule
	&&& \multicolumn{2}{c}{Full $L=0$} &
	\multicolumn{2}{c}{Full $L=1$} &
	\multicolumn{2}{c}{Full $L=2$} &
	\multicolumn{2}{c}{Full $L=3$} &
	\multicolumn{2}{c}{$y$-only $L=2$} \\
	\cmidrule(lr){4-5}\cmidrule(lr){6-7}\cmidrule(lr){8-9}
	\cmidrule(lr){10-11}\cmidrule(lr){12-13}
	$T$ & $N$ & $\rho^{\mathrm{true}}$ &
	$\rho$ & $\phi$ & $\rho$ & $\phi$ & $\rho$ & $\phi$ &
	$\rho$ & $\phi$ & $\rho$ & $\phi$ \\
	\midrule,
	table foot=\bottomrule
	]{code/output/simulation_initial_condition_projection_table_coverage.csv}{}{
	\Tval & \Nval & \truerho &
	\num[round-precision=3]{\Lzerorho} &
	\num[round-precision=3]{\Lzerophi} &
	\num[round-precision=3]{\Lonerho} &
	\num[round-precision=3]{\Lonephi} &
	\num[round-precision=3]{\Ltworho} &
	\num[round-precision=3]{\Ltwophi} &
	\num[round-precision=3]{\Lthreerho} &
	\num[round-precision=3]{\Lthreephi} &
	\num[round-precision=3]{\yonlyrho} &
	\num[round-precision=3]{\yonlyphi}
	}
	\caption{Long-running initial conditions: coverage of nominal $95\%$
	confidence intervals for $\rho$ and $\phi$. Each design contains 500 Monte
	Carlo panels, and the same panel is estimated under all five projection
	specifications. The eight dynamically stable  designs use
	$T\in\{5,10\}$, $N\in\{500,1000\}$, and $\rho\in\{0.2,0.5\}$.
	All designs evaluate the log determinant exactly by sparse LU
	factorization.
	Coverage is computed from the 461--500 strictly converged fits in each cell;
	there were no spatial-sandwich failures among those fits.
	Inference uses the martingale-sample covariance estimator with the pair-only HC1 correction.
	The full and $y$-only projections are defined in the text and summarized in
	Table~\ref{tab:sim-initial-projection-summary}.}
	\label{tab:sim-initial-projection-coverage}
\end{sidewaystable}

%% file: application_female_lfp_baseline.tex
\begin{table}[htbp]
\centering
\caption{Female labor-force participation: baseline QMLE estimates}
\label{tab:application-female-lfp-baseline}
\scriptsize
\setlength{\tabcolsep}{3pt}
\begin{tabular}{lrrrrrr}
\toprule
& \multicolumn{2}{c}{Wages and $WX$} & \multicolumn{2}{c}{No wages; with $WX$} & \multicolumn{2}{c}{No wages or $WX$} \\
\cmidrule(lr){2-3}\cmidrule(lr){4-5}\cmidrule(lr){6-7}
Variable & Estimate (SE) & Contribution & Estimate (SE) & Contribution & Estimate (SE) & Contribution \\
\midrule
$y_{i,t-1}$ & 0.372 (0.019) & -- & 0.467 (0.035) & -- & 0.364 (0.016) & -- \\
$Wy_t$ & 0.521 (0.017) & -- & 0.479 (0.016) & -- & 0.506 (0.011) & -- \\
Urban population & 0.011 (0.005) & -10.00 & 0.017 (0.005) & -12.41 & -0.001 (0.003) & -0.52 \\
Farm population & -0.079 (0.015) & 67.55 & -0.081 (0.006) & 36.47 & -0.045 (0.007) & 38.95 \\
Education & 0.555 (0.111) & 76.14 & 0.638 (0.079) & 49.47 & 0.680 (0.089) & 69.98 \\
Density/1000 & 0.084 (0.027) & -0.20 & 0.134 (0.043) & -0.24 & -0.053 (0.049) & -0.05 \\
Wages/1000 & -0.062 (0.017) & -17.27 & -- & -- & -- & -- \\
$W$ Urban population & -0.034 (0.008) & -- & -0.031 (0.005) & -- & -- & -- \\
$W$ Farm population & 0.016 (0.019) & -- & 0.063 (0.007) & -- & -- & -- \\
$W$ Education & 0.051 (0.147) & -- & -0.438 (0.089) & -- & -- & -- \\
$W$ Density/1000 & -0.263 (0.038) & -- & -0.243 (0.091) & -- & -- & -- \\
$W$ Wages/1000 & 0.001 (0.037) & -- & -- & -- & -- & -- \\
\midrule
Sum of contributions & -- & 116.22 & -- & 73.29 & -- & 108.36 \\
Stability sum & 0.893 & -- & 0.946 & -- & 0.870 & -- \\
Observations & 9414 & -- & 18396 & -- & 18396 & -- \\
\bottomrule
\end{tabular}
\begin{minipage}{0.98\textwidth}
\scriptsize Notes: All QMLE specifications use two common factors. The baseline model sets the coefficient on $Wy_{t-1}$ and the spatial-error coefficient to zero. Standard errors use the martingale-sample covariance estimator with the pair-only HC1 correction in \eqref{eq:Omega_martingale_pair_hc1}. Contributions are percentages of the 1940--2000 increase in female labor-force participation.
\end{minipage}
\end{table}

%% file: application_female_lfp.tex
\begin{table}[htbp]
\centering
\caption{Lagged-spatial-outcome extension: published BC-QMLE and proposed QMLE}
\label{tab:application-female-lfp-extension}
\scriptsize
\setlength{\tabcolsep}{3pt}
\begin{tabular}{lrrrr}
\toprule
& \multicolumn{2}{c}{Estimate (standard error)} & \multicolumn{2}{c}{Contribution (\%)} \\
\cmidrule(lr){2-3}\cmidrule(lr){4-5}
Variable & Published BC-QMLE & Proposed QMLE & Published BC-QMLE & Proposed QMLE \\
\midrule
\multicolumn{5}{l}{\emph{Wages and spatially lagged controls included}} \\
$y_{i,t-1}$ & 0.153 (0.012) & 0.436 (0.051) & -- & -- \\
$Wy_t$ & -0.074 (0.075) & 0.530 (0.015) & -- & -- \\
$Wy_{t-1}$ & 0.504 (0.026) & -0.073 (0.043) & -- & -- \\
Urban population & 0.017 (0.004) & 0.009 (0.005) & -2.02 & -9.81 \\
Farm population & -0.081 (0.007) & -0.075 (0.013) & 54.19 & 55.49 \\
Education & 0.017 (0.069) & 0.592 (0.113) & -0.51 & 88.99 \\
Density/1000 & 0.031 (0.064) & 0.076 (0.026) & -0.18 & -0.18 \\
Wages/1000 & -0.053 (0.015) & -0.062 (0.017) & -8.09 & -10.28 \\
$W$ Urban population & -0.036 (0.009) & -0.032 (0.007) & -- & -- \\
$W$ Farm population & -0.120 (0.015) & 0.022 (0.020) & -- & -- \\
$W$ Education & -0.033 (0.127) & 0.119 (0.158) & -- & -- \\
$W$ Density/1000 & -0.664 (0.127) & -0.239 (0.039) & -- & -- \\
$W$ Wages/1000 & -0.059 (0.035) & 0.025 (0.043) & -- & -- \\
$W$ error & 0.595 (0.073) & -- & -- & -- \\
Sum of contributions & -- & -- & 43.40 & 124.21 \\
Stability sum & 0.582 & 0.893 & -- & -- \\
Observations & 9414 & 9414 & -- & -- \\
\addlinespace
\multicolumn{5}{l}{\emph{Wages excluded; spatially lagged controls included}} \\
$y_{i,t-1}$ & 0.148 (0.010) & 0.595 (0.206) & -- & -- \\
$Wy_t$ & -0.340 (0.089) & 0.509 (0.021) & -- & -- \\
$Wy_{t-1}$ & 0.381 (0.017) & -0.150 (0.164) & -- & -- \\
Urban population & 0.011 (0.003) & 0.017 (0.005) & 1.77 & -2.57 \\
Farm population & -0.070 (0.005) & -0.054 (0.099) & 21.27 & 5.35 \\
Education & 0.060 (0.058) & 0.656 (0.972) & 1.88 & 68.67 \\
Density/1000 & 0.093 (0.077) & 0.022 (0.452) & -0.10 & -0.09 \\
$W$ Urban population & 0.021 (0.007) & -0.019 (0.058) & -- & -- \\
$W$ Farm population & -0.083 (0.012) & 0.052 (0.016) & -- & -- \\
$W$ Education & 0.053 (0.100) & -0.416 (0.331) & -- & -- \\
$W$ Density/1000 & -0.820 (0.162) & -0.060 (0.587) & -- & -- \\
$W$ error & 0.830 (0.090) & -- & -- & -- \\
Sum of contributions & -- & -- & 24.81 & 71.36 \\
Stability sum & 0.189 & 0.953 & -- & -- \\
Observations & 18396 & 18396 & -- & -- \\
\addlinespace
\multicolumn{5}{l}{\emph{Wages and spatially lagged controls excluded}} \\
$y_{i,t-1}$ & 0.127 (0.009) & 0.389 (0.043) & -- & -- \\
$Wy_t$ & 0.087 (0.052) & 0.509 (0.011) & -- & -- \\
$Wy_{t-1}$ & 0.454 (0.017) & -0.025 (0.034) & -- & -- \\
Urban population & 0.006 (0.003) & -0.002 (0.003) & 0.89 & -0.69 \\
Farm population & -0.090 (0.005) & -0.042 (0.008) & 30.76 & 37.43 \\
Education & 0.004 (0.055) & 0.702 (0.098) & 0.17 & 73.99 \\
Density/1000 & 0.089 (0.075) & -0.051 (0.048) & 0.03 & -0.05 \\
$W$ error & 0.432 (0.051) & -- & -- & -- \\
Sum of contributions & -- & -- & 31.86 & 110.67 \\
Stability sum & 0.668 & 0.873 & -- & -- \\
Observations & 18396 & 18396 & -- & -- \\
\addlinespace
\bottomrule
\end{tabular}
\begin{minipage}{0.98\textwidth}
\scriptsize Notes: The published BC-QMLE columns report the Tziolas--Elhorst bias-corrected QML results transcribed from their replication output. The proposed QMLE specifications use two common factors, and their standard errors use the martingale-sample covariance estimator with the pair-only HC1 correction in \eqref{eq:Omega_martingale_pair_hc1}. Contributions use the 1940--2000 changes reported in their Table~1. Their specifications include a spatial autoregressive error coefficient; the proposed specifications do not.
\end{minipage}
\end{table}

%% file: application_projection_robustness.tex
\begin{table}[htbp]
\centering
\caption{Robustness to the SERL enrichment order with two common factors}
\label{tab:application-projection-robustness}
\small
\begin{tabular}{lrrr}
\toprule
Parameter & $L=0$ & $L=1$ & $L=2$ \\
\midrule
$Wy_t$ & 0.471 (0.011) & 0.509 (0.011) & 0.507 (0.011) \\
$y_{i,t-1}$ & 0.523 (0.021) & 0.389 (0.043) & 0.378 (0.038) \\
$Wy_{t-1}$ & -0.142 (0.016) & -0.025 (0.034) & -0.017 (0.030) \\
Urban population & 0.002 (0.003) & -0.002 (0.003) & -0.001 (0.003) \\
Farm population & -0.021 (0.006) & -0.042 (0.008) & -0.044 (0.008) \\
Education & 0.692 (0.099) & 0.702 (0.098) & 0.677 (0.097) \\
Density/1000 & -0.001 (0.024) & -0.051 (0.048) & -0.049 (0.047) \\
\midrule
Stability sum & 0.852 & 0.873 & 0.869 \\
\bottomrule
\end{tabular}
\begin{minipage}{0.90\textwidth}
\small Notes: Standard errors use the martingale-sample covariance estimator with the pair-only HC1 correction in \eqref{eq:Omega_martingale_pair_hc1}. For $L=0$, the smallest estimated eigenvalue of $\Sigma_\eta$ is effectively zero, so those standard errors should be interpreted cautiously because the interior covariance condition is nearly binding.
\end{minipage}
\end{table}

%% file: application_female_lfp_second_order.tex
\begin{table}[htbp]
\centering
\caption{Preferred second-order QMLE estimates}
\label{tab:application-second-order}
\scriptsize
\setlength{\tabcolsep}{2.5pt}
\resizebox{\textwidth}{!}{%
\begin{tabular}{lrrrrrr}
\toprule
& \multicolumn{2}{c}{Wages and $WX$} & \multicolumn{2}{c}{No wages; with $WX$} & \multicolumn{2}{c}{No wages or $WX$} \\
\cmidrule(lr){2-3}\cmidrule(lr){4-5}\cmidrule(lr){6-7}
Variable & Estimate (SE) & Contribution (SE) & Estimate (SE) & Contribution (SE) & Estimate (SE) & Contribution (SE) \\
\midrule
$\rho_1$ & 0.203 (0.032) & -- & 0.199 (0.020) & -- & 0.219 (0.020) & -- \\
$\rho_2$ & 0.468 (0.038) & -- & 0.487 (0.025) & -- & 0.453 (0.027) & -- \\
$y_{i,t-1}$ & 0.614 (0.025) & -- & 0.552 (0.029) & -- & 0.568 (0.032) & -- \\
$Wy_{t-1}$ & 0.037 (0.033) & -- & 0.006 (0.026) & -- & 0.033 (0.025) & -- \\
$W^2y_{t-1}$ & -0.356 (0.043) & -- & -0.266 (0.038) & -- & -0.303 (0.035) & -- \\
Urban population & 0.013 (0.004) & -17.00 (9.03) & 0.019 (0.005) & -21.58 (12.86) & 0.014 (0.005) & -5.80 (7.14) \\
Farm population & -0.024 (0.009) & 37.61 (27.89) & -0.049 (0.007) & 55.67 (29.64) & -0.057 (0.008) & 19.97 (18.74) \\
Education & 0.527 (0.084) & 82.39 (38.63) & 0.646 (0.072) & 73.70 (37.17) & 0.536 (0.065) & 71.79 (26.21) \\
Density/1000 & 0.060 (0.015) & -0.21 (0.09) & 0.098 (0.040) & -0.46 (0.27) & 0.071 (0.039) & -0.23 (0.16) \\
Wages/1000 & -0.060 (0.017) & 3.92 (28.81) & -- & -- & -- & -- \\
$W$ Urban population & -0.011 (0.007) & -- & -0.003 (0.006) & -- & -0.018 (0.004) & -- \\
$W$ Farm population & 0.090 (0.015) & -- & 0.081 (0.011) & -- & 0.052 (0.007) & -- \\
$W$ Education & 0.608 (0.177) & -- & 0.081 (0.130) & -- & -0.373 (0.074) & -- \\
$W$ Density/1000 & -0.089 (0.044) & -- & -0.060 (0.067) & -- & -0.130 (0.070) & -- \\
$W$ Wages/1000 & 0.030 (0.040) & -- & -- & -- & -- & -- \\
$W^2$ Urban population & -0.015 (0.008) & -- & -0.026 (0.007) & -- & -- & -- \\
$W^2$ Farm population & -0.078 (0.018) & -- & -0.042 (0.013) & -- & -- & -- \\
$W^2$ Education & -0.935 (0.210) & -- & -0.607 (0.152) & -- & -- & -- \\
$W^2$ Density/1000 & -0.031 (0.042) & -- & -0.124 (0.074) & -- & -- & -- \\
$W^2$ Wages/1000 & 0.034 (0.048) & -- & -- & -- & -- & -- \\
\midrule
Sum of contributions & -- & 106.70 & -- & 107.33 & -- & 85.73 \\
Dynamic coefficient for uniform change & 0.901 & -- & 0.930 & -- & 0.908 & -- \\
Observations & 9414 & -- & 18396 & -- & 18396 & -- \\
\bottomrule
\end{tabular}
}
\begin{minipage}{0.98\textwidth}
\scriptsize Notes: The filter is $(I_N-\kappa_1W)(I_N-\kappa_2W)=I_N-\rho_1W-\rho_2W^2$. All specifications use two common factors and $L=1$. Standard errors use the martingale-sample covariance estimator with the pair-only HC1 correction in \eqref{eq:Omega_martingale_pair_hc1}; contribution standard errors use the delta method.
\end{minipage}
\end{table}

%% file: simulation_bias_sd_sparse.tex
\begin{sidewaystable}[t]
	\centering
	\scriptsize
	\csvreader[
	respect underscore=true,
	column names={
	T=\Tval,
	N=\Nval,
	knn_type=\knntype,
	rho=\truerho,
	phi=\truephi,
	beta1=\truebetaone,
	beta2=\truebetatwo,
	bias_x100_rho=\BIASrho,
	sd_rho=\SDrho,
	bias_x100_phi=\BIASphi,
	sd_phi=\SDphi,
	bias_x100_beta1=\BIASbetaone,
	sd_beta1=\SDbetaone,
	bias_x100_beta2=\BIASbetatwo,
	sd_beta2=\SDbetatwo,
	},
	filter strcmp={\knntype}{sparse},
	tabular=rrrrrr|rrrrrrrr,
	table head=
	\toprule
	\multicolumn{1}{c}{$T$} &
	\multicolumn{1}{c}{$N$} &
	\multicolumn{1}{c}{$\rho^{\text{true}}$} &
	\multicolumn{1}{c}{$\phi^{\text{true}}$} &
	\multicolumn{1}{c}{$\beta_1^{\text{true}}$} &
	\multicolumn{1}{c}{$\beta_2^{\text{true}}$} &
	\multicolumn{2}{c}{$\hat\rho$} &
	\multicolumn{2}{c}{$\hat\phi$} &
	\multicolumn{2}{c}{$\hat\beta_1$} &
	\multicolumn{2}{c}{$\hat\beta_2$} \\
	\cmidrule(lr){7-8}\cmidrule(lr){9-10}
	\cmidrule(lr){11-12}\cmidrule(lr){13-14}
	&&&&&&
	\multicolumn{1}{c}{$100\times\mathrm{Bias}$} &
	\multicolumn{1}{c}{SD} &
	\multicolumn{1}{c}{$100\times\mathrm{Bias}$} &
	\multicolumn{1}{c}{SD} &
	\multicolumn{1}{c}{$100\times\mathrm{Bias}$} &
	\multicolumn{1}{c}{SD} &
	\multicolumn{1}{c}{$100\times\mathrm{Bias}$} &
	\multicolumn{1}{c}{SD} \\
	\midrule,
	table foot = \bottomrule
	]{code/output/simulation_bias_sd.csv}{}{
	\Tval & \Nval &
	\num[group-digits=false,round-precision=1]{\truerho} &
	\num[group-digits=false,round-precision=1]{\truephi} &
	\num[group-digits=false,round-precision=1]{\truebetaone} &
	\num[group-digits=false,round-precision=1]{\truebetatwo} &
	\num[group-digits=false,round-precision=4]{\BIASrho} &
	\num[group-digits=false,round-precision=5]{\SDrho} &
	\num[group-digits=false,round-precision=4]{\BIASphi} &
	\num[group-digits=false,round-precision=5]{\SDphi} &
	\num[group-digits=false,round-precision=4]{\BIASbetaone} &
	\num[group-digits=false,round-precision=5]{\SDbetaone} &
	\num[group-digits=false,round-precision=4]{\BIASbetatwo} &
	\num[group-digits=false,round-precision=5]{\SDbetatwo}
	}
	\caption{\textsc{Sparse-matrix Design}.
	Bias and standard deviations.
	The implementation evaluates log determinant using a 30-term Hutchinson trace approximation with 25 Rademacher vectors.
	The columns and acceptance rule are defined as in Table~\ref{tab:sim-bias-sd-dense}.
	Each of the 18 designs requests 1000 replications;
	between 995 and 1000 estimates per design enter the calculations,
	for 17,982 accepted estimates in total.}
	\label{tab:sim-bias-sd-sparse}
\end{sidewaystable}

%% file: simulation_coverage_sparse.tex
\begin{sidewaystable}[t]
	\centering
	\scriptsize
	\csvreader[
	respect underscore=true,
	column names={
	2=\Tval,
	3=\Nval,
	4=\knntype,
	5=\truerho,
	6=\truephi,
	7=\truebetaone,
	8=\truebetatwo,
	9=\uncorrectedrho,
	10=\uncorrectedphi,
	11=\hcspatialrho,
	12=\hcspatialphi,
	13=\spatialbetaone,
	14=\spatialbetatwo
	},
	filter strcmp={\knntype}{"sparse"},
	tabular=rr|rrrr|rr|rr|rr,
	table head=
	\toprule
	&&&&&&
	\multicolumn{2}{c}{No correction} &
	\multicolumn{2}{c}{Pair-only HC1} &
	\multicolumn{2}{c}{Regression slopes} \\
	\cmidrule(lr){7-8}\cmidrule(lr){9-10}\cmidrule(lr){11-12}
	$T$ & $N$ &
	$\rho^{\mathrm{true}}$ & $\phi^{\mathrm{true}}$ &
	$\beta_1^{\mathrm{true}}$ & $\beta_2^{\mathrm{true}}$ &
	$\rho$ & $\phi$ & $\rho$ & $\phi$ & $\beta_1$ & $\beta_2$ \\
	\midrule,
	table foot=\bottomrule
	]{code/output/simulation_coverage_spatial_comparison.csv}{}{
	\Tval & \Nval &
	\num[group-digits=false,round-precision=1]{\truerho} &
	\num[group-digits=false,round-precision=1]{\truephi} &
	\num[group-digits=false,round-precision=1]{\truebetaone} &
	\num[group-digits=false,round-precision=1]{\truebetatwo} &
	\num[group-digits=false,round-precision=3]{\uncorrectedrho} &
	\num[group-digits=false,round-precision=3]{\uncorrectedphi} &
	\num[group-digits=false,round-precision=3]{\hcspatialrho} &
	\num[group-digits=false,round-precision=3]{\hcspatialphi} &
	\num[group-digits=false,round-precision=3]{\spatialbetaone} &
	\num[group-digits=false,round-precision=3]{\spatialbetatwo}
	}
	\caption{\textsc{Sparse-matrix Design}.
	Coverage of nominal $95\%$ confidence intervals.
	The implementation evaluates log determinant using a 30-term Hutchinson trace approximation with 25 Rademacher vectors.
	The first two coverage columns use the martingale-sample covariance estimator without a finite-sample correction,
	while the next two use the preferred pair-only HC1 correction in \eqref{eq:Omega_martingale_pair_hc1}.
	Each of the 18 designs requests 1000 Monte Carlo replications;
	between 979 and 1000 replications per design produce strictly converged estimates and complete sandwich inference,
	for 17,924 successful replications in total.}
	\label{tab:sim-coverage-sparse}
\end{sidewaystable}

%% file: simulation_sparse_logdet_diagnostic.tex
\begin{table}[t]
	\centering
	\small
	\begin{filecontents*}{code/output/simulation_sparse_logdet_diagnostic_table.csv}
"N","rho","phi","coverage_rho_trace_30_25","coverage_phi_trace_30_25","coverage_rho_trace_60_100","coverage_phi_trace_60_100","coverage_rho_sparse_lu","coverage_phi_sparse_lu"
500,0.2,0.4,0.885,0.935,0.945,0.945,0.97,0.945
500,0.5,0.25,0.91,0.96,0.945,0.96,0.95,0.955
500,0.8,0.1,0.88,0.945,0.905,0.95,0.93,0.95
1000,0.2,0.4,0.865,0.945,0.915,0.945,0.94,0.945
1000,0.5,0.25,0.835,0.955,0.91,0.96,0.925,0.955
1000,0.8,0.1,0.855,0.945,0.925,0.955,0.94,0.96
\end{filecontents*}
\csvreader[
	column names={
	1=\Nval,
	2=\truerho,
	3=\truephi,
	4=\tracebaserho,
	5=\tracebasephi,
	6=\tracerichrho,
	7=\tracerichphi,
	8=\lurho,
	9=\luphi
	},
	tabular=rrr|rr|rr|rr,
	table head=
	\toprule
	&&&
	\multicolumn{2}{c}{Trace 30/25} &
	\multicolumn{2}{c}{Trace 60/100} &
	\multicolumn{2}{c}{Sparse LU} \\
	\cmidrule(lr){4-5}\cmidrule(lr){6-7}\cmidrule(lr){8-9}
	$N$ & $\rho$ & $\phi$ &
	$\rho$ & $\phi$ & $\rho$ & $\phi$ & $\rho$ & $\phi$ \\
	\midrule,
	table foot=\bottomrule
	]{code/output/simulation_sparse_logdet_diagnostic_table.csv}{}{
	\Nval &
	\num[group-digits=false,round-precision=1]{\truerho} &
	\num[group-digits=false,round-precision=2]{\truephi} &
	\num[group-digits=false,round-precision=3]{\tracebaserho} &
	\num[group-digits=false,round-precision=3]{\tracebasephi} &
	\num[group-digits=false,round-precision=3]{\tracerichrho} &
	\num[group-digits=false,round-precision=3]{\tracerichphi} &
	\num[group-digits=false,round-precision=3]{\lurho} &
	\num[group-digits=false,round-precision=3]{\luphi}
	}
	\caption{Coverage under alternative sparse log-determinant calculations.
	All designs set $T=20$ and use the same 200 simulated panels under each method.
	Trace 30/25 denotes a 30-term Hutchinson approximation with 25 Rademacher vectors;
	Trace 60/100 uses 60 terms and 100 vectors.
	Intervals use the martingale-sample covariance estimator with the pair-only HC1 correction.}
	\label{tab:sim-sparse-logdet}
\end{table}

%% file: simulation_symmetric_design_pilot.tex
\begin{table}[t]
	\centering
	\small
	\begin{filecontents*}{code/output/simulation_symmetric_design_pilot_summary.csv}
"spec_id","T","N","knn_type","rho","phi","stability_sum","beta1","beta2","reps","n_success","bias_rho","bias_phi","bias_beta1","bias_beta2","sd_rho","sd_phi","sd_beta1","sd_beta2","mean_se_rho","mean_se_phi","mean_se_beta1","mean_se_beta2","coverage_rho","coverage_phi","coverage_beta1","coverage_beta2","covariance_method","finite_sample_correction"
1,5,500,"dense",0.2,0.7,0.9,0.8,-0.3,200,194,-0.0102984410222472,-0.00322367868822604,-0.00902225534938572,0.00465380198878812,0.0287120405809548,0.0569147315814397,0.0633706939234438,0.0592538309514958,0.0596194982665554,0.297053294373765,0.383554995951671,0.188134982826728,0.93298969072165,0.845360824742268,0.948453608247423,0.963917525773196,"martingale_sample","pair_hc1"
3,5,500,"dense",0.3,0.6,0.9,0.8,-0.3,200,196,-0.0108973942514006,-0.00271076661959213,0.00100387751554876,0.00273842835991545,0.0299847030908159,0.0523461965740738,0.0659677367371511,0.0659178195296354,0.0360970149975494,0.100654619491773,0.082509207471633,0.0728284235053955,0.928571428571429,0.887755102040816,0.943877551020408,0.938775510204082,"martingale_sample","pair_hc1"
5,5,500,"dense",0.4,0.5,0.9,0.8,-0.3,200,198,-0.00875710449925349,-0.00395226555720346,-0.000137508667311548,-0.00369221614736109,0.0270592289261413,0.0482413706503211,0.0625248528139342,0.0622480623664141,0.0296118537925394,0.0508841402282724,0.0671943896641804,0.0665414619596504,0.94949494949495,0.909090909090909,0.94949494949495,0.944444444444444,"martingale_sample","pair_hc1"
7,5,500,"dense",0.5,0.4,0.9,0.8,-0.3,200,194,-0.00975832498594364,0.00173020023181852,-0.00367537404875251,-0.000136546911379709,0.025502646529208,0.0399719608340952,0.0678777422423905,0.0651520464463845,0.0270955761271003,0.0425076908223016,0.0660204916444638,0.0634015535971841,0.922680412371134,0.871134020618557,0.93298969072165,0.938144329896907,"martingale_sample","pair_hc1"
9,5,500,"dense",0.6,0.3,0.9,0.8,-0.3,200,191,-0.00724020899633426,-0.00116129925437125,-0.00315778706348879,-0.000529351607372774,0.0231120678131055,0.0369161490939638,0.0593446189957098,0.0620441376575495,0.026541831265966,0.0574264583188388,0.0734880508039438,0.0720194767815553,0.942408376963351,0.926701570680628,0.973821989528796,0.968586387434555,"martingale_sample","pair_hc1"
11,5,500,"dense",0.7,0.2,0.9,0.8,-0.3,200,198,-0.00707894700038647,-0.000799842755809826,0.00593139535493814,-0.00957931783807203,0.0185305298283434,0.0276806741521537,0.0614445842623871,0.0565820972100995,0.0194734277296005,0.0307725617030655,0.0639837202969095,0.0644345342180176,0.94949494949495,0.939393939393939,0.954545454545455,0.95959595959596,"martingale_sample","pair_hc1"
2,20,500,"dense",0.2,0.7,0.9,0.8,-0.3,200,200,-0.00458405470426157,6.64386599510403e-05,-0.000512011050740946,-0.00236905227108741,0.0124640930510508,0.007890446056397,0.0257644118881488,0.0285193444003693,0.0141468148400491,0.0079425378487364,0.0279160160770207,0.0280488245962045,0.96,0.95,0.955,0.94,"martingale_sample","pair_hc1"
4,20,500,"dense",0.3,0.6,0.9,0.8,-0.3,200,200,-0.00627762882506566,0.00148409138446948,0.000918788238943646,-0.000833350297569146,0.0130925796281435,0.00889913081691344,0.0275126225708073,0.029854972179597,0.0139799766161577,0.0086285961065482,0.0279524249179729,0.0280327987197886,0.955,0.935,0.945,0.92,"martingale_sample","pair_hc1"
6,20,500,"dense",0.4,0.5,0.9,0.8,-0.3,200,200,-0.00602720283668012,0.000941109037924659,0.000375131120611494,-0.00249624648342184,0.0114511814388521,0.00864391135272872,0.0271075216634722,0.0296475030156994,0.0132655886178262,0.00895251233206118,0.0279873778950226,0.0280513329384277,0.965,0.965,0.95,0.955,"martingale_sample","pair_hc1"
8,20,500,"dense",0.5,0.4,0.9,0.8,-0.3,200,200,-0.00554634776718987,0.00122616871955489,0.00194075801366755,-0.00257390615278036,0.0114627813420148,0.00854851085724239,0.0281764025308254,0.0280273023498096,0.0125306354794001,0.00908586879447913,0.0279256521989985,0.0279394177898086,0.94,0.96,0.955,0.945,"martingale_sample","pair_hc1"
10,20,500,"dense",0.6,0.3,0.9,0.8,-0.3,200,200,-0.00530558615561383,0.000949509085967748,0.000136527677795595,-0.000348165413227353,0.0106219923146129,0.00892510005122907,0.0256531871357383,0.0314228048319751,0.0113033494290962,0.00897816115896214,0.0279963978455628,0.0280615659272759,0.95,0.96,0.98,0.915,"martingale_sample","pair_hc1"
12,20,500,"dense",0.7,0.2,0.9,0.8,-0.3,200,200,-0.00509173904016968,0.00161665603461492,0.000917101235997864,-0.000126748614718508,0.00930157299571559,0.00840996594033286,0.0282205606665268,0.0271387508808806,0.00966846661735593,0.00840016395543023,0.02807396256688,0.0280006093368572,0.95,0.95,0.965,0.96,"martingale_sample","pair_hc1"
\end{filecontents*}
\csvreader[
	column names={
	2=\Tval,
	5=\truerho,
	6=\truephi,
	24=\coverrho,
	25=\coverphi,
	26=\coverbetaone,
	27=\coverbetatwo
	},
	tabular=rrr|rrrr,
	table head=
	\toprule
	&&&
	\multicolumn{4}{c}{Coverage} \\
	\cmidrule(lr){4-7}
	$T$ & $\rho$ & $\phi$ &
	$\rho$ & $\phi$ & $\beta_1$ & $\beta_2$ \\
	\midrule,
	table foot=\bottomrule
	]{code/output/simulation_symmetric_design_pilot_summary.csv}{}{
	\Tval &
	\num[group-digits=false,round-precision=1]{\truerho} &
	\num[group-digits=false,round-precision=1]{\truephi} &
	\num[group-digits=false,round-precision=3]{\coverrho} &
	\num[group-digits=false,round-precision=3]{\coverphi} &
	\num[group-digits=false,round-precision=3]{\coverbetaone} &
	\num[group-digits=false,round-precision=3]{\coverbetatwo}
	}
	\caption{Coverage in the symmetric stability-sum pilot.
	The design fixes $N=500$ and holds $\rho+\phi=0.9$ while varying its spatial and dynamic components.
	Each design requests 200 replications;
	between 191 and 200 replications per design complete estimation and inference,
	for 2371 successful replications in total.
	Intervals use the martingale-sample covariance estimator with the pair-only HC1 correction.}
	\label{tab:sim-symmetric-pilot}
\end{table}

%% file: simulation_explosive_stress.tex
\begin{sidewaystable}[t]
	\centering
	\small
	\sisetup{round-precision=4,group-digits=none}
	\setlength{\tabcolsep}{5pt}
	\begin{tabular}{rrl|r S[table-format=-1.4] S[table-format=1.4] S[table-format=-1.4] S[table-format=1.4] r S[table-format=1.4] S[table-format=1.4]}
		\toprule
		&&&
		\multicolumn{5}{c}{Bias and sampling dispersion} &
		\multicolumn{3}{c}{Spatial-sandwich coverage} \\
		\cmidrule(lr){4-8}\cmidrule(lr){9-11}
		$T$ & $N$ & Design & $R$ & {$100\,\mathrm{Bias}(\hat\rho)$} & {$\mathrm{SD}(\hat\rho)$} & {$100\,\mathrm{Bias}(\hat\phi)$} & {$\mathrm{SD}(\hat\phi)$} & $R$ & {$\rho$} & {$\phi$} \\
		\midrule
		5  & 500  & dense-matrix & 901  & -0.3156 & 0.00938 &  0.6968 & 0.02147 & 500 & 0.902 & 0.892 \\
		5  & 1000 & dense-matrix & 898  & -0.1918 & 0.00753 &  0.4243 & 0.01767 & 500 & 0.862 & 0.850 \\
		10 & 500  & dense-matrix & 984  & -0.0014 & 0.00055 &  0.0033 & 0.00130 & 500 & 0.964 & 0.958 \\
		10 & 1000 & dense-matrix & 991  &  0.0005 & 0.00033 & -0.0013 & 0.00080 & 500 & 0.942 & 0.938 \\
		20 & 500  & dense-matrix & 995  & -0.0013 & 0.00027 &  0.0031 & 0.00064 & 478 & 0.921 & 0.916 \\
		20 & 1000 & dense-matrix & 1000 & -0.0005 & 0.00005 &  0.0013 & 0.00012 & 482 & 0.925 & 0.927 \\
		\addlinespace
		5  & 500  & sparse-matrix & 884  & -0.3226 & 0.01004 &  0.7021 & 0.02291 & 499 & 0.878 & 0.872 \\
		5  & 1000 & sparse-matrix & 905  & -0.1837 & 0.00757 &  0.4160 & 0.01793 & 500 & 0.838 & 0.842 \\
		10 & 500  & sparse-matrix & 986  &  0.0003 & 0.00060 & -0.0013 & 0.00141 & 500 & 0.954 & 0.952 \\
		10 & 1000 & sparse-matrix & 989  &  0.0010 & 0.00033 & -0.0023 & 0.00079 & 499 & 0.958 & 0.956 \\
		20 & 500  & sparse-matrix & 992  &  0.0021 & 0.00046 & -0.0051 & 0.00112 & 482 & 0.932 & 0.934 \\
		20 & 1000 & sparse-matrix & 1000 & -0.0001 & 0.00007 &  0.0003 & 0.00016 & 479 & 0.923 & 0.925 \\
		\bottomrule
	\end{tabular}
	\caption{Estimation and Inference under an Unstable Design.
	The design fixes $(\rho,\phi)=(0.8,0.5)$,
	so a spatially uniform component has dynamic coefficient $\phi/(1-\rho)=2.5$.
	The bias experiment requests 1000 replications and the coverage experiment requests 500; $R$ reports successful fits.
	In the dense-matrix design,
	the log determinant is evaluated using the eigenvalue method;
	in the sparse-matrix design,
	it is evaluated using a 30-term Hutchinson trace approximation with 25 Rademacher vectors.
	The coverage columns report nominal $95\%$ intervals
	based on standard errors estimated from the outer product of the rearranged scores
	$g_i(\widehat\alpha)$.
	}
	\label{tab:simulation-explosive-stress}
\end{sidewaystable}

%% file: application_second_order_rank.tex
\begin{table}[htbp]
\centering
\caption{Second-order QML factor-rank sensitivity}
\label{tab:application-second-order-rank}
\small
\begin{tabular}{llrrrr}
\toprule
Specification & Rank & $\rho_1$ & $\rho_2$ & Education & Farm population \\
\midrule
Wages and $WX$ & $r=2$ & 0.203 & 0.468 & 0.527 & -0.024 \\
Wages and $WX$ & $r=3$ & 0.214 & 0.442 & 0.594 & -0.021 \\
\midrule
No wages; with $WX$ & $r=2$ & 0.199 & 0.487 & 0.646 & -0.049 \\
No wages; with $WX$ & $r=3$ & 0.211 & 0.456 & 0.793 & -0.032 \\
\midrule
No wages or $WX$ & $r=2$ & 0.219 & 0.453 & 0.536 & -0.057 \\
No wages or $WX$ & $r=3$ & 0.225 & 0.433 & 0.687 & -0.041 \\
\bottomrule
\end{tabular}
\end{table}